\documentclass[onecolumn]{IEEEtran}

\usepackage{epsfig,graphicx,color,url,cite,ctable}
\usepackage{ifpdf}
\usepackage[cmex10]{amsmath}
\usepackage{caption}
\usepackage{epstopdf}
\usepackage{epsfig}
\usepackage{verbatim}
\usepackage{amsthm}
\usepackage{esdiff}
\usepackage{lipsum}
\usepackage{algorithm}
\usepackage{algpseudocode}
\usepackage{etoolbox, subfigure,multirow, booktabs,changepage,blkarray,enumerate}

\algnewcommand{\Input}[1]{%
	\State \textbf{Input:}
	\Statex \hspace*{\algorithmicindent}\parbox[t]{.8\linewidth}{\raggedright #1}
}
\algnewcommand{\Output}[1]{%
	\State \textbf{Output:}
	\Statex \hspace*{\algorithmicindent}\parbox[t]{.8\linewidth}{\raggedright #1}
}

\graphicspath{{figs/}}
\newtheorem{defn}{Definition}
\newtheorem{lem}{Lemma}

\newtheorem{prop}{Proposition}
\newtheorem*{prop*}{Proposition}

\newtheorem{remark}{Remark}

\newcommand{\F}{\mathcal{F}}

\newcommand{\A}{\mathcal{A}}

\newfont{\bbb}{msbm10 scaled 500}

\newfont{\bb}{msbm10 scaled 1100}

\def\etal{\emph{et al.~}}

\title{Repair Strategies for Storage on Mobile Clouds}

\author{\IEEEauthorblockN{Gokhan Calis, Swetha Shivaramaiah, O. Ozan Koyluoglu, and Loukas Lazos}\\
\IEEEauthorblockA{
Department of Electrical and Computer Engineering\\
The University of Arizona\\
Email: \{gcalis, sshivaramaiah, ozan, llazos\}@email.arizona.edu}
}

\begin{document}

\maketitle


\begin{abstract}
	We study the data reliability problem for a community of devices forming a mobile cloud storage system. We consider the application of regenerating codes for file maintenance within a geographically-limited area. Such codes require lower bandwidth to regenerate lost data fragments compared to file replication or reconstruction. We investigate threshold-based repair strategies where data repair is initiated after a threshold number of data fragments have been lost due to node mobility. We show that at a low departure-to-repair rate regime, a {\em lazy repair} strategy in which repairs are initiated after several nodes have left the system outperforms {\em eager repair}  in which repairs are initiated after a single departure. This optimality is reversed when nodes are highly mobile. We further compare distributed and centralized repair strategies and derive the optimal repair threshold for minimizing the average repair cost per unit of time, as a function of underlying code parameters. In addition, we examine cooperative repair strategies and show performance improvements compared to non-cooperative codes. We investigate several models for the time needed for node repair including a simple fixed time model that allows for the computation of closed-form expressions and a more realistic model that takes into account the number of repaired nodes. We derive the conditions under which the former model approximates the latter. Finally, an extended model where additional failures are allowed during the repair process is investigated. Overall, our results establish the joint effect of code design and repair algorithms on the maintenance cost of distributed storage systems.
\end{abstract}

\begin{IEEEkeywords}
	Distributed storage, regenerating codes, mobile cloud, data reliability.
\end{IEEEkeywords}


\section{Introduction}

Local caching and content distribution from a community of mobile devices has been proposed as an alternative architecture to traditional centralized storage \cite{golrezaei2012device,paakkonen2013device, golrezaei2014scaling}. The so-called mobile cloud storage systems reduce the traffic load of the already over-burdened infrastructure network and improve content availability in the event of network outages. In a mobile cloud storage scenario, a file $\F$ is stored within a geographically-limited area  $\A$ by a community of mobile devices.  A user within $\A$ can download $\F$ from the community of mobile devices. This can be done via direct communication between the mobile devices without accessing the network infrastructure, or via a base station, but without accessing the backhaul network. Without loss of generality, we abstract our setup to the former model.   

File storage at mobile devices leads to frequent data loss due to mobility. When a mobile device storing $\F$ or any fragment of $\F$ exits $\A,$ the stored data is lost. To deal with such losses, redundancy is introduced in the form of data replication or coding \cite{bhagwan2004, dabek2004designing}. In replication storage, copies of $\F$ are stored at multiple devices within the community. More sophisticated coding schemes such as erasure coding achieve the same reliability at lower storage overhead \cite{hu2010cooperative, weatherspoon2002}. Despite the application of coding, a stored file $\F$ will eventually be lost when a certain number of mobile devices (storage nodes) depart from $\A$. To maintain $\F$ over long time periods, the mobile cloud system must be capable of recovering the lost data. A repair scenario is shown in Fig. \ref{fig:net_model}. Lost data is recovered by downloading fragments from the storage nodes that remain within $\A$. The amount of  data downloaded for repair is referred to as the \textit{repair bandwidth}. For mobile communities, the repair bandwidth can be significant. 

\begin{figure}[tb]
	\begin{center}
		\includegraphics[width=3in]{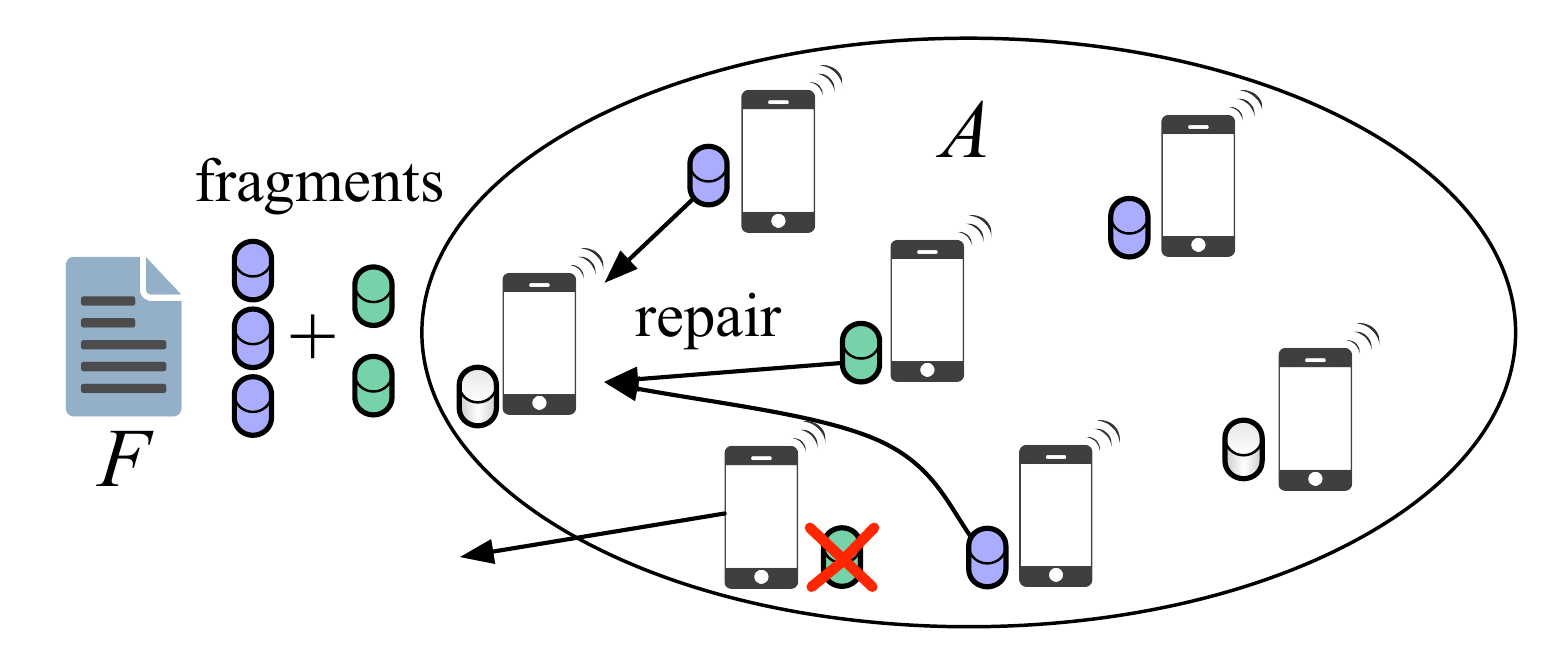}
	\end{center}
	\vspace{-0.15in}
	\caption{File maintenance in a mobile cloud storage system.}
	\vspace{-0.2in}
	\label{fig:net_model}
\end{figure}

The file maintenance problem for distributed storage systems has been primarily studied assuming that erasure codes are applied for redundancy \cite{giroire2010, hu2010cooperative}.  However, erasure codes are not repair bandwidth efficient. The repair bandwidth can be reduced by applying regenerating codes, which allow fragment recovery without file reconstruction (see \cite{dimakis2010k,  Rashmi:Optimal11,Tamo:Zigzag13,Cadambe:Asymptotic13} and references therein). Although regenerating codes lower the repair bandwidth (per single node repair), the design of an efficient repair strategy for a mobile cloud storage system involves cost optimizations with respect to many parameters, including the code redundancy factor, the device departure and fragment repair rates, the threshold for initiating repair operations, and the available communication bandwidth.  In this paper, {\em we study the problem of minimizing the file maintenance cost, as a function of the network dynamics, the code parameters, and the communication model for repairing lost data fragments}. Specifically, we make the following contributions. 
\vspace{-0.03in}
\begin{itemize}
	\item We focus on threshold-based file maintenance strategies, in which repairs are initiated when a threshold number of fragments is lost. We analyze two communication models, namely \emph{distributed repair} and \emph{centralized repair}. In distributed repair, the new storage nodes  independently download data from existing nodes to recover lost fragments. In centralized repair, a {\em leader} node first recovers  $\F$ via reconstruction, before regenerating and distributing the repaired fragments to new storage nodes. In both scenarios, we assume that repairs are performed in parallel, taking the same amount of time and there are no additional failures during fragment recovery. This simplified model allows us to derive closed-form expressions.
	
	\item We derive the \emph{optimal repair threshold} that minimizes the average repair cost per unit of time for each communication model. Our results show that no one strategy is optimal for all possible system configurations and mobility patterns. At the low mobility-to-repair rate regime, repairing at the regeneration threshold yields the optimal strategy. On the other hand, at the high mobility-to-repair rate regime, regenerating after a single fragment loss minimizes the average repair cost per unit of time.  
	
	\item We further investigate the application of cooperative repair codes. We show that the repair bandwidth is minimized at full cooperation, i.e., when all nodes to be repaired cooperate. We then investigate the centralized repair of multiple node failures, which suits our centralized repair model that is described earlier. The advantage of such a model is that a leader node does not need to download the file $\F$, which reduces the average repair cost per time.
	
	\item We revise the fixed-rate repair model originally assumed in distributed and centralized repair with a more realistic node-dependent model. In the latter, the repair time depends on the number of nodes that are repaired. We compare the resulting average repair cost with our earlier model and show that in the low mobility-to-repair rate regime, the simplified repair-all-at-one model faithfully approximates the node-dependent one. 
	
	\item We further consider a distributed repair model under a repair process which not only depends on the threshold in terms of recovery time but also may involve additional failures during recovery. We express the average repair cost through a system of equations and verify our analytical findings through simulations. We then verify our analytical findings through simulations. Lastly, we compare all of the discussed distributed repair models employing regenerating codes. 
	
\end{itemize}

\begin{figure*}[tb]
	\begin{center}
		\renewcommand{\arraystretch}{0.1}
		\begin{tabular}{cc}
			\includegraphics[width=2.8in,height=1.8in]{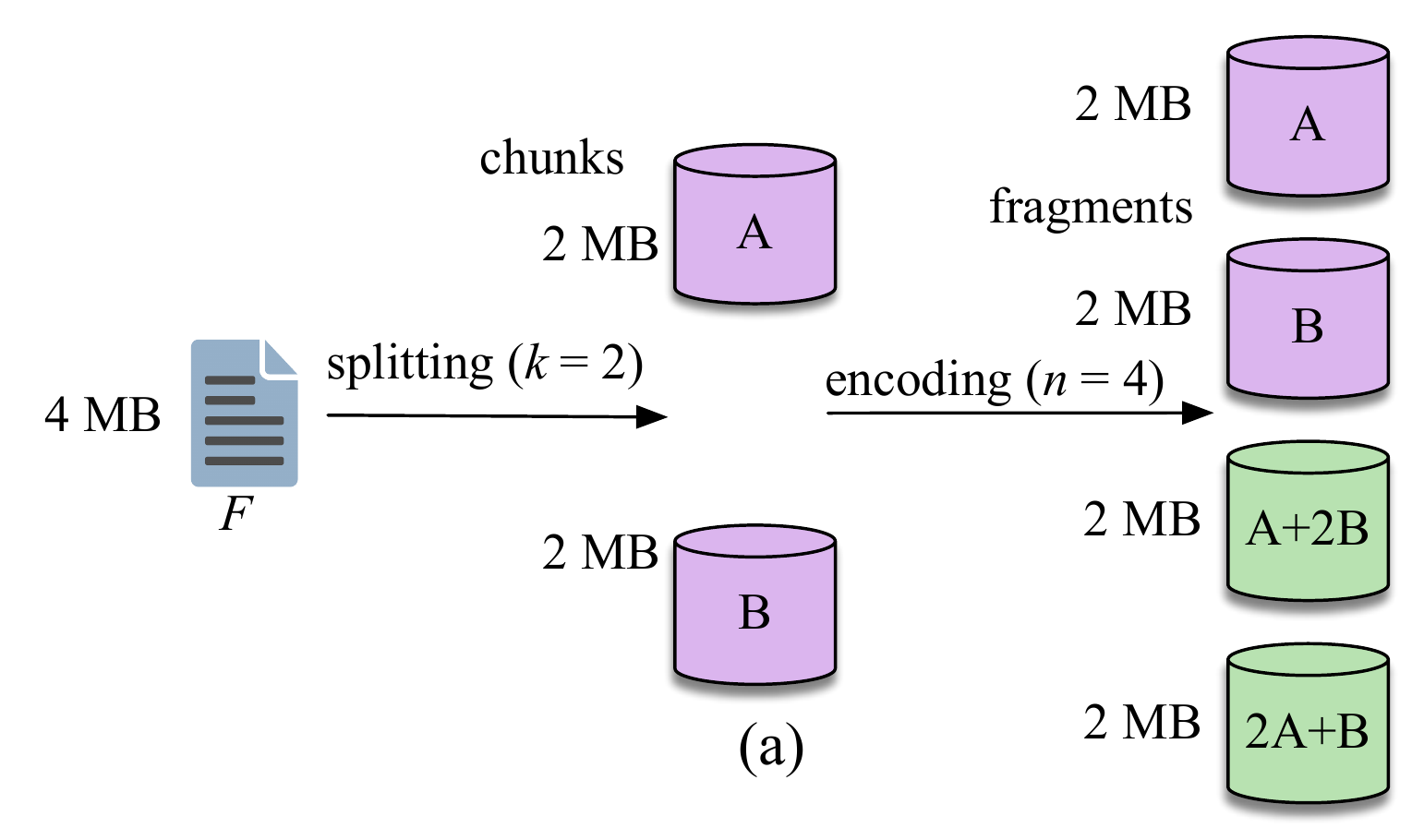} ~~~~~~~&~~~~~~~ \includegraphics[width=2.8in]{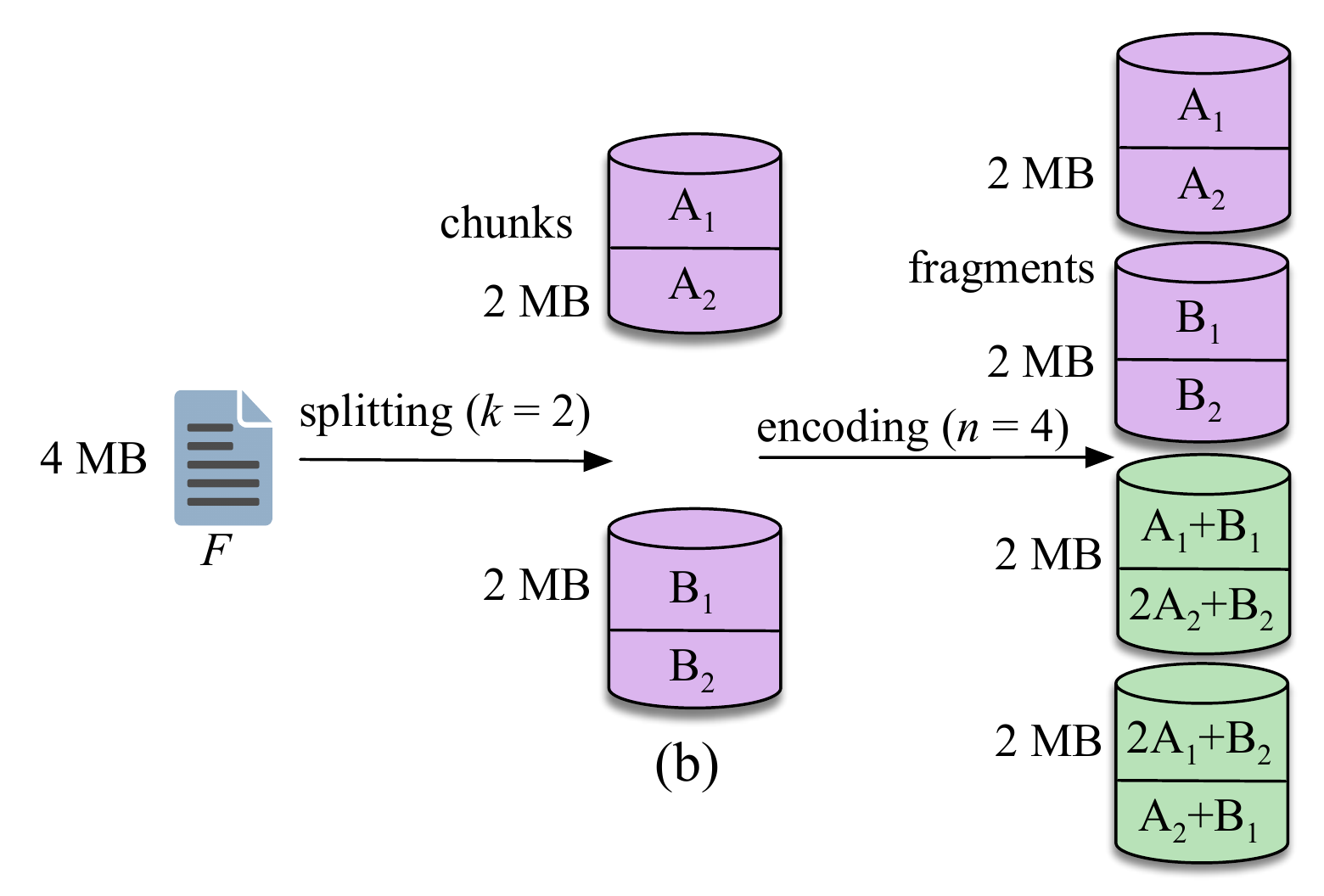} 
		\end{tabular}
	\end{center}
	\vspace{-0.15in}
	\caption{Storage of $F$ using (a) a $(n=4, k=2)$ erasure code and, (b)  a $(n=4, k=2, d=3, \alpha=2, \beta=1)$ regenerating code.}
	\vspace{-0.2in}
	\label{fig:erasure}	
\end{figure*}

Although we present our analysis in the context of mobile storage systems, we emphasize that our work is applicable in any distributed storage system where fragment losses can occur. This includes popular wired distributed storage architectures such as HDFS \cite{borthakur2008hdfs}, in which fragment loss can be frequent due to server failure and updates. Similar to a mobile storage cloud system, the optimal threshold repair strategy for the wired domain would depend on the fragment loss rate (server failure or unavailability rate) and the various system parameters.



\section{Related Work}
\label{sec:related}

In reliable storage systems,  information is replicated or coded such that the original content can be recovered if some limited fraction of the stored data is lost.  Replication is the most intuitive way to introduce redundancy. This method refers to the maintenance of verbatim copies of the same file $\F$. Although replication is easy to implement, it suffers from high storage and repair overhead.

Erasure codes incur less storage overhead compared to replication while maintaining the same degree of reliability. In particular, Maximum Distance Separable (MDS) codes achieve the optimal tradeoff between failure tolerance and storage overhead \cite{Singleton:Maximum64, Blaum:EVENODD95}. An $(n, k)$ MDS code encodes $k$ data chunks to $n$ fragments  and can tolerate up to $n-k$ fragment losses. Any $k$ encoded fragments can be used to reconstruct $\F$. Fig. \ref{fig:erasure}(a) shows the encoding process for a file $\F$ of size $4$MB using a (4, 2) erasure code. File $\F$ is split into $k=2$ chunks $A$ and $B$, each of size $2$MB. The two chunks are then encoded into $n=4$ fragments. The repair bandwidth for this scheme equals the size of the original file.  Reed-Solomon codes are a classical example of MDS codes and are deployed in many existing storage systems (e.g. \cite{weatherspoon2002, Huang:STAR05, calder2011windows, ford2010availability}). 

Although erasure codes offer significant savings in storage, their repair bandwidth is suboptimal, a data amount equal to the file size must be retrieved to repair a single fragment. Regenerating codes, on the other hand, can recover lost fragments without reconstructing the entire file, at the expense of a small storage overhead. They were initially investigated in the seminal work of Dimakis \etal \cite{dimakis2010k}, which focuses on the following setup. A file $\mathcal{F}$ of size $\mathcal{M}$ symbols is encoded into $n$ fragments, each of size $\alpha$ symbols, such that (i) the file can be reconstructed from any $k$ fragments, and (ii) a lost fragment can be repaired by downloading $\beta\leq \alpha$ symbols from any $d\geq k$ fragments, resulting in a repair bandwidth of $\gamma=d\beta$.  Dimakis \etal characterized the tradeoff between the per node storage ($\alpha$) and the repair bandwidth ($\gamma$) \cite{dimakis2010k}. 

Fig.~\ref{fig:erasure}(b) shows an example of a $(n,k,d,\alpha,\beta)=(4,2,3,2,1)$ regenerating code. Here, the file $\F$ is split into $k = 2$ chunks each of size $\alpha=2$MB. The chunks are encoded in $n = 4$ fragments, with each fragment being $2$MB. A failed node in this scenario can be regenerated by retrieving fragments of size $\beta=1$MB from $d=3$ surviving nodes. This yields a repair bandwidth of $d\beta=3$MB which is less that  $k\alpha=4$MB. Note, however, that regeneration can be applied only if at least $d$ fragments are available. If fewer than $d$ but more than $k$ fragments remain available, the lost fragments can only be repaired through file reconstruction.  

During the repair process of regenerating codes, there is no coordination among the nodes to be repaired. In  \cite{kermarrec2011repairing,hu2010cooperative,Shum2013Cooperative,wang2010mfr}, the authors consider the case where $t$ storage nodes are repaired simultaneously in a cooperative manner. Specifically, referring to this set of $t$ nodes as the newcomers, and the existing nodes storing fragments of $\F$ as live nodes, each newcomer contacts $d$ live nodes and downloads $\beta$ symbols from each. Moreover, newcomers cooperate and download $\beta'$ symbols from each of the remaining $t-1$ newcomers. The tradeoff between per node storage and repair bandwidth is established similarly to \cite{dimakis2010k}. Rawat \etal in \cite{rawal2016centralized} also consider the cooperative repair of $t$ nodes such that only one node among $t$ nodes downloads data from live nodes. After downloading the necessary information at the leader node, the remaining $t-1$ nodes are cooperatively repaired. Two points corresponding to minimum storage and minimum bandwidth regeneration are characterized. 

In the context of mobile cloud systems, P{\"a}{\"a}kk{\"o}nen \etal  considered a wireless device-to-device network used for distributed storage \cite{paakkonen2012distributed}. The authors showed the energy consumption for maintaining data using regenerating codes is lower compared to retrieving a lost file from a remote source. This result holds if the per-bit energy cost for communication between the mobile devices is lower than the cost for communicating with the remote source. 

In a follow-up work, P{\"a}{\"a}kk{\"o}nen \etal  compared replication with regeneration for a similar wireless P2P storage system \cite{paakkonen2013device}. 
They derived closed-form expressions for the expected total energy cost of file retrieval using replication and regeneration. They showed that the expected total cost of $2$-replication is lower than the cost of regeneration. However, only an eager repair strategy was considered in the analysis. Moreover, the advantages of regeneration were not fully exploited by considering codes with different parameters.  P{\"a}{\"a}kk{\"o}nen \etal also addressed the problem of tolerating multiple simultaneous failures \cite{paakkonen2014device}. They investigated the energy overhead of regenerating codes in a cellular network. They showed that large energy  gains can be obtained by employing regenerating codes. These gains depend on the file popularity. The authors  provided decision rules for choosing between simple caching, replication, MSR and MBR codes, based on numerical results on certain application scenarios. In our work, we analytically provide decision rules to choose optimal repair strategies that minimize the repair bandwidth per unit of time.

Pedersen \etal recently studied the cost of content caching on mobile devices using erasure codes \cite{pedersen2016distributed}. They derived analytical expressions for the cost of content download and repair bandwidth as a function of the repair interval. These expressions were used to evaluate the communication cost of distributed storage for MDS codes, regenerating codes, and locally repairable codes. Their results show that in high churn, distributed storage can reduce the communication cost compared to downloading from a base station. They conclude that MDS codes are the best performers in this setup. 



\section{System Model}
\label{sec:model}

\subsection{Network Model} 

We consider a distributed storage system (DSS) consisting of mobile storage nodes that enter and exit a geographically-limited area $\A$. When a node departs from $\A$, its data is lost. The nodes that store file fragments within area $\A$ are said to be \emph{live nodes}. New nodes that are used to store repaired fragments are said to be \emph{newcomer nodes} or newcomers. We assume that there are always sufficient newcomers to perform repairs. Moreover, as we are interested in the system performance due to network dynamics, we do not consider data loss due to hardware failures. Such failures occurs orders of magnitude less frequently than node departures. Following the network dynamics model of prior works \cite{giroire2010, paakkonen2012distributed}, we model the time $X_i$ spent by each node within $\A$ as an exponentially distributed random variable with parameter $\lambda$ (i.e., $X_i \sim \textrm{Exp}(\lambda),~\forall i)$. Random variables $\{X_i\}$ are assumed independent and identically distributed.  

The repair time is modeled by an exponentially distributed random variable with parameter $\mu$. For ease of analysis, we initially assume that $\mu$ is independent of the number of fragments that need to be repaired. We later revise our analysis and consider a more realistic model in which repairs proceed in parallel at different nodes with the same rate $\mu$. This corresponds to the distributed nature of mobile DSS. Finally, we define $\rho=\frac{\lambda}{\mu}$ as the ratio of the departure-to-repair rate.

\begin{figure*}[tb]
	\begin{center}
		\begin{tabular}{cc}
			\includegraphics[height=1.3in]{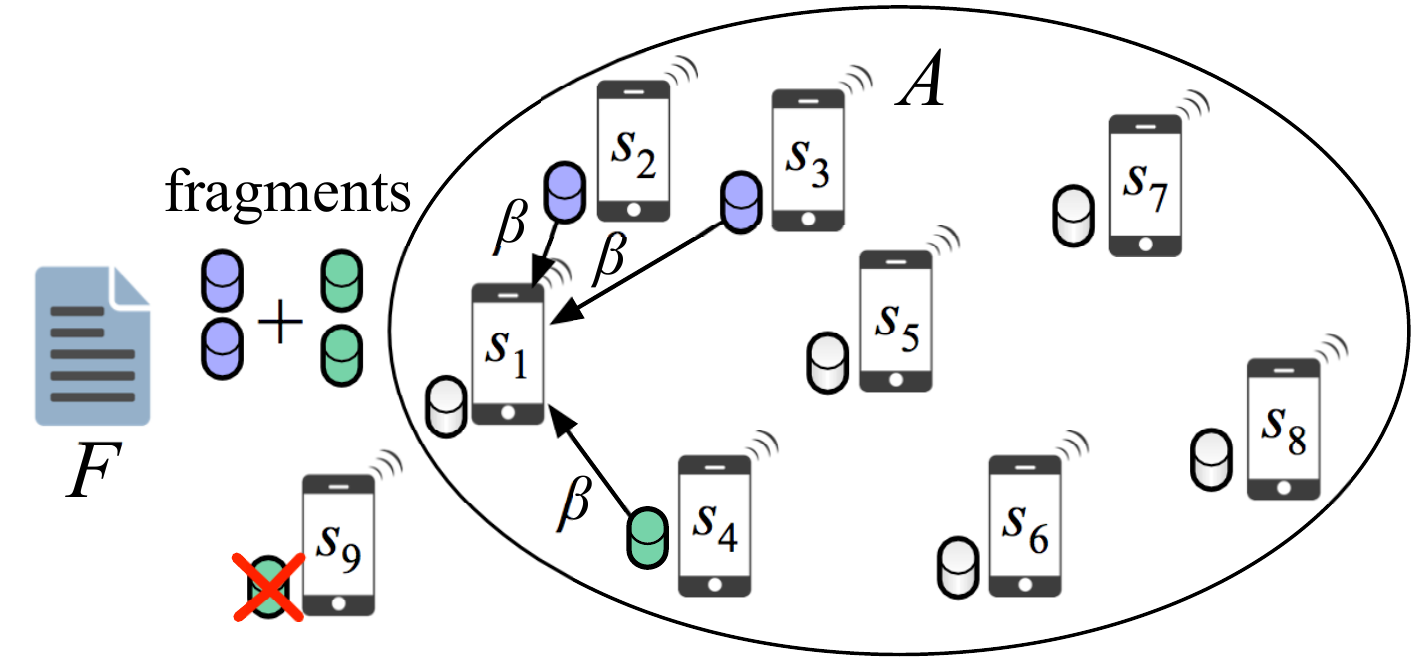}~~~~~~&~~~~~~\includegraphics[height=1.3in]{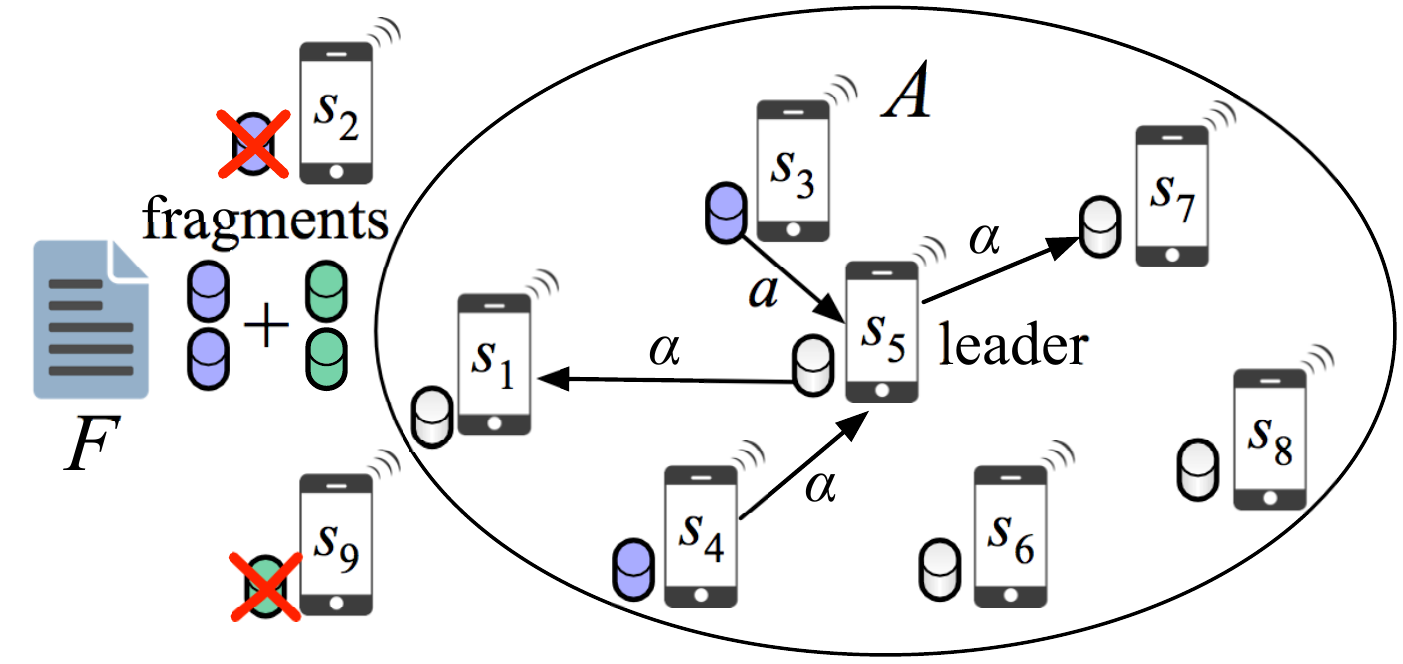}\\			
			(a) distributed repair ~~~&~~~ (b) centralized repair 
		\end{tabular}
		\vspace{-0.1in}
	\end{center}
	\caption{  (a) Distributed repair: nodes independently regenerate a lost fragment by obtaining symbols from other nodes, (b) centralized repair: a leader node reconstructs $\mathcal{F}$ and distributes lost fragments to new nodes.}
	\label{fig:repairStrat}
	\vspace{-0.2in}
\end{figure*}

\subsection{Storage Model} 

A file $\mathcal{F}$ of size $\mathcal{M}$ bits is stored in $n$ storage nodes using a regenerating code with parameters $(n,k,d,\alpha,\beta)$ (see Fig.~\ref{fig:erasure}(b)). We focus on the two most popular types of regenerating codes, namely  Minimum Storage Regenerating (MSR) codes and Minimum Bandwidth Regenerating (MBR) codes. These two classes of codes operate at the end points of the tradeoff between per node storage and repair bandwidth, as introduced in \cite{dimakis2010k}. MSR codes achieve minimum storage by setting $\alpha=\mathcal{M}/k$ and minimize the repair bandwidth under this constraint. Their operating point is given by:
\begin{equation}
\left(\alpha_{\textrm{MSR}}, \gamma_{\textrm{MSR}}\right)=\left(\frac{\mathcal{M}}{k},\frac{\mathcal{M}d}{k\left(d-k+1\right)}\right).
\label{msr}
\end{equation}
Note that, for MSR codes, $\alpha_{\textrm{MSR}} \leq \gamma_{\textrm{MSR}}$ and hence, the per-node storage is smaller than the repair bandwidth. MBR codes, on the other hand, minimize the repair bandwidth (achieved when $\gamma=\alpha$), and operate at:
\begin{equation}
\left(\alpha_{\textrm{MBR}}, \gamma_{\textrm{MBR}}\right)=\left(\frac{2\mathcal{M}d}{2kd-k^{2}+k},\frac{2\mathcal{M}d}{2kd-k^{2}+k}\right).
\label{mbr}
\end{equation}
Instances of these codes can be found in \cite{Cadambe:Asymptotic13, Rashmi:Optimal11, Tamo:Zigzag13}.


\subsection{File Repair Model}

In our model, the system continuously monitors the redundancy level and initiates a repair when $\tau$ live nodes remain within $\A$. The determination of $\tau$, the type of repair (regeneration, reconstruction, or both) and the communication model for fragment retrieval (centralized or distributed) form a {\em file maintenance strategy}. We note that the practical implementation details of the redundancy monitoring mechanism and of the communication protocols for retrieving various fragments are beyond the scope of the present work. We focus on the theoretical aspects of the maintenance process. Since repairs are initiated only when the number of remaining nodes reaches threshold $\tau$, a repair strategy can be viewed as an i.i.d. system recovery process occurring every $\Delta$ seconds, where $\Delta$ is a random variable denoting the time elapsed between two instances of a fully repaired system. For this recovery process, we define the following costs.

\begin{defn}[Repair cost $c(\tau)$] The number of bits $c(\tau)$ that must be downloaded from the $\tau$ remaining nodes to restore $n$ fragments in $\A,$ when $n - \tau$ nodes have departed $\mathcal{A}$. 
\end{defn}

\begin{defn} [Average repair cost per unit of time~$r(\tau)$] The average cost per unit of time for maintaining $n$ fragments in $\A$, defined as $c(\tau)$ over the average time between two instances of a fully repaired system, i.e., $\textrm{E}[\Delta]$, with $n$ fragments ($r(\tau)$ is measured in bits per unit of time). 
\end{defn}

We determine the optimal file maintenance strategy for different node departure rates, code parameters, and communication models for fragment retrieval. 


\section{File Maintenance Strategies}
\label{sec:repair_strategies}

Let $\tau$ denote the number of live nodes remaining within $\A$ after the departure of $n-\tau$ nodes. We focus on determining the optimal repair threshold $\tau^*,$ which minimizes the average repair cost per unit of time. We first compare the distributed repair strategy with centralized repair strategy.

\subsection{Distributed Repair}

In distributed repair, newcomers recover lost fragments by independently downloading relevant symbols from live nodes. The repair process is initiated when  $\tau$ live nodes remain within $\A$, where $k \leq \tau <n-1$ (when $\tau<k$, the data is irrecoverably lost).  If $\tau \geq d$, fragment recovery can be performed through regeneration. Each of the $n-\tau$ newcomers downloads $\beta$ symbols from $d$ live nodes and independently regenerates a lost fragment. Fig.~\ref{fig:repairStrat}(a) demonstrates the distributed repair process for a file $\mathcal{F}$ stored with a $(n=4, k=2, d=3, \alpha=2, \beta=1)$ regenerating code. One fragment of $\mathcal{F}$ is lost because node $s_9$ departed from $\mathcal{A}$. The lost fragment is regenerated at $s_1$ by independently downloading $\beta=1$ symbol from three nodes. The total repair bandwidth is equal to 3 symbols.

If $\tau<d,$ regeneration cannot be directly applied. To reduce the repair cost, we consider a hybrid scheme consisting of regeneration and reconstruction. First,  $d-\tau$ nodes are repaired by downloading $\alpha$ symbols from $k$ live nodes and reconstructing $\F$. When $d$ fragments become available, regeneration is applied to repair the remaining $n-d$ newcomers. Accordingly, the repair cost is expressed by:
\begin{equation}
\label{eq1}
c_D(\tau)=\begin{cases}
k\alpha(d-\tau) + \gamma(n-d),   &\mbox{  if  } \tau<d \\
\gamma(n-\tau),  &\mbox{  if  } \tau \geq d.
\end{cases}
\end{equation}
The subscript $D$  in $c_D(\tau)$ is used to denote the cost of distributed repair and $\gamma$ denotes the regeneration cost of a single fragment which depends on the underlying regeneration code (see eqs. \eqref{msr} and \eqref{mbr} for MSR and MBR codes, respectively). From \eqref{eq1}, it is evident that $c_D(\tau)$ monotonically decreases with $\tau.$ Moreover, the rate of cost change (with respect to $\tau$) is higher when $\tau < d.$ To determine the optimal threshold $\tau^*$, we are interested in minimizing $r_D(\tau)$, which captures the repair cost for maintaining $n$ fragments per unit of time. 

To calculate $r_D(\tau)$, we use the continuous-time Markov chain (CTMC) model shown in Fig.~\ref{fig:mc_d}. This model captures the periodic repair process when node departures occur independently, the time spent by each node in $\mathcal{A}$ is exponentially distributed with parameter $\lambda$,  and the system recovery process is exponentially distributed with parameter $\mu$. 
\begin{figure}[hb]
	\begin{center}
		\includegraphics[width=2.8in]{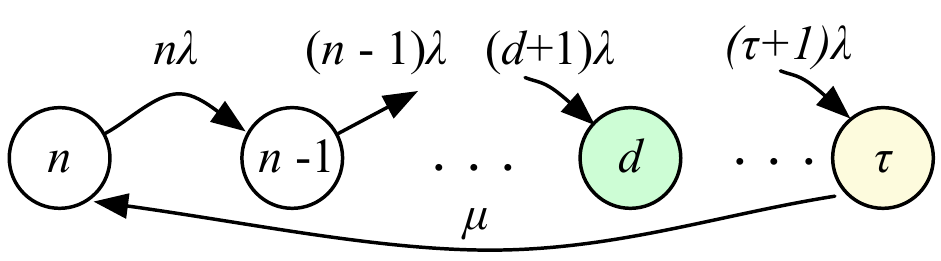}
	\end{center}
	\vspace{-0.1in}
	\caption{Markov chain for a threshold-based file maintenance.}
	\label{fig:mc_d}
	\vspace{-0.1in}
\end{figure}

The CTMC consists of $n-\tau+1$ states representing the number of fragments that remain within $\A$ after each node departure, until a repair at state $\tau$ is initiated. Note that we have omitted states after $\tau$ in the CTMC model, because we are interested in optimizing the periodic cost of repairing the DSS at threshold $\tau$. Moreover, the transition probability to state  $\tau-1$ is negligible for most realistic scenarios in which $\mu \gg \tau\lambda.$ For cases when $\mu \not \gg \tau\lambda$, we compute the mean time it takes to depart from the optimal repair strategy of repairing at state $\tau$ and interpret this event as a form of system error which leads to data loss (see Section~\ref{ssec:mttdl}). 

For the CTMC in Fig~\ref{fig:mc_d}, the departure rate from a state $i$ equals the node departure rate $\lambda$, times the number of nodes which store fragments at state $i$. When the repair process is initiated, the system transitions from state $\tau$ to state $n$ because all fragment repairs nodes proceed in parallel. For the CTMC, we define the expected average cost $r_D(\tau)$ per unit of time as
\begin{equation}
r_D(\tau) = \frac{c_D(\tau)}{\textrm{E}[\Delta]},
\label{rtau}
\end{equation}
where $E[\Delta]$ is the average time between two transitions through the $n^{th}$ state in the periodic repair process\footnote{The alternative definition of  $r_D(\tau) = E\left [ \frac{c_D(\tau)}{\Delta}\right]$ is not useful because the expectation is infinite. This is due to the infinitesimally small values that can be obtained by $\Delta$, whereas  $c_D(\tau)$ remains lower bounded.}. For $\Delta$,
\begin{equation}
\Delta = T_n + T_{n-1} + \dots + T_{\tau+1} + T_{\tau},
\label{delta}
\end{equation}  
where $T_i$ denotes the time that the system stays at state $i$ (inter-departure time) and $T_{\tau}$ is the expected time for completing repairs so that $n-\tau$ fragments are recovered (return to state $n$). The random variables $T_i$ are independent and exponentially distributed with parameter $i\lambda$, whereas $T_{\tau}$ is exponentially distributed with parameter $\mu$. In particular, $\textrm{E}[T_i]=\frac{1}{i\lambda}$ and $\textrm{E}[T_\tau]=\frac{1}{\mu}$. Therefore, $\textrm{E}[\Delta]$ is the sum expectation of independent exponential random variables.
\begin{equation}
\textrm{E}[\Delta] = \sum_{i=\tau+1}^{n}\frac{1}{i\lambda}  + \frac{1}{\mu}= \frac{H_{n,\tau}}{\lambda} + \frac{1}{\mu},
\label{e[delta]}
\end{equation}
where $H_{n,\tau} = \sum_{i=\tau+1}^{n} \frac{1}{i}$. Combining \eqref{rtau} and \eqref{e[delta]}, we obtain the average repair cost per unit of time  as follows.
\begin{equation}
r_D(\tau) = \frac{c_D(\tau)}{\textrm{E}[\Delta]} =
\begin{cases}
\frac{\lambda \mu(k\alpha(d-\tau)+\gamma (n-d))}{\mu H_{n,\tau} + \lambda}, & \textrm{if } \tau < d \\
\frac{\lambda \mu(\gamma (n-\tau))}{\mu H_{n,\tau} + \lambda}, & \textrm{if } \tau \geq d. \\
\end{cases}
\label{rtau_d}
\end{equation}

We use \eqref{rtau_d} to determine the optimal threshold $\tau^*$ which minimizes $r_D(\tau).$ This is given by Propositions 1 and 2.

\begin{prop}
	\label{thm:prop1}
	For regeneration ($d \leq \tau \leq n-1$), the optimal repair threshold $\tau^*$ is given by 
	\begin{equation}
	\tau^* = \begin{cases}
	d, & \rho \leq \frac{H_{n-1,d}}{n-d-1} - \frac{1}{n} \\
	n-1, & \text{otherwise.}
	\end{cases}
	\label{opttau}
	\end{equation}
\end{prop}
\begin{proof}
Proof is provided in Appendix \ref{sec:lambdaProp1}.			
\end{proof}

Proposition \ref{thm:prop1} determines the $\rho$ regime for which repairs at $\tau = d$, an instance of {\it lazy repair}, is more efficient than initiating repairs at $\tau= n-1$, referred to as  {\it eager repair}. In the following Lemma, we show that there is always a positive $\rho$ for which lazy repair is more efficient that eager repair.

\begin{lem}
	\label{thm:lambdaProp1}
	There is always some $\rho>0$ for which lazy repair ($\tau^{\ast} = d$) is more efficient than eager repair $(\tau=n-1)$, independent of the code parameters used for regeneration.
\end{lem}

\begin{proof}
Proof is provided in Appendix \ref{sec:LemmalambdaProp1}.			
\end{proof}

We now examine if there is a $\rho$ regime for which the hybrid scheme, i.e., reconstruction plus regeneration results in a lower expected cost per unit of time compared to regeneration only. This rate regime is given by the following proposition. 
	
\begin{prop}
	\label{thm:prop2}
	For regeneration plus reconstruction ($k \leq \tau \leq d$), the optimal repair threshold $\tau^*$  is given by 
	\begin{equation}
	\tau^* = \begin{cases}
	k, & \rho \leq  \frac{\gamma (n-d)H_{d,k}}{k\alpha (d-k)} - H_{n,d} \\
	d, & \textrm{otherwise.}
	\end{cases}
	\label{opttau_2}
	\end{equation}
\end{prop}

\begin{proof}
Proof is provided in Appendix \ref{sec:AppA}.
\end{proof}

Similar to Lemma~\ref{thm:lambdaProp1}, we investigate if the highest departure-to-repair rate for which reconstruction at $k$ is more efficient than regeneration is always positive independent of the code parameters. Unlike the case of Lemma~\ref{thm:lambdaProp1}, we show that for a certain relationship between $n, k, \gamma,$ and $\alpha,$ regeneration is strictly more efficient than regeneration plus reconstruction, independent of $\rho$. For any other code parameters, the most efficient strategy  depends on $\rho$. 
\begin{lem}
	\label{thm:lambdaProp2}
	For any departure-to-repair ratio $\rho$, regeneration is strictly more efficient than regeneration plus reconstruction for codes satisfying $n\gamma < k^2\alpha$.
\end{lem}

\begin{proof}
Proof is provided in Appendix \ref{sec:lemmalambdaProp2}. 
\end{proof}

We further explore the condition in Lemma~\ref{thm:lambdaProp2} for MSR and MBR codes. For MSR codes, we obtain that $dn <k^2 (d - k +1)$
by substituting the operation points of MSR from \eqref{msr}. Similarly, for MBR codes, we obtain that $n <k^2$ by substituting the operation points of MBR from \eqref{mbr}. Note that Lemma~\ref{thm:lambdaProp2} does not enumerate all possible codes for which regeneration is strictly more efficient than regeneration plus reconstruction for any $\lambda.$ This is because we have used bounds on the harmonic function to derive the analytic formulas. Numerical bounds could provide a more accurate range of code parameters for which Lemma~\ref{thm:lambdaProp2} is true. 

\begin{table*}[]
	\centering
	\caption{Cost comparison of repair strategies at different thresholds.}
	\label{table:costTable}
	\begin{tabular}{|c|c|c|c|c|c|}
		\hline
		\multirow{2}{*}{} & \multicolumn{3}{c|}{Distributed Repair}                                                                                            & \multicolumn{2}{c|}{Centralized Repair} \\ \cline{2-6} 
		& \multicolumn{2}{c|}{Regeneration} & \multicolumn{1}{c|}{\begin{tabular}[c]{@{}c@{}}Regeneration\\ +\\ Reconstruction\end{tabular}} & \multicolumn{2}{c|}{Reconstruction}     \\ \hline
		Code         & $r_D(n-1)$          & $r_D(d)$          & $r_D(k)$                                                                                       & $r_C(n-1)$              & $r_C(k)$              \\  \hline 
		MSR  & $\frac{n\mathcal{M} d \lambda \mu}{k(d-k+1)(\mu+n\lambda)} $  & $\frac{\mathcal{M}(n-d)d\lambda\mu}{k(d-k+1)(\lambda+\mu H_{n,d})}$   &  $\frac{\mathcal{M}\left[k(d-k+1)(d-k)+d(n-d)\right] \lambda\mu}{k(d-k+1)(\lambda+\mu H_{n,k})}$   & $\frac{n\mathcal{M}\lambda\mu}{\mu+n\lambda}$   &     $\frac{(n-1)\mathcal{M}\lambda\mu}{k(\lambda + \mu H_{n,k})}$               \\ [5pt] \hline
		MBR  & $\frac{2n\mathcal{M}d\lambda \mu}{k(2d-k+1)(\mu+n\lambda)}$  &  $\frac{2\mathcal{M}(n-d)d\lambda \mu}{k(2d-k+1)(\lambda+H_{n,d})}$   &  $\frac{2\mathcal{M}d( n+ kd-k^2 -d)\lambda\mu}{k(2d-k+1)(\lambda+H_{n,k})}$  &     $\frac{2n\mathcal{M}d\lambda\mu}{(2d-k+1)(\mu+n\lambda)}$   &           $\frac{2(n-1)\mathcal{M}d\lambda\mu}{k(2d-k+1)(\lambda + \mu H_{n,k})}$         \\ [5pt] \hline
	\end{tabular}
	\vspace{-0.2in}
\end{table*}

\subsection{Centralized Repair}
\label{subsec:centralized}
In the centralized strategy, repairs are performed by a {\em leader node} in two stages. In the first stage, the leader newcomer node downloads  $\alpha$ symbols from $k$ live nodes and reconstructs $\F$. In the second stage, the leader node transmits $\alpha$ bits to each of the remaining $(n-\tau-1)$ newcomers to restore the remaining  $(n-\tau-1)$ fragments. Fig.~\ref{fig:repairStrat}(b) shows an example of centralized repair for a $(n=4, k=2, d=3, \alpha=2, \beta=1)$ regenerating code. Nodes $s_2$ and $s_9$ have departed from area $\mathcal{A}$, leading to the loss of their respective fragments. Node $s_5,$ who acts as a leader, downloads $\alpha = 2$ symbols from $k=2$ other nodes to reconstruct $\mathcal{F}$. It then distributes $\alpha =2$ symbols to $s_1$ and $s_7$ to restore the  system reliability. The repair cost of centralized repair is given by:
\begin{equation}
\label{c_cent}
c_C(\tau)=\alpha\left(k+ n-\tau-1\right).
\end{equation}
In \eqref{c_cent}, the subscript $C$  in $c_C(\tau)$ is used to denote the cost of centralized repair. The node departure process does not vary with the repair strategy. Therefore, the same CTMC model shown in Fig. \ref{fig:mc_d} applies for the centralized repair. According to \eqref{rtau}, the average repair cost $r_C(\tau)$ is given by:
\begin{equation}
\label{rtau_c}
r_C(\tau) = \frac{c_C(\tau)}{\textrm{E}[\Delta]} = \frac{\lambda \mu \alpha (k+n-\tau-1)}{ \mu H_{n,\tau} + \lambda}. 
\end{equation}
The optimal threshold $\tau^*$ which minimizes $r(\tau)$ is obtained in Proposition \ref{thm:prop3}.

\begin{prop}\label{thm:prop3}
	The optimal repair threshold $\tau^*$ which minimizes $r(\tau)$ for centralized repair is given by 
	\begin{equation}
	\tau^* = \begin{cases}
	k, & \rho \leq \frac{kH_{n-1,k}}{n-k-1} - \frac{1}{n} \\
	n-1, & \textrm{otherwise}
	\end{cases}
	\label{opttau_1}
	\end{equation}
\end{prop}
\begin{proof}
Proof is provided in Appendix \ref{sec:AppB}. 
\end{proof}

Using Proposition \ref{thm:prop3}, we can determine the optimal repair strategy for any $\rho$, when centralized repair is employed. We note that according to Lemma~\ref{thm:lambdaProp1}, the value $\frac{kH_{n-1, k}}{(n-k-1)} - \frac{1}{n}$ is strictly positive for any code parameters. Therefore, there is always a departure-to-repair ratio for which lazy repair is more efficient than eager repair, independent of the code used for regeneration and reconstruction.


\section{Analysis of Maintenance Strategies}
\label{sec:numerical}

In this section, we characterize the $\rho$ regime for which lazy repair is more cost-efficient than eager repair. Moreover, we determine the optimal repair strategy (decentralized vs. centralized) as a function of the code parameters, when the departure and repair rates are fixed. To ease the  reader to our analysis,  we summarize the cost of repair in Table~\ref{table:costTable}. 

\subsection{Eager vs. Lazy Repair} 

According to the results of Propositions \ref{thm:prop1}, \ref{thm:prop2}, and \ref{thm:prop3}, we classify the departure-to-repair ratios into a \textit{low departure-to-repair rate regime $(\rho_{\textrm{low}}$)} and a \textit{high departure-to-repair rate regime $(\rho_{\textrm{high}})$.} The two regimes are defined by finding the lowest and highest rates, based on the bounds stated in the three propositions. 
\begin{eqnarray}
\rho_{\textrm{low}}&=\min\Big\{\frac{H_{n-1, d}}{n-d-1} - \frac{1}{n}, \frac{\gamma(n-d)H_{d,k}}{k\alpha(d-k)}-H_{n,d}, \frac{kH_{n-1, k}}{(n-k-1)} - \frac{1}{n}\Big\}.\\
\rho_{\textrm{high}}&=
\max\Big\{\frac{H_{n-1, d}}{n-d-1} - \frac{1}{n},\frac{\gamma(n-d)H_{d,k}}{k\alpha(d-k)}-H_{n,d}, \frac{kH_{n-1, k}}{(n-k-1)} - \frac{1}{n}\Big\}.
\end{eqnarray}
Noting that   $\frac{H_{n-1, d}}{n-d-1} - \frac{1}{n}< \frac{kH_{n-1, k}}{(n-k-1)} - \frac{1}{n}$ for $k<d,$ the two regime expressions can be simplified to 
\begin{align}
&\rho_{\textrm{low}}=\min\Big\{\frac{H_{n-1, d}}{n-d-1} - \frac{1}{n}, \frac{\gamma(n-d)H_{d,k}}{k\alpha(d-k)}-H_{n,d}\Big\}.\\
&\rho_{\textrm{high}}=
\max\Big\{\frac{\gamma(n-d)H_{d,k}}{k\alpha(d-k)}-H_{n,d},\frac{kH_{n-1, k}}{(n-k-1)} - \frac{1}{n}\Big\}. 
\end{align}
For any $\rho \leq \rho_{\textrm{low}}$, the repair cost per unit of time is minimized when lazy repair is applied  since that choice of $\rho$ would be lower than the bounds found in \eqref{opttau}, \eqref{opttau_2} and \eqref{opttau_1} and the corresponding repair thresholds are the lowest possible. On the other hand, for any $\rho \geq \rho_{\textrm{high}}$, eager repair (i.e., repair at $\tau^{*}=n-1$) yields the lowest $r(\tau)$. These findings hold for both distributed and centralized repair. If the departure-to-repair rates do not lie in either of the $\rho$ regimes, then the optimal repair policy (eager vs. lazy) depends on the relationship of the code parameters and the repair strategy (centralized or distributed).

\subsection{Centralized vs. Distributed Repair} 

We now fix the departure rate $\lambda$ and  repair rate $\mu$ to compare the repair cost of centralized vs. distributed repair per unit of time, as a function of the code parameters. Specifically, we determine relationships between $n,k,d$ and the code type (MSR vs. MBR) for which an optimal strategy can be derived. Our results are stated in the following two propositions. 

\begin{prop}
	For $d \leq \tau^{*} \leq n-1$, using MBR codes and distributed repair minimizes the average repair cost per unit of time, if $d>\frac{n+k-1}{3}$.
	\label{thm:prop4}
\end{prop}
\begin{proof}
Proof is provided in Appendix \ref{sec:Prop4}.
\end{proof}

We now prove that if $\tau^{*}$ lies between $k$ and $d$, using MSR codes with centralized repair is optimal.
\begin{prop}
	For $k \leq \tau^{*} < d$, the optimal repair strategy is given by centralized repair with MSR codes.
	\label{thm:prop5}
\end{prop}
\begin{proof}
Proof is provided in Appendix \ref{sec:Prop5}.
\end{proof}

\begin{figure*}[tb]
	\begin{center}
		\setlength{\tabcolsep}{-0.01in}
		\begin{tabular}{cccc}
			\includegraphics[height=1.6in,width=1.85in]{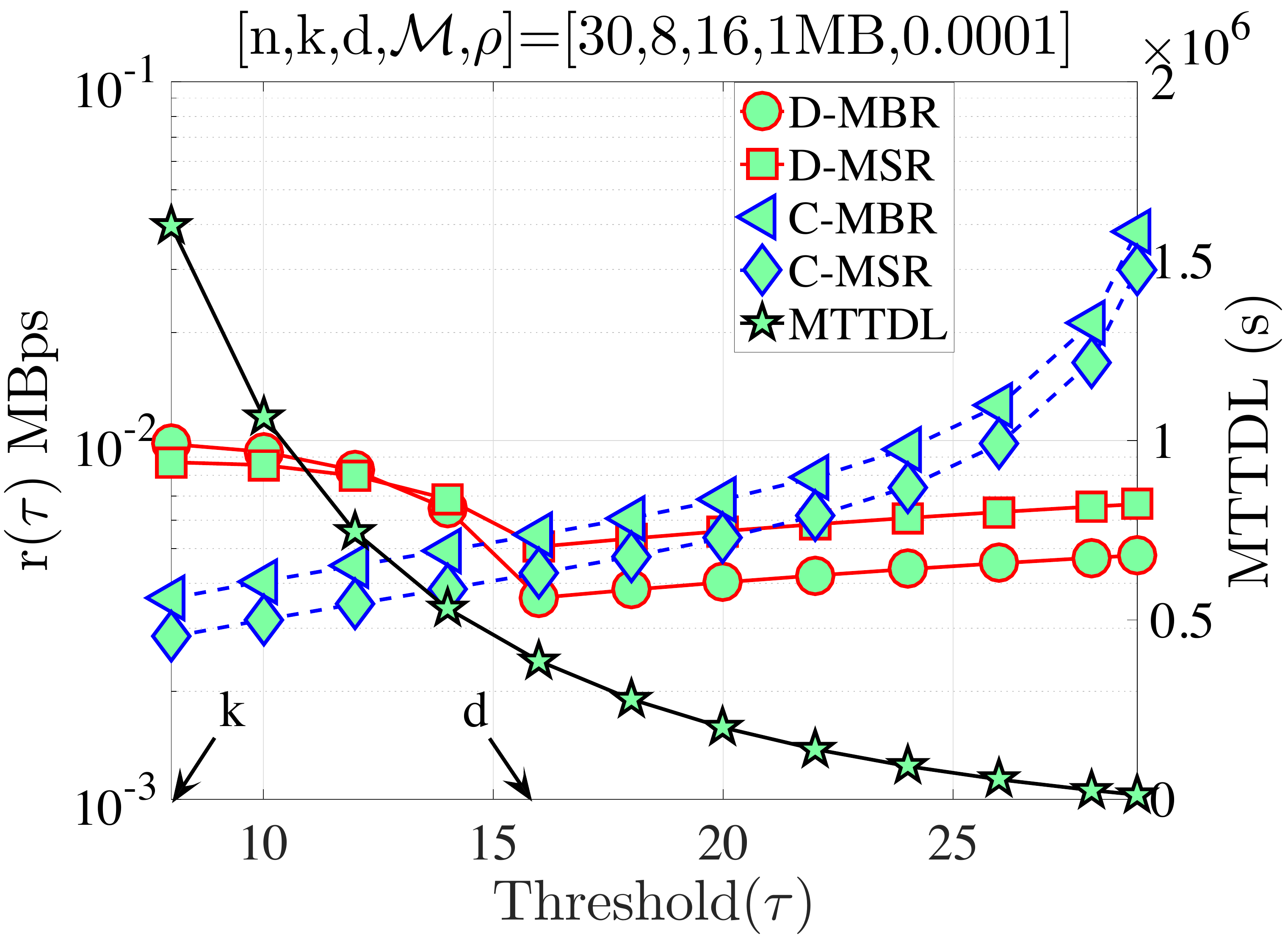} &
			\includegraphics[height=1.6in,width=1.85in]{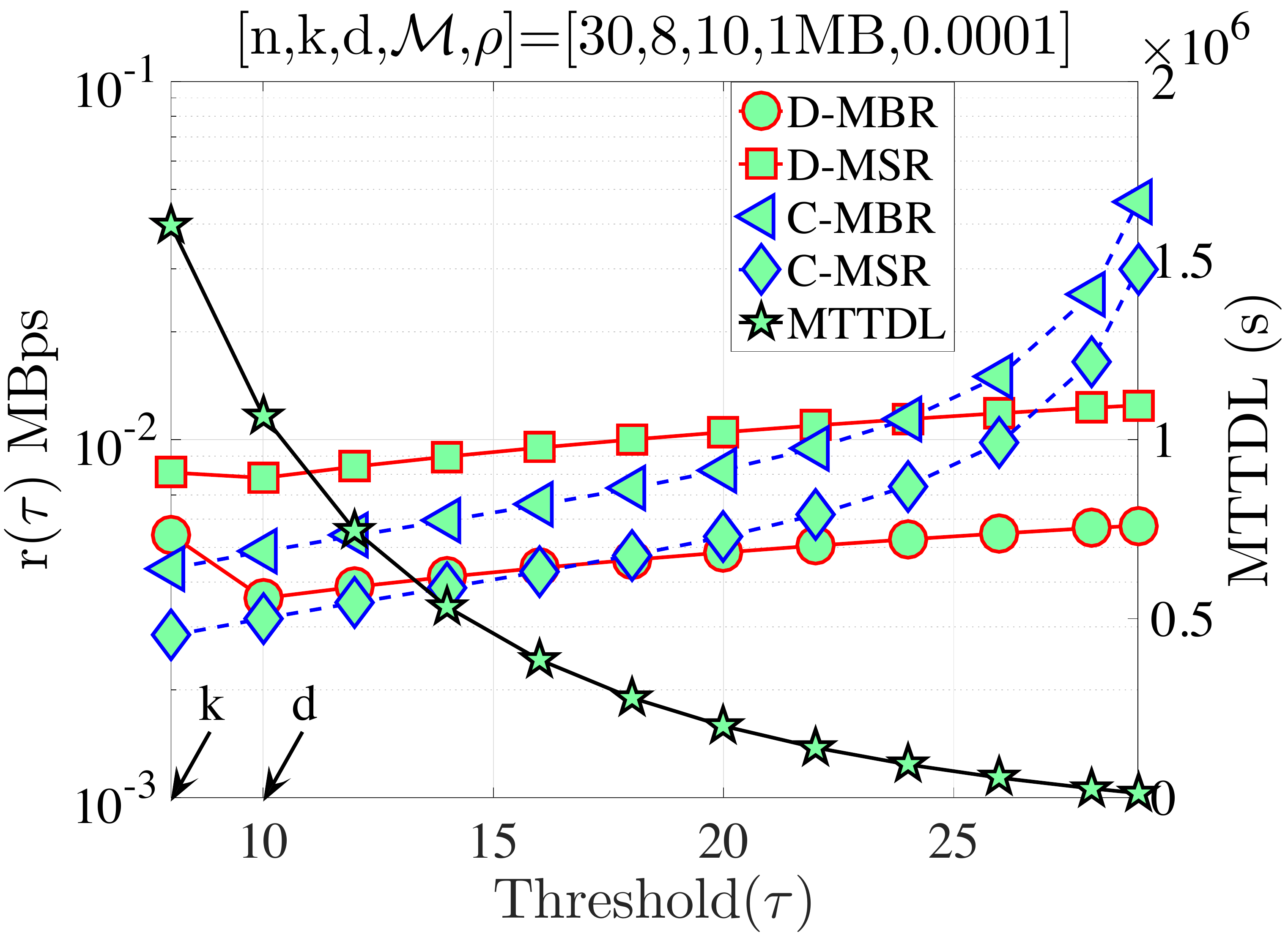} &
			\includegraphics[height=1.6in,width=1.85in]{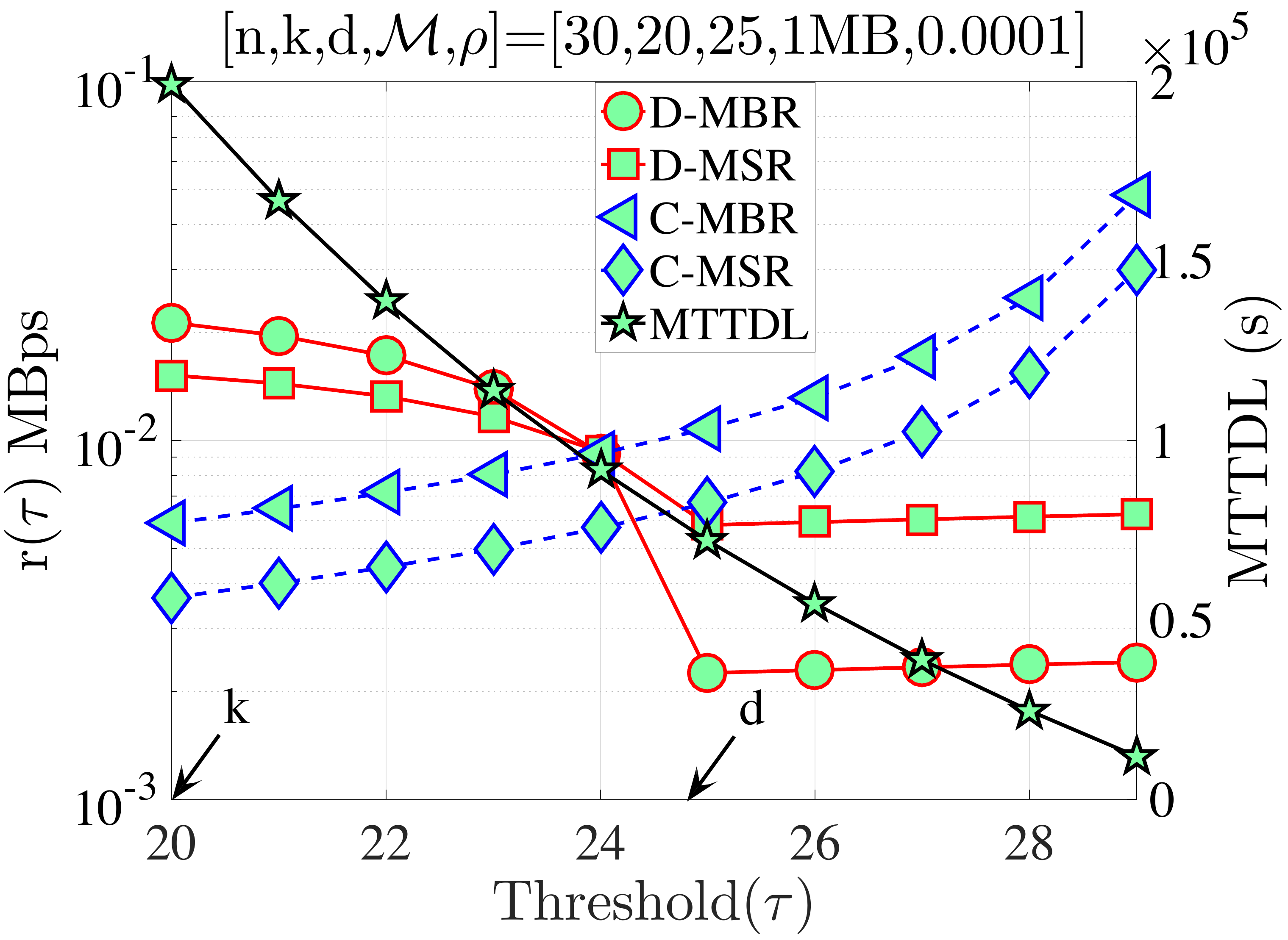} &
			\includegraphics[height=1.6in,width=1.85in]{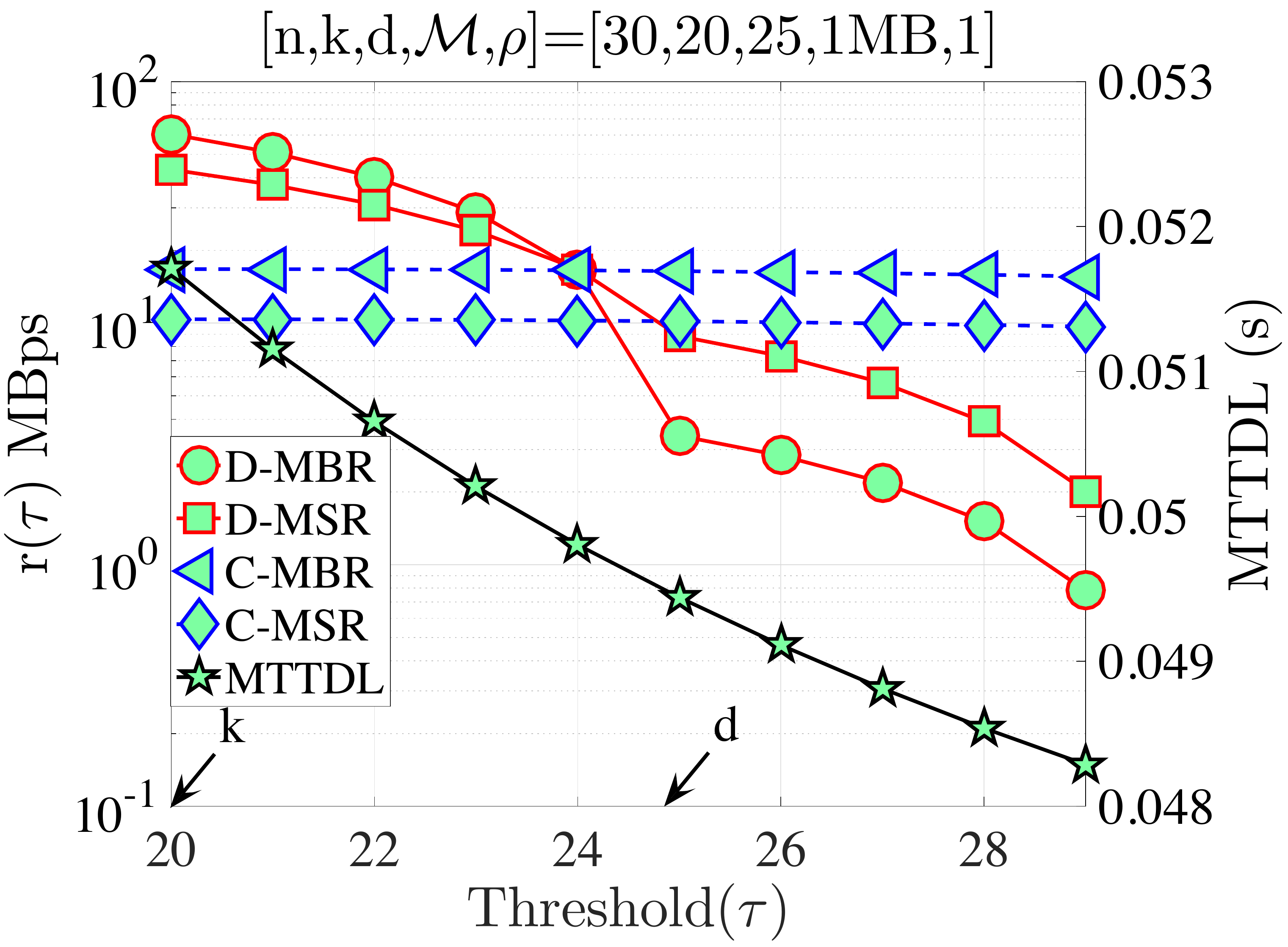}  \\
			(a) &(b) & (c) & (d) 
		\end{tabular}
	\end{center}
	\vspace{-0.1in}
	\caption{Cost $r(\tau)$ vs. repair threshold ($\tau$) for:  (a) $d > \frac{n+k-1}{3}$, (b) $d < \frac{n+k-1}{3}>$, (c) $\rho=10^{-4}$, (d)  $\rho=1$.}
\label{fig:comp_schemes2} 
	\vspace{-0.2in}
\end{figure*}


\subsection{Mean Time to Data Loss for Periodic Repairs}
\label{ssec:mttdl}
We now examine the \emph{Mean Time to Data Loss (MTTDL)} for the periodic threshold repair process. For our purposes, we consider that data is lost if the DSS transitions from state $\tau$ to state $\tau-1$ instead of state $n$. That is, if a node leaves the system before repairs are completed when initiated at state $\tau$, the repair process is abandoned and the system eventually reaches state $k-1$, at which data is lost. In this case, the file $F$ is reinstated at the mobile nodes by a central entity. Note that when $\tau > k$ repairs could be re-initiated at state $\tau-1$, because at least $k$ fragments remain available. We opted not to consider this option for the MTTDL calculation to capture the periodic nature of the threshold repair strategy. The MTTDL reflects the period of time at which the DSS oscillates  between states $n$ and $\tau$. The time to reach state $k-1$ assuming no repairs are attempted after state $\tau$ is given by:

\begin{prop}
	\label{thm:propMTTDL}
For a threshold-based repair strategy attempting regeneration at state $\tau$, the MTTDL is given by

\begin{equation}
\textrm{MTTDL} = \sum_{i=1}^{\infty} \Big( \frac{iH_{n,\tau}}{\lambda} + \frac{i-1}{\mu} +\frac{H_{\tau,k-1}}{\lambda} \Big)(1-p)^{(i-1)}p, 
\end{equation}
where $p=\frac{\tau\lambda}{\tau\lambda+\mu}$.
\end{prop}
\begin{proof}
Proof is provided in Appendix \ref{App:MTTDL}.
\end{proof}

The MTTDL is a decreasing function of $\tau$. This is intuitive considering that the number of nodes that need to depart for reaching state $k-1$ increases with $\tau$. Moreover, the average time it takes to reach state $\tau$ from state $n$ increases with $\tau$. This indicates that the periodic repair of the DSS will on average last longer if  a lazy repair strategy is adopted.

\subsection{Numerical Examples}

In this section, we validate our theoretical results be providing numerical examples.  Fig.~~\ref{fig:comp_schemes2}(a) shows $r(\tau)$ when $d > \frac{n+k-1}{3}$ and $\rho= 10^{-4}$ . According to Proposition~\ref{thm:prop4}, for this combination of code parameters, a distributed repair strategy with MBR codes (D-MBR) achieves the minimum $r(\tau)$ for all $d \leq \tau^* \leq n-1$. The minimum occurs at $\tau^* = d$. 
Moreover, according to Proposition~\ref{thm:prop5}, centralized MSR codes (C-MSR) minimize $r(\tau)$ for   $k \leq \tau <d.$ This is verified in all plots of Fig.~\ref{fig:comp_schemes2}, for which the cost is minimized by the C-MSR strategy when $\tau^*= k$, if $\tau<d$.  In Fig.~\ref{fig:comp_schemes2}(b), we show $r(\tau)$ when $d < \frac{n+k-1}{3}$ and $\rho= 10^{-4}$. For this case, there is no one scheme with optimal cost for any value of $d \leq \tau \leq n-1$. For $\tau>16,$ D-MBR is optimal, whereas for $10 \leq \tau \leq 15$, C-MSR becomes optimal. C-MSR achieves the lowest overall cost at $\tau=k.$

We also studied the impact of $\rho$, when the code parameters are fixed to $(n=30,k=20,d=25)$. Fig.~\ref{fig:comp_schemes2}(c) shows the average cost per unit of time ($r(\tau)$) when $\rho=10^{-4}.$ For this $\rho$ regime, a lazy repair strategy with $\tau^{\ast} = d$ minimizes $r(\tau)$, with D-MBR codes achieving the lowest cost. On the other hand, eager repair becomes optimal for any $\rho> \rho_{\textrm{high}}$. This is observed in Fig.~\ref{fig:comp_schemes2}(d), in which the value of $\rho$ has been increased to one. D-MBR codes still remain the optimal option, however, the optimal repair threshold is now shifted to  $\tau^* = n-1$. Note that at the high $\rho$ regime, all codes exhibit the same behavior. The average cost per unit of time becomes a decreasing function of $\tau$.

Finally, on the right $y$-axis of the plots in Fig.~\ref{fig:comp_schemes2}, we show the MTTDL values for the given set of parameters. As expected, the MTTDL is an decreasing function of $\tau$ due to the corresponding increase in departure rate from state $\tau$ with the value of $\tau$. The MTTDL becomes impractical in the high $\rho$ regime, because nodes frequently leave area $\mathcal{A}$ before repairs can be completed.

\section{Codes with Cooperative Repair}
\label{sec:coopRep}
In the case of multiple node failures, regenerating the failed nodes individually is not optimal in terms of repair bandwidth. To regenerate multiple failed nodes more efficiently, the newcomers can also communicate with each other to lower the repair bandwidth, which is called \emph{cooperative repair}. Specifically, the newcomer nodes communicate to not only the existing live nodes but also each of the other newcomers for regeneration. In this section, we analyze examples of such codes and their performance for the Markov-model that is described earlier. 

\subsection{Cooperative Regenerating Codes}
\label{subsec:coopReg}
When multiple nodes are to be repaired simultaneously, in addition to contacting live nodes and downloading symbols from those, newcomers can also communicate between each other to complete the recovery process. Formally, assume that $t$ nodes are to be repaired. Each of the $t$ newcomer nodes can contact $d$ live nodes and download $\beta$ symbols as well as download $\beta'$ from each other. In this scenario, the repair bandwidth can be calculated as $\gamma=d\beta + (t-1)\beta'$. Such codes are studied in \cite{kermarrec2011repairing} (referred to also as coordinated regenerating codes) and the tradeoff between per node storage $\alpha$ and repair bandwidth $\gamma$ is analyzed. Two ends of tradeoff curve is named Minimum Storage Cooperative Regenerating (MSCR) and Minimum Bandwidth Cooperative Regenerating (MBCR). Accordingly, operating points are given as follows:
\begin{equation}
\left(\alpha_{\textrm{MSCR}}, \beta_{\textrm{MSCR}},  \beta_{\textrm{MSCR}}' \right)=\left(\frac{\mathcal{M}}{k},\frac{\mathcal{M}}{k\left(d-k+t\right)},\frac{\mathcal{M}}{k\left(d-k+t\right)}\right).
\label{mscr}
\end{equation}
For MSCR codes, $\gamma_{\textrm{MSCR}}=d\beta_{\textrm{MSCR}} + (t-1)\beta_{\textrm{MSCR}}'=\frac{\mathcal{M}(d+t-1)}{k\left(d-k+t\right)} \geq \frac{\mathcal{M}}{k}$, where the per-node storage is smaller than the repair bandwidth. On the other hand, MBCR codes provide the minimum repair bandwidth which operates at: 
\begin{equation}
\left(\alpha_{\textrm{MBCR}}, \beta_{\textrm{MBCR}},  \beta_{\textrm{MBCR}}' \right)=\left(\frac{(2d+t-1)\mathcal{M}}{k\left(2d-k+t\right)},\frac{2\mathcal{M}}{k\left(2d-k+t\right)},\frac{\mathcal{M}}{k\left(2d-k+t\right)} \right).
\label{mbcr}
\end{equation}
Note that for MBCR codes, we have $\alpha_{\textrm{MBCR}} = \gamma_{\textrm{MBCR}}=d\beta_{\textrm{MBCR}} + (t-1)\beta_{\textrm{MBCR}}'$.

We define $r_t(\tau)$ as the average repair cost for a system with repair threshold $\tau$ under cooperative repair using groups of nodes of size $t$ (similarly $c_t(\tau)$ for the cost). Since $n-\tau$ nodes need to be repaired, any cooperative regenerating codes with $t$ such that $t|n-\tau$ can be used in practice. In the following proposition, we compare the performance of cooperative regenerating codes at all possible $t$ values to find the value of $t$ that minimizes the average repair cost.

\begin{prop}
	The average repair cost of cooperative repair is a monotonically decreasing function of the cooperation group size $t$. That is, for two cooperation groups $t_1$ and $t_2$, with $t_1|n-\tau$ and $t_2|n-\tau$ and $t_1 < t_2$, it follows that $r_{t_1}(\tau) > r_{t_2}(\tau)$.
\end{prop} 
\begin{proof}
	Since in both scenarios, repairs are performed after $n-\tau$ node departures and node repairs are performed parallel, it's enough to compare only the required bandwidths. First, assume MSCR case, then we want to show that
	\begin{equation}
	\frac{\mathcal{M}(d+t_1-1)(n-\tau)}{k(d-k+t_1)} > \frac{\mathcal{M}(d+t_2-1)(n-\tau)}{k(d-k+t_2)},
	\end{equation} 
	which is equivalent to
	\begin{equation}
	(t_2-t_1)(k-1)>0.
	\label{eq:prop_t1t2}
	\end{equation}
	Since $t_2 > t_1$ and we can conclude that $c_{t_1}(\tau) > c_{t_2}(\tau)$ therefore $r_{t_1}(\tau) > r_{t_2}(\tau)$. Similarly, for MBCR case, we need to have 
	\begin{equation}
	\frac{\mathcal{M}(2d+t_1-1)(n-\tau)}{k(2d-k+t_1)} > \frac{\mathcal{M}(2d+t_2-1)(n-\tau)}{k(2d-k+t_2)},
	\end{equation}
	from which we can obtain the same condition as \eqref{eq:prop_t1t2}. Henceforth, $c_{t_1}(\tau) > c_{t_2}(\tau)$ and $r_{t_1}(\tau) > r_{t_2}(\tau)$. Combining MSCR and MBCR cases, we can conclude that $r_{t_1}(\tau) > r_{t_2}(\tau)$. 
\end{proof}

\begin{remark}
	As a result of the above proposition, one can minimize the average repair cost by performing cooperative repairs with $t = n-\tau$. In other words, for the $n-\tau$ nodes that are to be repaired, the optimal cooperative regenerating code is with $n-\tau$, all nodes should cooperate at the same time. 
\end{remark}

In Fig.~\ref{fig:comp-coops}, we show the average repair cost for cooperative regenerating codes at all possible $t$ values for a given $n-\tau$. We observe that for the same $n-\tau$, if $t_1 < t_2$, then $r_{t_1}(\tau)>r_{t_2}(\tau)$ and the minimum is achieved when $t=n-\tau$. 

In the remaining of this section, we suppress the subindex $t$ from $r_t(\tau)$ since we established that $t=n-\tau$ minimizes the cost, i.e., $r(\tau)=r_{n-\tau}(\tau)$ and $c(\tau)=c_{n-\tau}(\tau)$. However, we may still need to distinguish $\gamma$ and $\alpha$ values under different cooperative regenerating codes. We denote the per-node storage for cooperative repairs with $n-\tau$, which results in the minimum average repair cost for the system with threshold $\tau$, by $\alpha_\tau$. Similarly we use $\gamma_\tau$ to denote the repair bandwidth at $\tau$.

\begin{figure*}[tb]
	\begin{center}
		\setlength{\tabcolsep}{-0.01in}
		\begin{tabular}{cc}
			\includegraphics[width=0.5\textwidth]{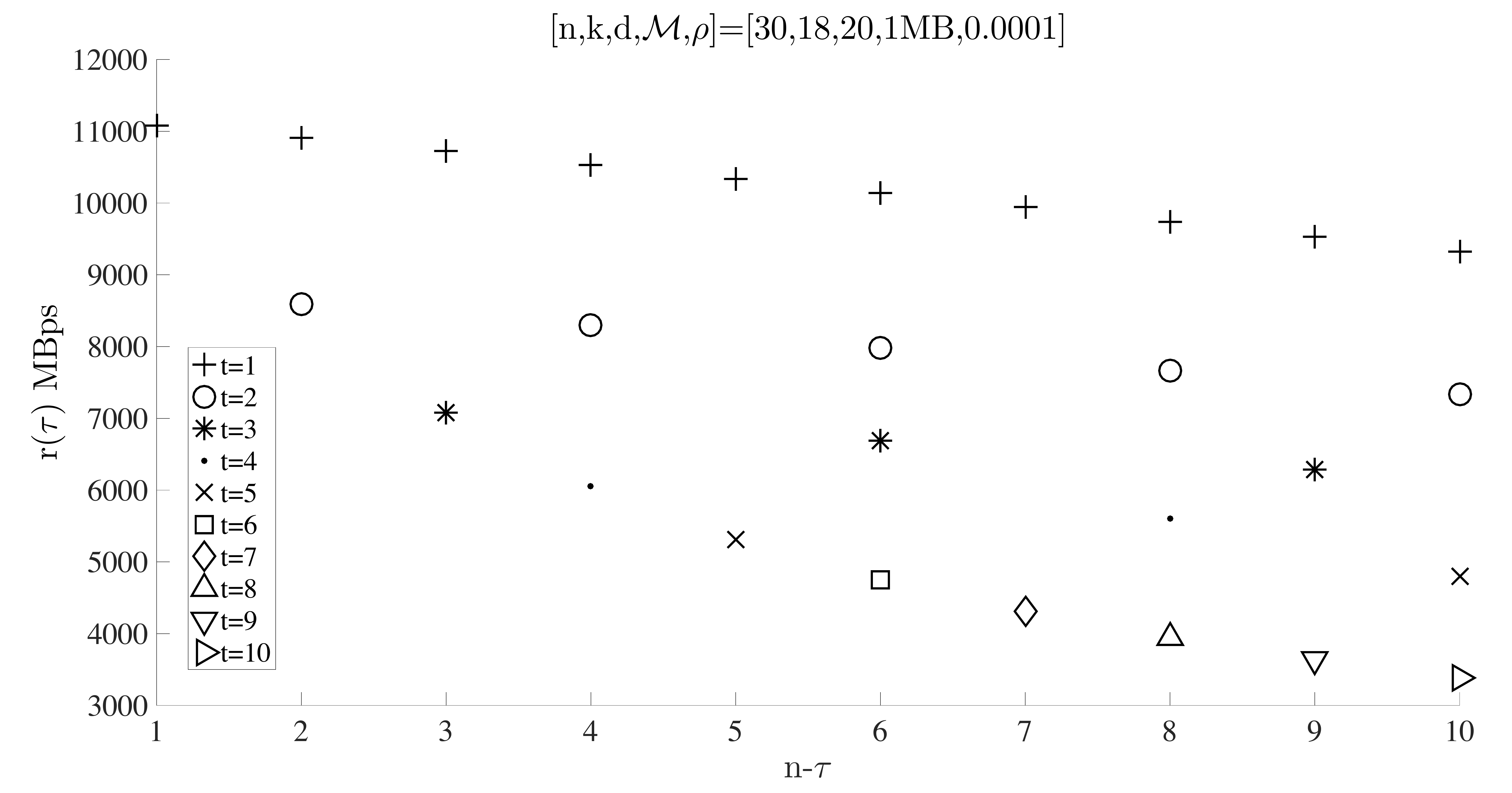} &
			\includegraphics[width=0.5\textwidth]{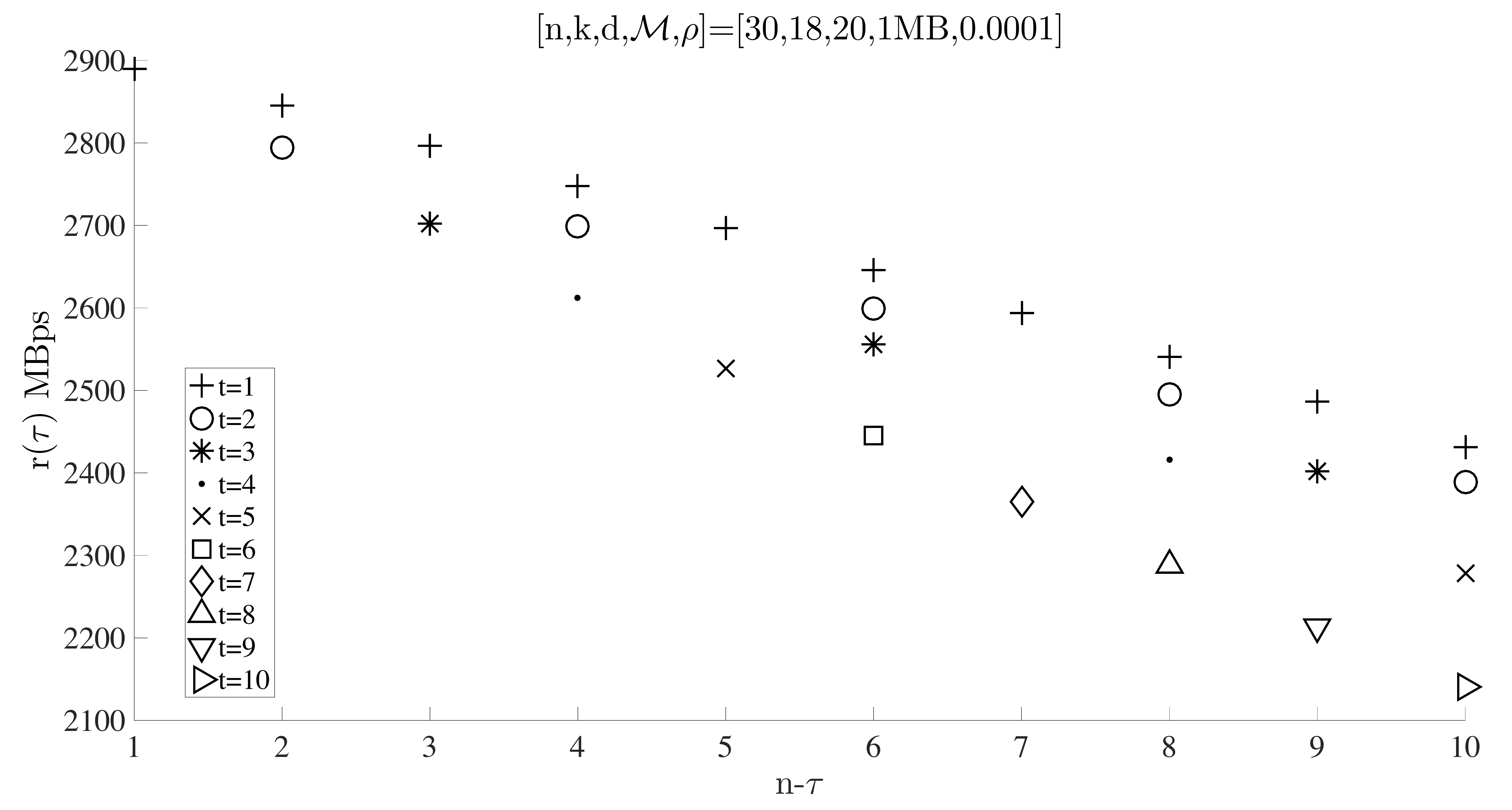} \\
			(a) & (b)\\
		\end{tabular}
	\end{center}
	\vspace{-0.1in}
	\caption{Cost $r(\tau)$ vs. number of nodes to be repaired ($n-\tau$) for:  (a) MSCR, (b) MBCR.}
	\label{fig:comp-coops} 

\end{figure*}

Note that, cooperative regenerating codes, it is required that $d+n-\tau \leq n$. In other words, there should be at least $d$ live nodes when repairs are initiated. Otherwise, it is not possible to regenerate $n-\tau$ nodes from d live nodes. Henceforth, the repair cost $c(\tau)$ and the repair cost per time for cooperative codes are as follows:

\begin{equation}
	c(\tau) = \gamma_{\tau}(n-\tau), \quad r(\tau) = \frac{c(\tau)}{\textrm{E}[\Delta]} = \frac{\lambda \mu(\gamma_{\tau} (n-\tau))}{\mu H_{n,\tau} + \lambda}, \quad \mbox{  if  } \tau \geq d.
\end{equation}

The minimum $r(\tau)$ with respect to $\tau$ does not have a closed-form analytical expression. In Section~\ref{subsec:NumRes}, we present numerical results to study the change of $r(\tau)$ with $\tau$ and determine the optimal repair threshold $\tau$ that minimizes $r(\tau)$ for distributed cooperative repair.

\subsection{Centralized Repair of Multiple Node Failures}

Centralized repair of multiple node repairs is introduced in \cite{rawal2016centralized}. Using this model, a dedicated node among the $t$ newcomers downloads $\beta$ from any $d$ live nodes such that it can repair multiple node failures of size $t$. Such codes can be used in the centralized repair process that is proposed in Section~\ref{sec:repair_strategies} when we set $t = n-\tau$. Rawat \etal in \cite{rawal2016centralized} characterize the tradeoff between per-node storage and repair bandwidth for this centralized repair model. Accordingly, the following operation points are derived for minimum storage multi-node regeneration (MSMR) and minimum bandwidth multi-node regeneration (MBMR):

\begin{equation}
\left(\alpha_{\textrm{MSMR}}, \gamma_{\textrm{MSMR}} \right)=\left(\frac{\mathcal{M}}{k},\frac{\mathcal{M}d(n-\tau)}{k\left(d-k+n-\tau\right)}\right).
\end{equation}

Let $k \textrm{ mod} (n-\tau)=b$. If $H_{b} \geq \binom{\beta}{n-\tau}\left[ b(\frac{2d+n-\tau-1}{2})-\binom{b}{2}\right]$ (where $H_b$ denotes entropy of information stored on $b$ nodes), then  
\begin{equation}
\left(\alpha_{\textrm{MBMR}}, \gamma_{\textrm{MBMR}} \right)=\left(\frac{\mathcal{M}2d}{k\left(2d-k+n-\tau\right)},\frac{\mathcal{M}2d(n-\tau)}{k\left(2d-k+n-\tau\right)}\right).
\end{equation}

Under centralized repair model discussed here, a dedicated node first downloads $\gamma=d\beta$ and then distributes $\alpha$ to remaining $n-\tau-1$. The difference between these codes and the earlier proposed method for centralized repair in Section~\ref{subsec:centralized} is that the dedicated node may not need to download the whole file. Therefore, we have the following repair cost

\begin{equation}
c_C(\tau) = \gamma + \alpha(n-\tau-1),
\end{equation}
from which one can obtain

\begin{equation}
r_C(\tau) = \frac{c_C(\tau)}{\textrm{E}[\Delta]}=\frac{\lambda\mu(\gamma + \alpha(n-\tau-1))}{\mu H_{n,\tau} + \lambda}.
\label{eq:r_tau_d_mr}
\end{equation}

In order to find the optimal threshold that minimizes the average repair cost, we need to find the minimum value of $r_C(\tau)$. We can replace $H_{n,\tau}$ with its approximation, $\ln (\frac{n}{\tau})$, and take the derivative with respect to $\tau$. Note that both $\gamma$ and $\alpha$ depend on $\tau$ and there is no tractable analytical solution for $\tau$ that minimizes $r_C(\tau)$. However, we can still analyze \eqref{eq:r_tau_d_mr} numerically with respect to $\tau$ and observe the optimal threshold from numerical results.

\begin{remark}
	In this section, we analyzed different cooperative codes that are suitable for mobile clouds. For both scenarios, the problem of finding $\tau$ for which $r(\tau)$ is minimized does not have a closed-form solution due to the dependence of $\alpha$ and $\gamma$ on $\tau$. We, therefore, resort to the numerical analysis of the optimal threshold.  Next, we present numerical results instead, which are presented below.
	Note that, for the regenerating codes that were analyzed in Section~\ref{sec:numerical}, $\alpha$ and $\gamma$ do not depend on the threshold $\tau$.
\end{remark}

\subsection{Numerical Results}
\label{subsec:NumRes}
We study the performance of cooperative codes with $n=30, d=25$ and $k=19$ under different $\rho$ regimes. 
In Fig.~\ref{fig:comp_schemes3}, we compare the codes studied in Sections~\ref{sec:repair_strategies} and \ref{sec:coopRep} for different $\rho$ regimes. We first compare regenerating codes vs. cooperative regenerating codes for the distributed repair scenario. Note that we focus only on $d \leq \tau \leq n-1$ because at least $d$ live nodes must exist for cooperation (see Section~\ref{subsec:coopReg}). We observe that cooperative regenerating codes always have lower cost than regenerating codes for all values of $\rho$. Additionally, the gap between the cost of D-MBR and D-MBCR is much smaller than the gap between the cost of D-MSR and D-MSCR. Furthermore, we can observe two opposing regimes: In Fig.~\ref{fig:comp_schemes3}(a), the optimal cost is at $\tau = d$, whereas in Fig.~\ref{fig:comp_schemes3}(b), the cost is minimized at $\tau = n-1$. We also compare the centralized regenerating codes to the centralized repair of multiple node departures. As expected, since in the latter scheme one does not need the whole file for file reconstruction at the dedicated node, centralized repair of multiple node departures results in lower repair cost. At $\tau = n-1$, in Fig.~\ref{fig:comp_schemes3}(d) average repair cost is minimized for all schemes, on the other hand we observe different optimum $\tau$ values for different coding schemes in Fig.~\ref{fig:comp_schemes3}(c). Finally, we compare all schemes in Fig.~\ref{fig:comp_schemes3}(e)-(f) for different values of $\rho$ for completeness. It is observed that the centralized repair of multiple node departures model discussed in this section approaches to the distributed repair model in Section~\ref{sec:repair_strategies} as $\tau$ approaches $n$ and diverges from the centralized repair model in Section~\ref{sec:repair_strategies}. The reason for this behavior is that the dedicated node does not need to download the whole file now (as opposed to the centralized repair model in Section~\ref{sec:repair_strategies}, which incurs high average repair cost for large $\tau$).

	
\begin{figure*}
	\begin{center}
		\setlength{\tabcolsep}{-0.01in}
		\begin{tabular}{cc}
			\includegraphics[height=2.4in,width=2.8in]{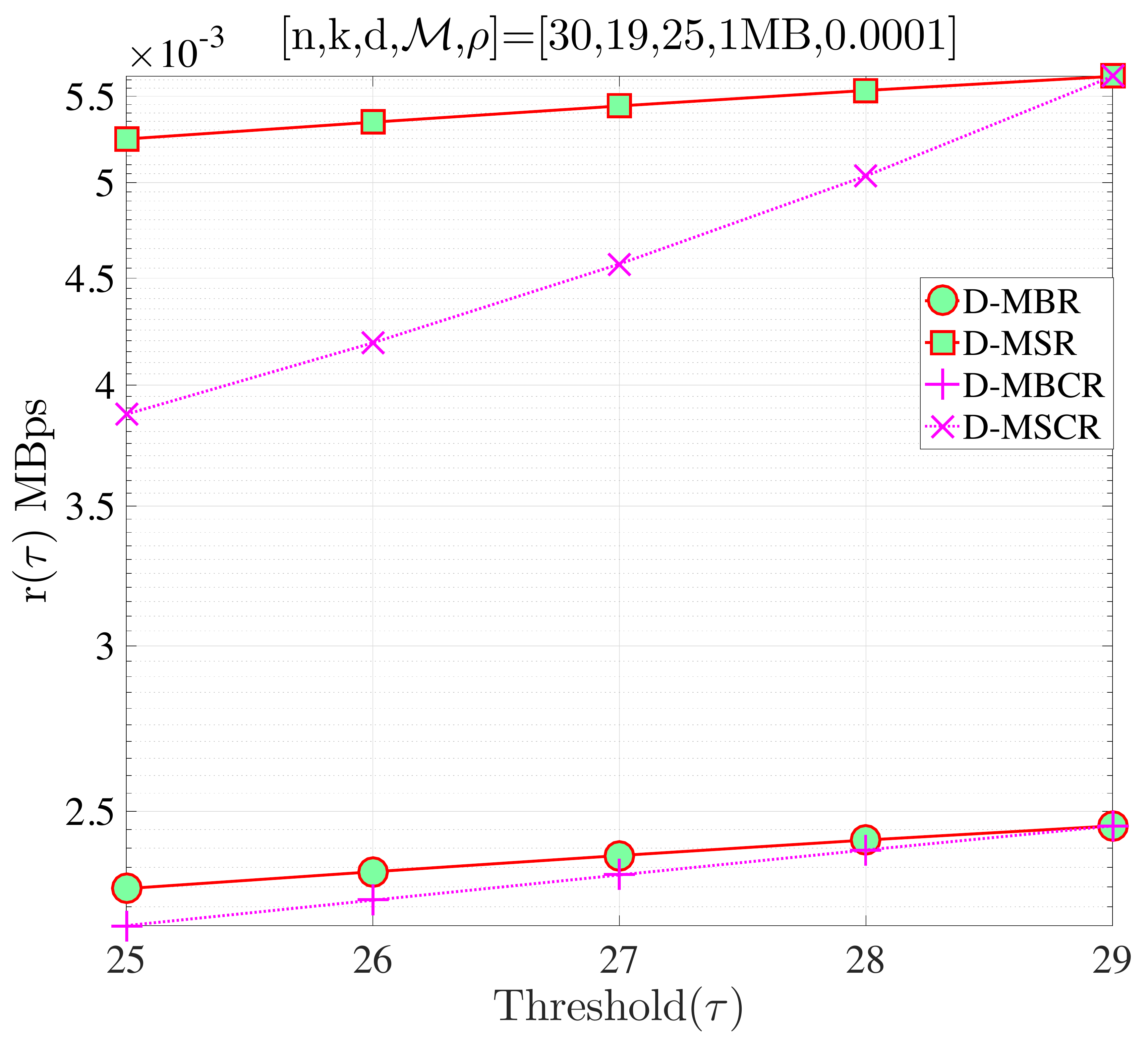} &
			\includegraphics[height=2.4in,width=2.8in]{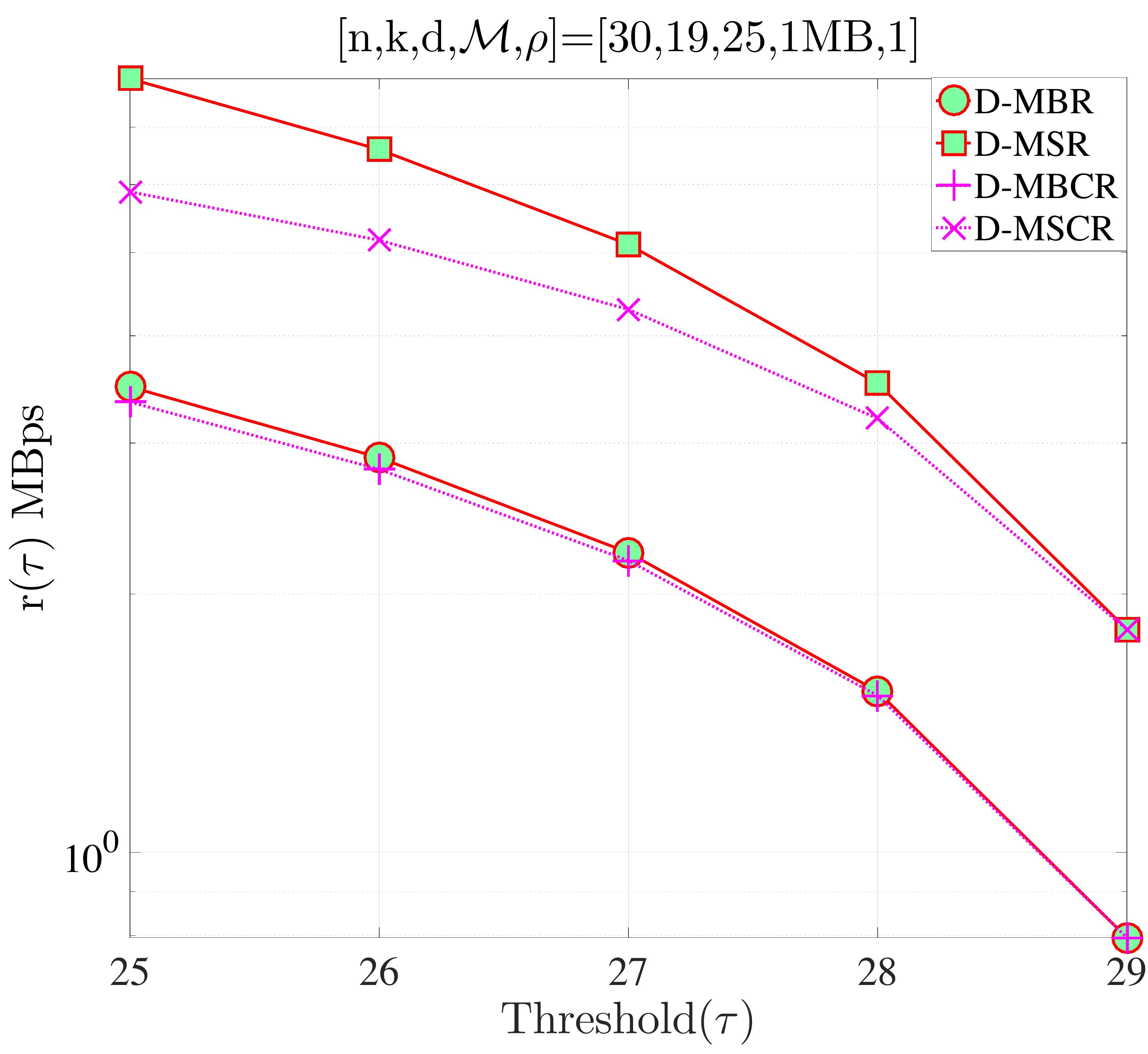}  \\
			(a) & (b) \\
			\includegraphics[height=2.4in,width=2.8in]{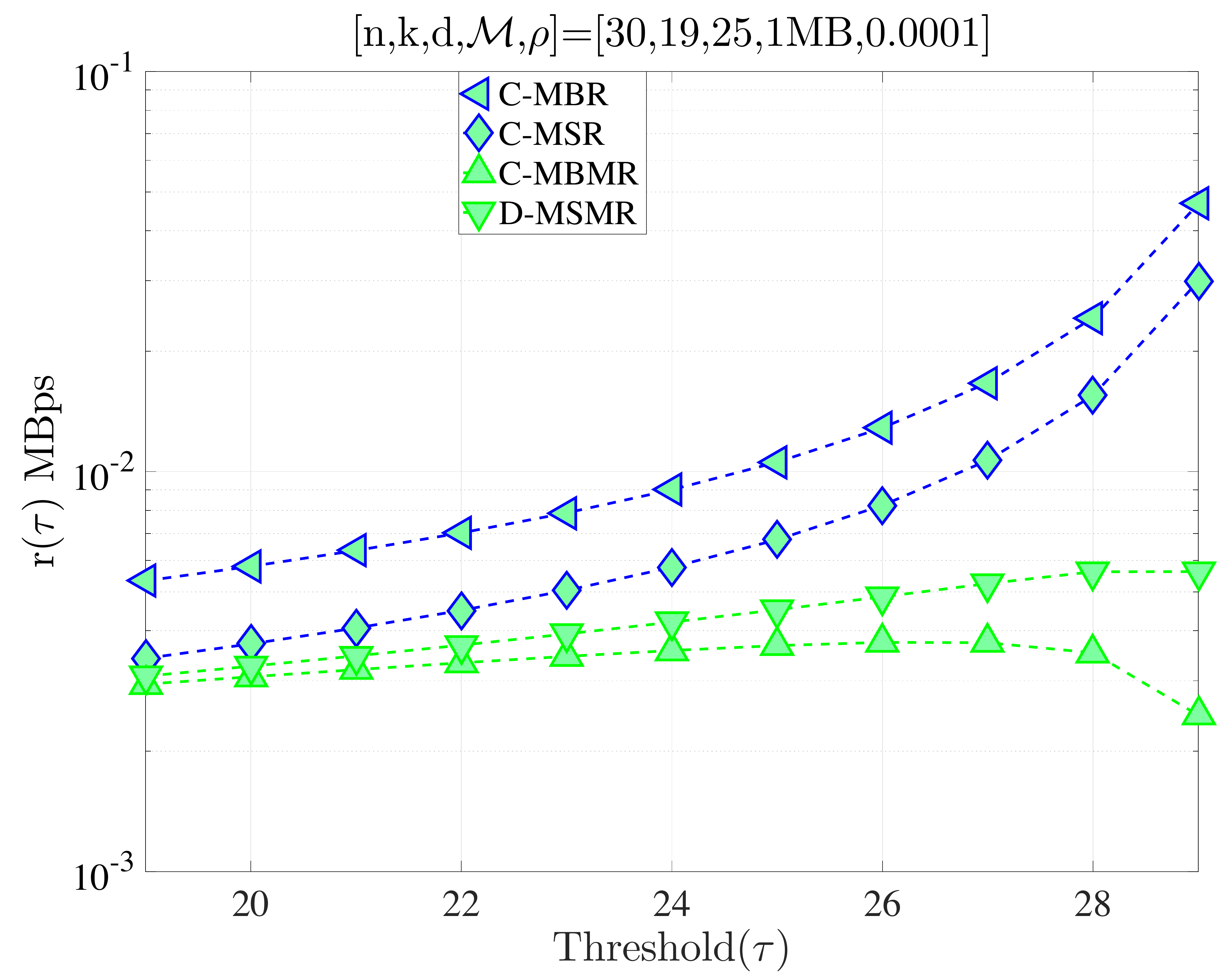} &
			\includegraphics[height=2.4in,width=2.8in]{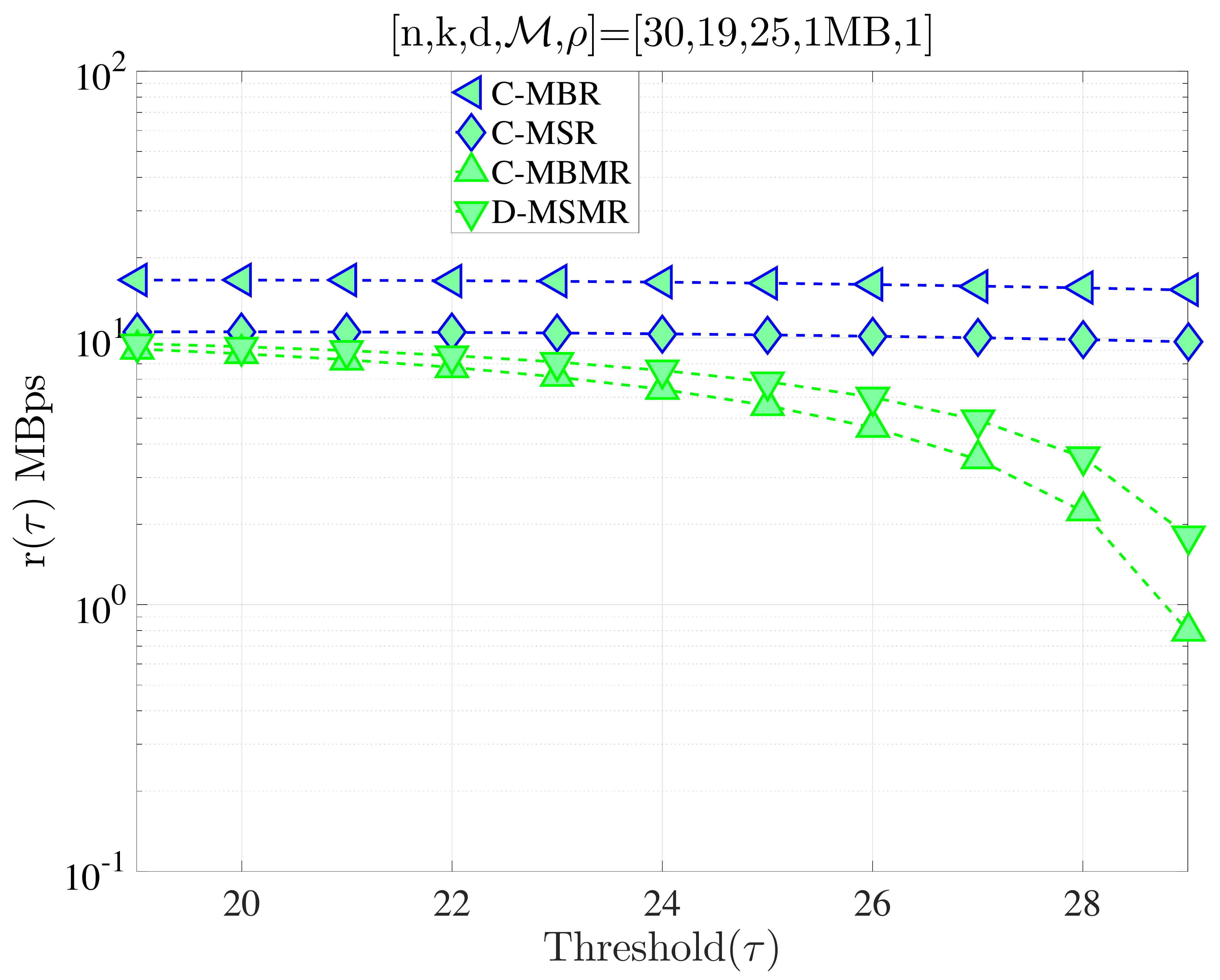} \\
			(c) & (d) \\
			\includegraphics[height=2.4in,width=2.9in]{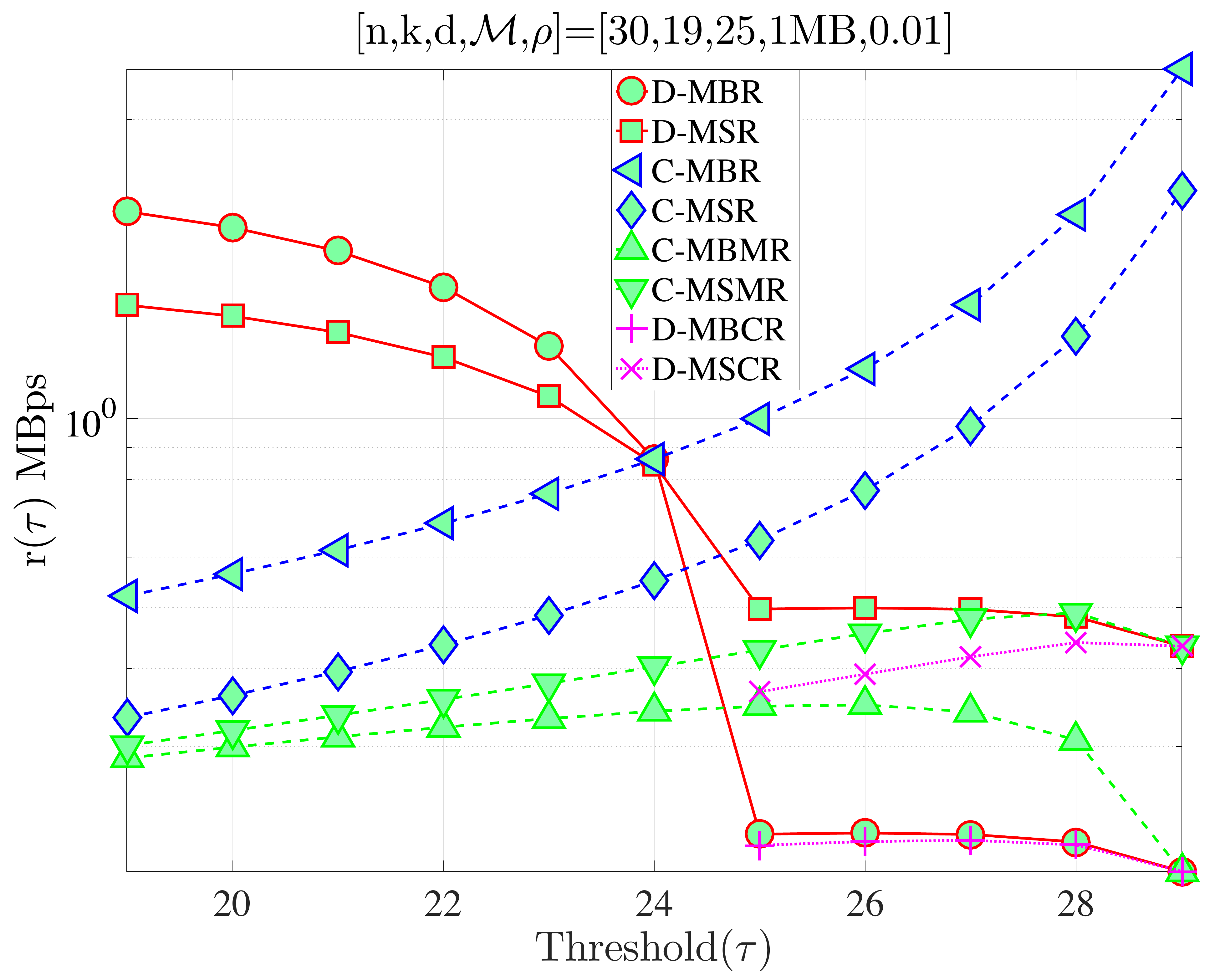} &
			\includegraphics[height=2.4in,width=2.9in]{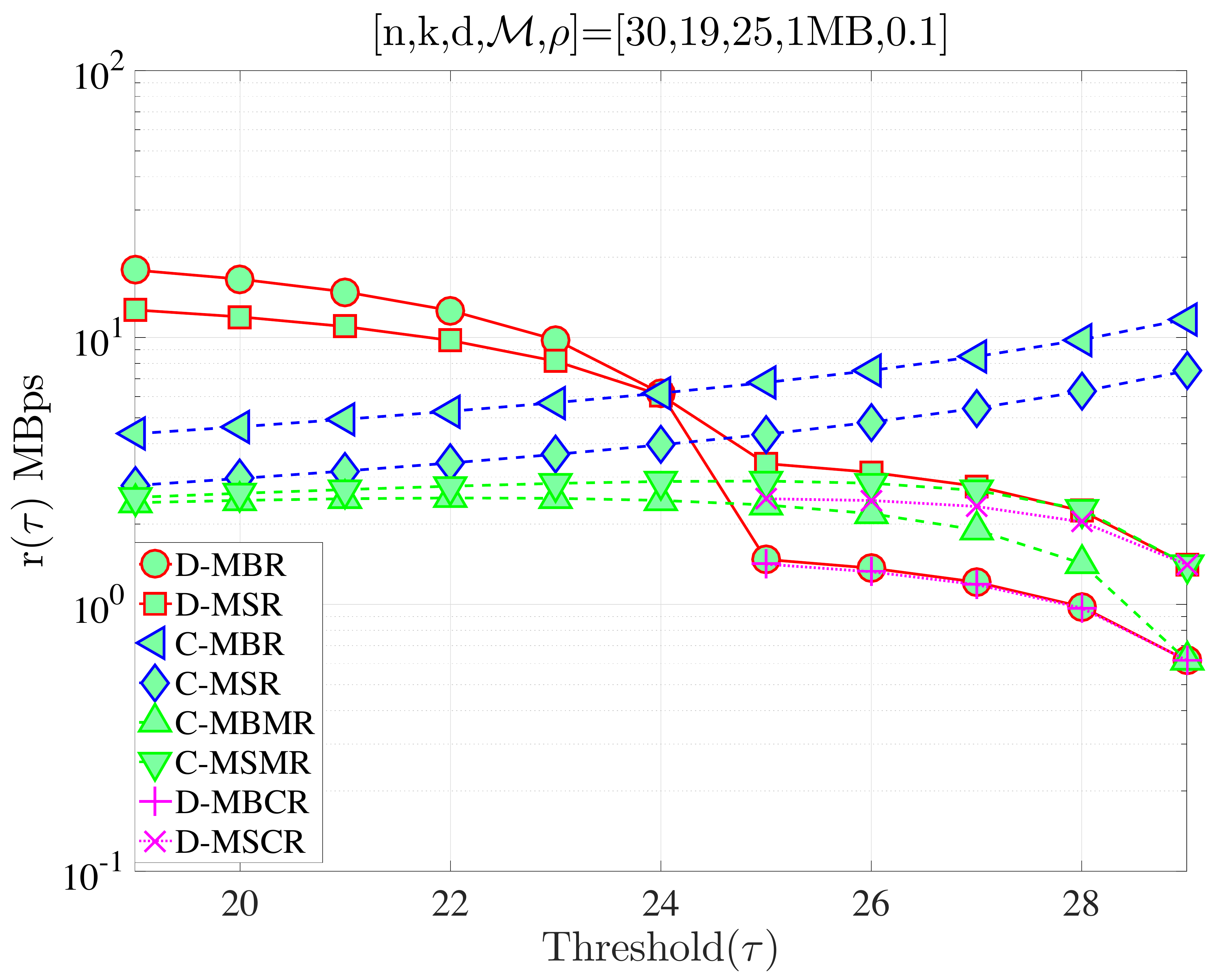} \\
			(e) & (f) 
		\end{tabular}
	\end{center}
	\vspace{-0.1in}
	\caption{Cost $r(\tau)$ vs. repair threshold ($\tau$) for: (a) distributed regenerating codes vs. cooperative regenerating codes when $\rho = 0.0001$, (b) distributed regenerating codes vs. cooperative regenerating codes when $\rho = 1$ (c) centralized regenerating codes vs. centralized repair of multiple node failures when $\rho = 0.0001$, (d) centralized regenerating codes vs. centralized repair of multiple node failures when $\rho = 1$, (e) all schemes when $\rho = 0.01$, (f)  all schemes when $\rho = 0.1$.}
\label{fig:comp_schemes3} 
	\vspace{-0.2in}
\end{figure*}

\section{A Repair Process Analogous to the number of repaired nodes}
\label{sec:dependent_tau}

In the analysis presented in Sections \ref{sec:numerical}-\ref{sec:coopRep}, we have assumed that the repair time is exponentially distributed with parameter $\mu$, irrespective of the number of nodes to be repaired. This model is mathematically tractable and we are able to find analytical results on the optimal repair threshold. In this section, we consider a revised model in which the repair time is analogous to the number of nodes that need to be repaired. Specifically, we model the repair process as the maximum of $n-\tau$ exponential random variables, each with rate $\mu$. In other words, when the repair process is initiated, one can consider starting $n-\tau$ exponential clocks, each with rate $\mu$. The repair process ends when all the clocks end. We note that the maximum value of such clocks is not exponentially distributed (as opposed to the minimum of such clocks), however, its expected value is known, which is enough for the purpose of finding average repair cost. Let $T_i^r$ denote the repair time of $i^{\textrm{th}}$ newcomer node, then we have the following \cite{Eisenberg:On08}

\begin{equation}
\textrm{E}[\max{(T_1^r, \dots, T_{n-\tau}^r)}] = \sum_{i=1}^{n-\tau} \frac{1}{i\mu} = \frac{H_{n-\tau,0}}{\mu}.
\end{equation}
The expected time between two instances of fully operational system with $n$ live nodes is given by

\begin{equation}
\textrm{E}[\Delta] = \frac{H_{n,\tau}}{\lambda} +\frac{H_{n-\tau,0}}{\mu},
\label{e[delta]_2}
\end{equation}
where the first term is the expected time from state $n$ to state $\tau$, and the second term is the expected time  from state $\tau$ back to state $n$. Accordingly, the resulting average repair cost for the distributed repair case is
\begin{equation}
r_D(\tau) = \frac{c_D(\tau)}{\textrm{E}[\Delta]} =
\begin{cases}
\frac{\lambda \mu(k\alpha(d-\tau)+\gamma (n-d))}{\mu H_{n,\tau} + \lambda H_{n-\tau,0}}, & \textrm{if } \tau < d \\
\frac{\lambda \mu(\gamma (n-\tau))}{\mu H_{n,\tau} + \lambda H_{n-\tau,0}}, & \textrm{if } \tau \geq d. \\
\end{cases}
\label{rtau_d_2}
\end{equation}
Note that for centralized repair, a dedicated newcomer node first downloads the file, and then distributes symbols to the remaining $n-\tau-1$ newcomer nodes. Therefore, the expected repair time is given by $\frac{1}{\mu}+\frac{H_{n-\tau-1}}{\mu}$, where we have first one clock with exponential rate of $\mu$, followed by a maximum of $n-\tau-1$ clocks with rate $\mu$. Accordingly, we obtain 
\begin{equation}
r_C(\tau) = \frac{c_C(\tau)}{\textrm{E}[\Delta]} = \frac{\lambda \mu \alpha (k+n-\tau-1)}{ \mu H_{n,\tau} + \lambda (1 + H_{n-\tau-1,0})}. 
\end{equation}
The optimal $\tau$ which minimizes the above equation is difficult to track due to the complexity of the formula (due to $\lambda H$ in the denominator). Instead, we perform numerical analysis of the average repair cost with respect to the threshold later in this section. 

In the context of this model, we next focus on MTTDL. Note that, if we start $(n-\tau)$ clocks, the probability that no data loss occurs within a cycle, denoted by $1-p$, can be calculated as $1-p=\Pr(T_{\tau-1}>T_i^r, \forall i \in 1,\dots,n-\tau)$. In other words, the exponential random variable with rate $\tau \lambda$ should be greater than all $n-\tau$ exponential random variables with rate $\mu$. Since we have i.i.d. exponential random variables, $1-p=(\Pr(T_{\tau-1}>T_1^r))^{(n-\tau)}=(\frac{\mu}{\tau\lambda+\mu})^{(n-\tau)}$. Accordingly, we have 

\begin{equation}
\textrm{MTTDL} = \sum_{i=1}^{\infty} \Big( \frac{iH_{n,\tau}}{\lambda} + \frac{(i-1)H_{n-\tau,0}}{\mu} +\frac{H_{\tau,k-1}}{\lambda} \Big)(1-p)^{(i-1)}p, 
\label{eq:MTTDL_clock}
\end{equation}
where $p=1-(\frac{\mu}{\tau\lambda+\mu})^{(n-\tau)}$. Note that the above equation is for the calculation of MTTDL for distributed repair. For centralized repair, the random variable with rate $\tau \lambda$ should be larger than sum of two exponential random variables with rate $\mu$, which is a gamma distribution since first a dedicated node downloads the whole file and then it repairs the other nodes. Henceforth, we did not perform MTTDL analysis for centralized repair. \footnote{The resulting equation is not in compact form, thus we omit this analysis in this text.}

In Fig.~\ref{fig:max_clock-exp}, we compare how the revised model affects the average repair cost relative to the simplified repair model used in Section~\ref{sec:repair_strategies}, in which all nodes are repaired under the same clock. We denote the values calculated within this section by appending $\tau$ at the end, i.e., $\textrm{D-MBR-}\tau$, to specify that the repair process that takes into account the number of nodes to be repaired (i.e., $n-\tau$). It can be observed that changing the model does not affect the behavior of the $r(\tau)$ curves substantially. We observe in the modified model that average cost is decreased slightly. Even though we have the same costs in both cases, the expected times to complete the repair process are different. Specifically, in Sections~\ref{sec:repair_strategies} and \ref{sec:numerical}, it takes $\frac{1}{\mu}$ time to finish the repair process. On the other hand, in the modified model, we change this value to maximum of $n-\tau$ exponential random variables, each with mean $\frac{1}{\mu}$ and this maximum value is larger than $\frac{1}{\mu}$ (unless $\tau = n-1 $). Therefore, it takes longer to complete repairs in the modified model, which results in smaller values of average repair cost. On the other hand, in  Fig.~\ref{fig:mttdl}, a different behavior is observed for the MTTDL. In the low $\rho$ regime, the MTTDL decreases with $\tau$ for both single clock and the maximum of multiple clock models. However, in the high $\rho$ regime, the increase in $\tau$ decreases MTTDL for the single clock model, whereas the MTTDL for the multiple clocks model increases with $\tau$. This is because in \eqref{eq:MTTDL_clock}, $p$ converges to 1 as we decrease $\tau$ in the high $\rho$ regime, which reduces \eqref{eq:MTTDL_clock}. On the other hand, in the low $\rho$ regime, $p$ converges to zero as we increase $\tau$, which reduces \eqref{eq:MTTDL_clock}. Finally, the model discussed in this section results in lower MTTDL values compared to the previous one.

\begin{figure*}[tb]
	\begin{center}
		\setlength{\tabcolsep}{-0.01in}
		\begin{tabular}{cc}
			\includegraphics[height=2.8in,width=3.5in]{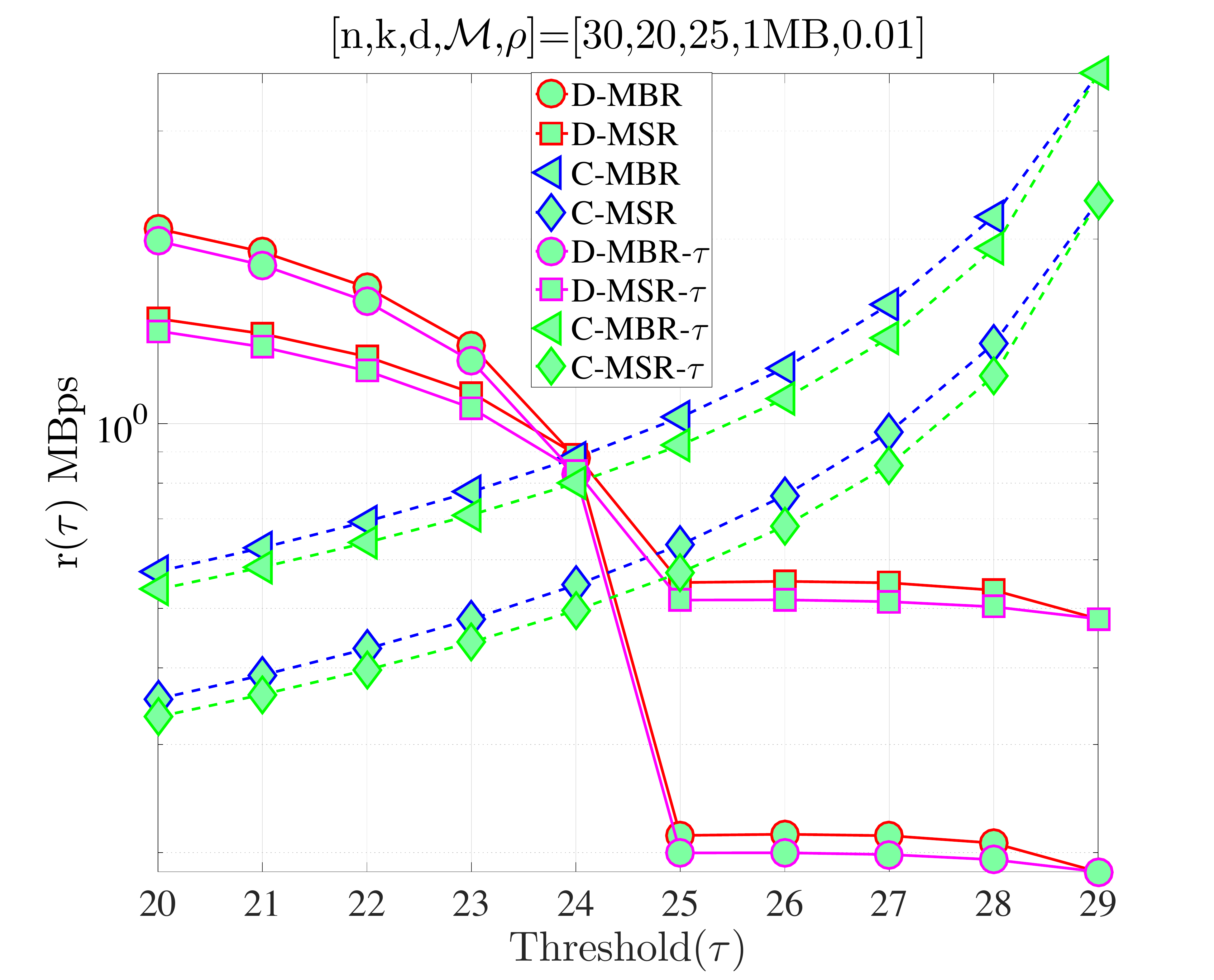} &
			\includegraphics[height=2.8in,width=3.5in]{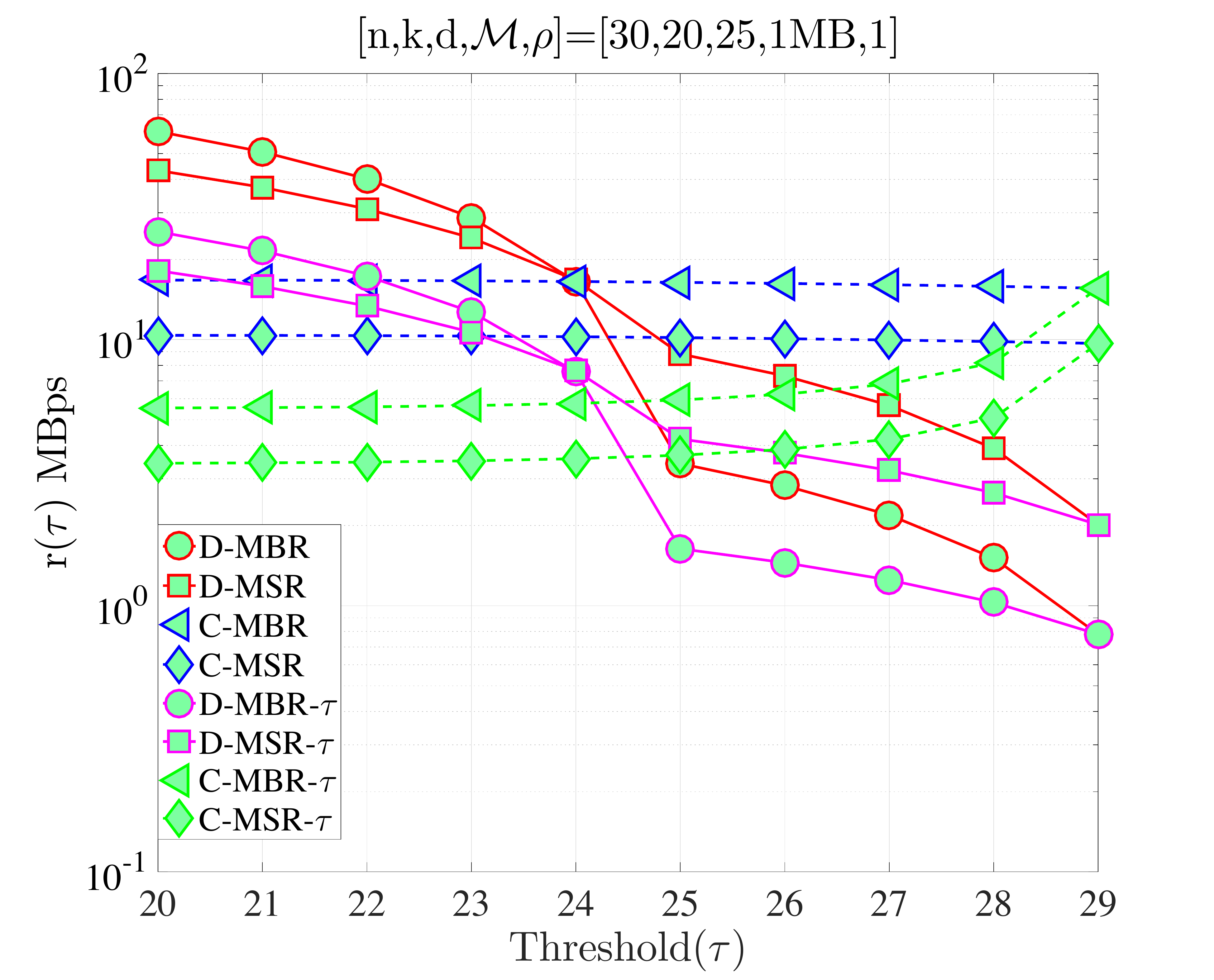} \\
			(a) & (b)\\
		\end{tabular}
	\end{center}
	\vspace{-0.1in}
	\caption{Cost $r(\tau)$ vs repair threshold $(\tau)$ for: (a) $\rho = 0.01$, (b) $\rho = 1$.}
	\label{fig:max_clock-exp} 
\end{figure*}

\begin{figure*}[tb]
	\begin{center}
		\setlength{\tabcolsep}{-0.01in}
		\begin{tabular}{ccc}
			\includegraphics[height=2.3in,width=2.45in]{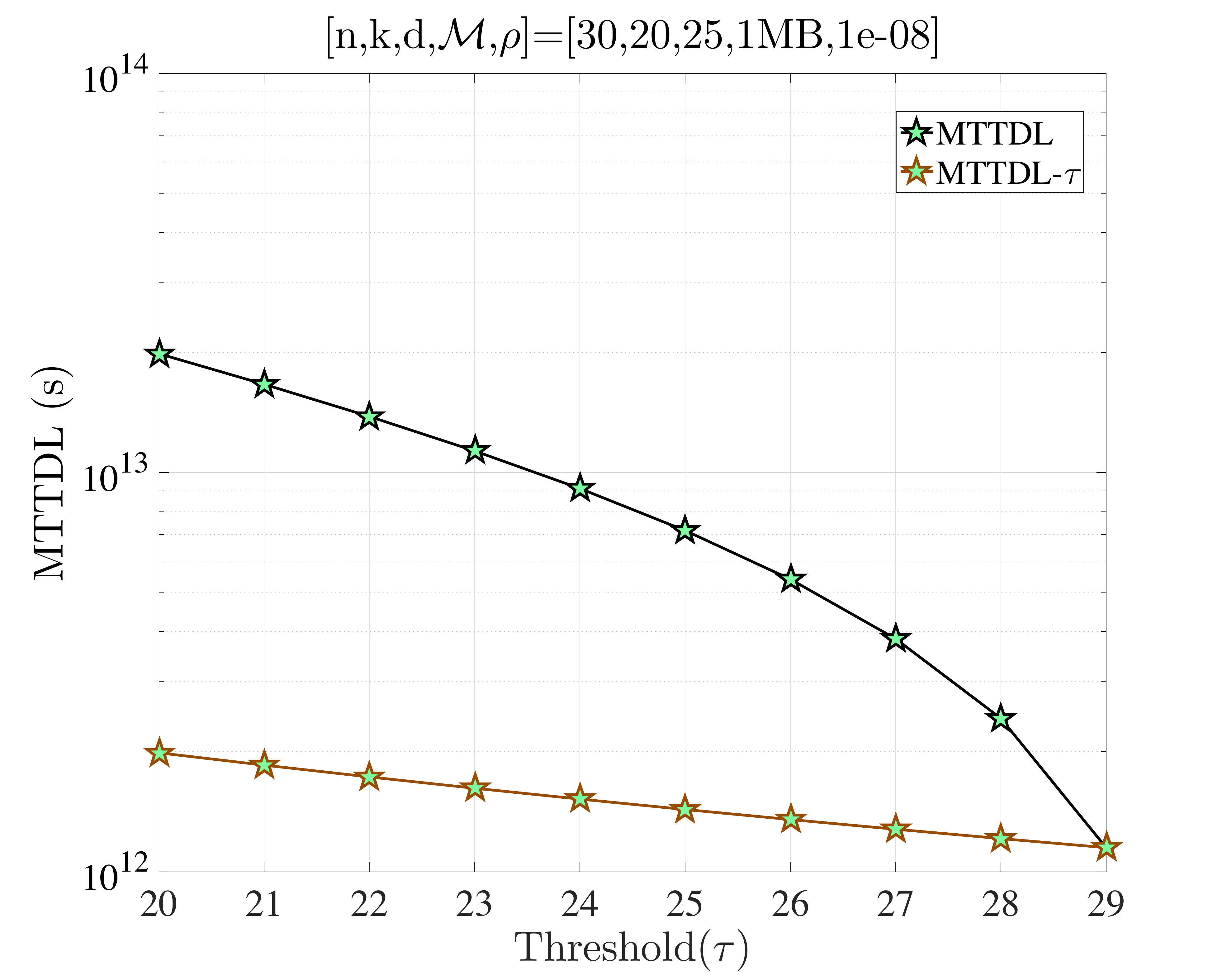} &
			\includegraphics[height=2.3in,width=2.45in]{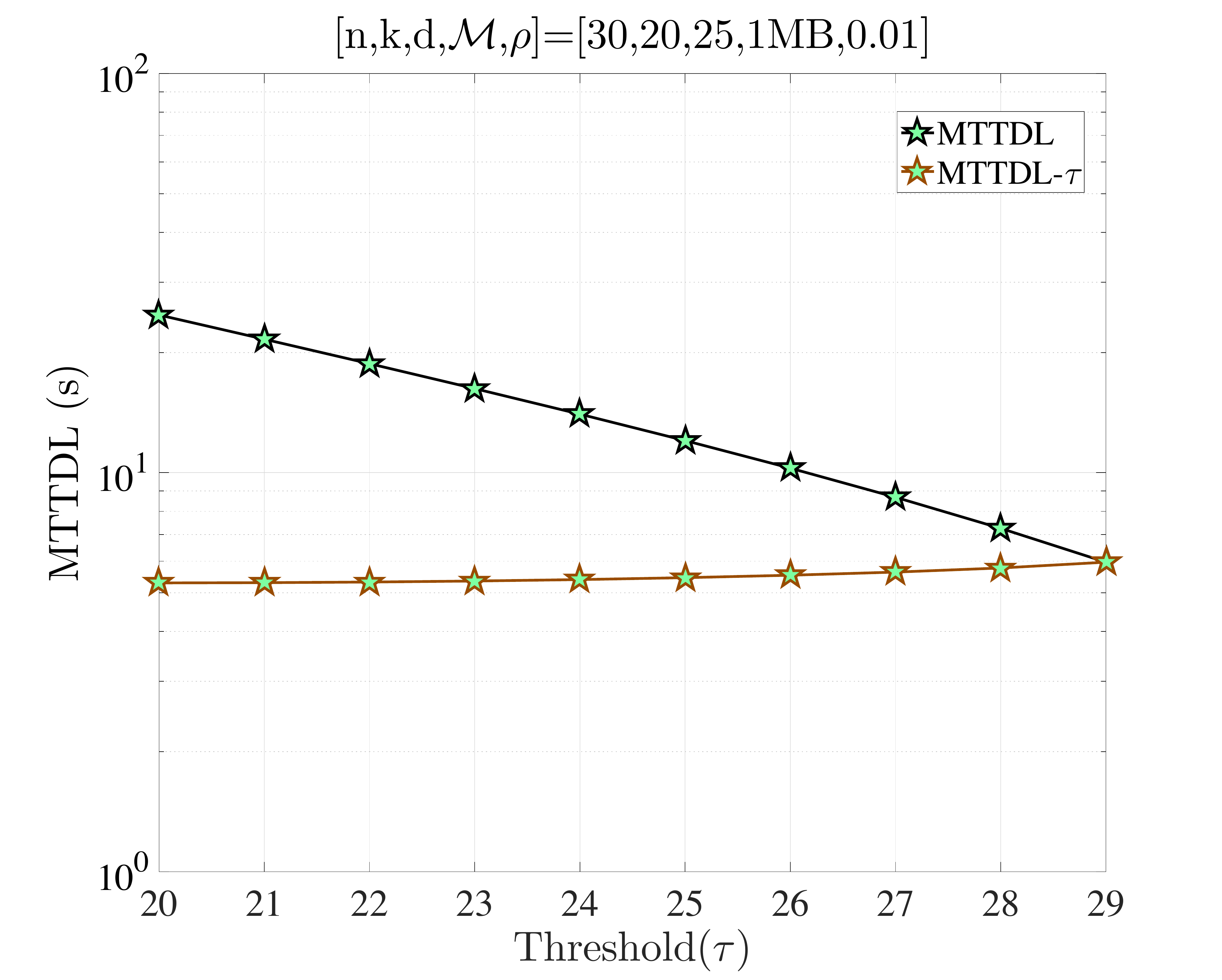} &
			\includegraphics[height=2.3in,width=2.45in]{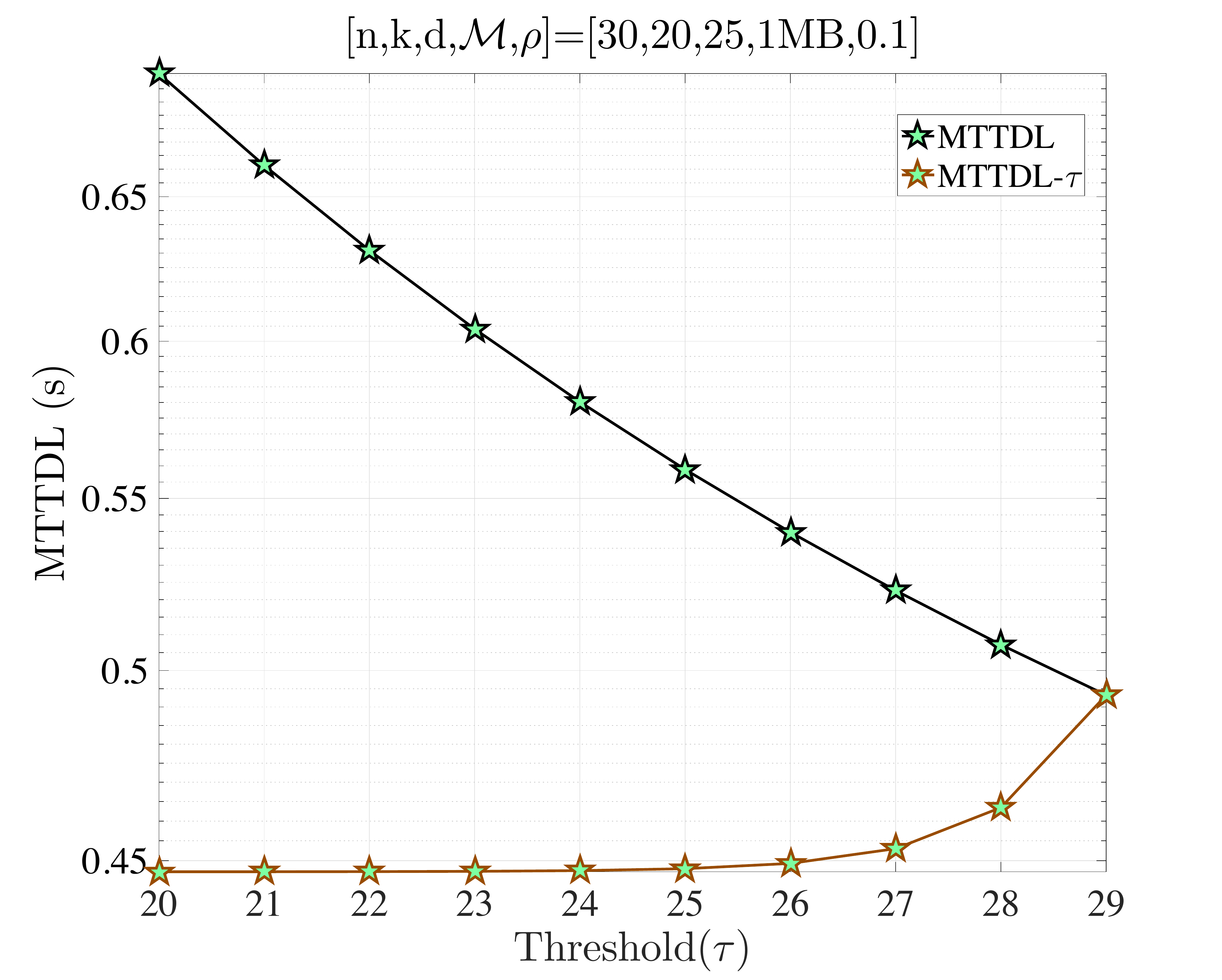} \\
			(a) &(b) & (c)
		\end{tabular}
	\end{center}
	\vspace{-0.1in}
	\caption{Mean time to data loss vs repair threshold $(\tau)$ for: (a) $\rho = 0.004$, (b) $\rho = 0.1$, (c) $\rho = 1$.}
	\label{fig:mttdl} 
	\vspace{-0.2in}
\end{figure*}

\section{Allowing Node Departures within the Repair Process}
\label{sec:allow_failure}

Up to this point, we have assumed that once the repair process is started, no live nodes depart from $\A$ until the repair is completed. In this section, we analyze the case in which additional node departures are allowed within repair process. Fig.~\ref{fig:general_mc} shows the revised CTMC model that accounts for departures during the repair process. The chain consists of two sets of states. In the first row, states represent the phase where nodes depart but no repairs are initiated. Once state $\tau$ is reached, repairs are initiated and the chain transitions into the second row of states where departures may occur while repairs are performed. We focus on the case where no data loss occurs as we are interested in the system dynamics, while the system remains operational. 

Once the repair process is initiated, $n-\tau$ nodes are to be repaired. Since we assume an exponential random variable for repair times (with rate $\mu$), a newcomer is repaired after the minimum of the $n-\tau$ exponential random variables, which is also an exponential random variable with rate $(n-\tau)\mu$. Due to the memoryless property of the exponential random variable, we can perform the same procedure for the second repair and so on. The corresponding rates are depicted on the second row of states of Fig.~\ref{fig:general_mc}. Note that if no departures occur during the repair process, the expected repair time would be the same as in the model of Section~\ref{sec:dependent_tau}, that is, the maximum of $n-\tau$ exponential random variables with rate $\mu$. 

\begin{figure}[t]
	\begin{center}
		\includegraphics[width=2.8in]{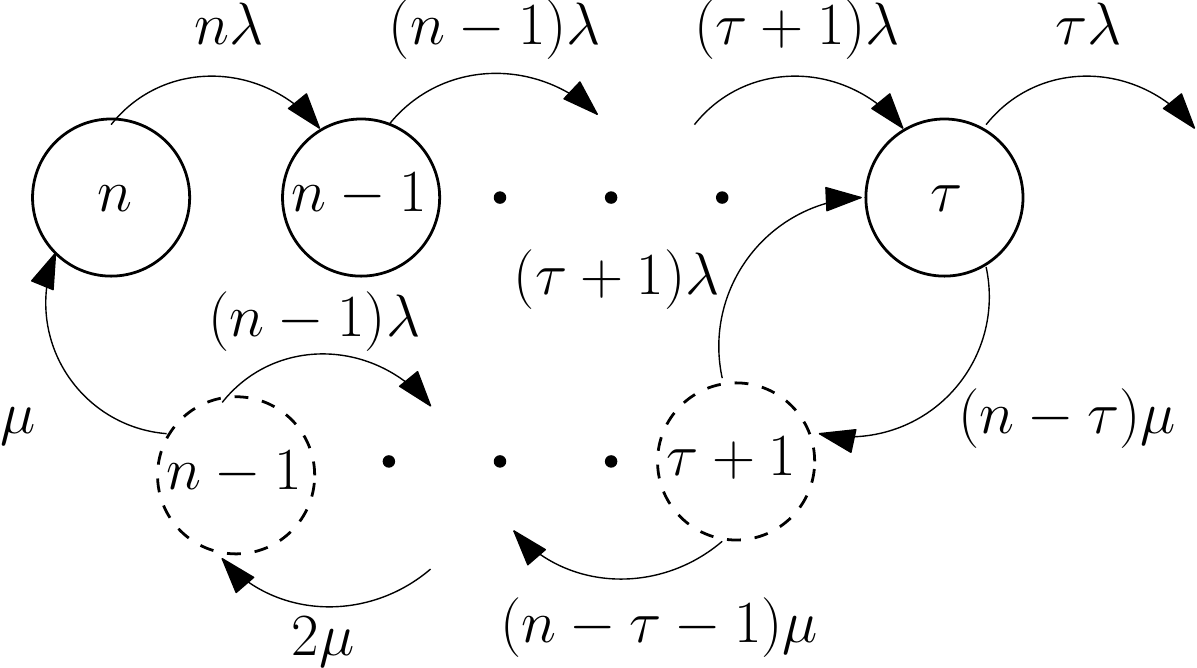}
	\end{center}
	\vspace{-0.1in}
	\caption{Markov chain for a threshold-based file maintenance.}
	\label{fig:general_mc}
	\vspace{-0.1in}
\end{figure}

In such a model, the expected number of node repairs is larger than $n-\tau$ due to the possible additional departures during the repair process. In the section, we analyze the system under the condition that the CTMC always chooses the transition $\tau\rightarrow \tau+1$ when in the lower set of states so that all nodes are repaired eventually. Otherwise, the system would suffer from data loss. To find the average repair cost, we are interested in two statistics of the CTMC once it reaches state $\tau$ for the first time: i) the total number of lower arc transitions before reaching state $n$, i.e., $\tau \rightarrow \tau+1, n-2 \rightarrow n-1$, which will determine the number of node repairs, ii) the expected total time for reaching state $n$ from state $\tau$. For simplicity, the states $n-1,\dots,\tau+1$ can be ignored for now since they do not affect any of these two statistics. We are also interested in the number of times the chain revisits state $\tau$ before reaching state $n$ since that will determine the probability of no data loss. 

\textbf{Total number of revisits to state $\tau$:} Since we are interested in the case where no data loss occurs, we need to find the number of times the chain revisits state $\tau$. This is equivalent to transitioning from state $\tau$ to $\tau+1$ at each revisit to state $\tau$, we need the transition $\tau \rightarrow \tau+1$ at each revisit to state $\tau$, which occurs with probability $\frac{(n-\tau)\mu}{\tau\lambda + (n-\tau)\mu}$. For the transitions in the lower row of the CTMC, denote by $X^r_j$, the total number of revisits to state $\tau$ before reaching state $n$ at state $j$. Then, we can state the following balance equations.

\begin{align}
\label{eq:revisits}
X^r_{n-1}&=X^r_{n-2}\frac{(n-1)\lambda}{(n-1)\lambda+\mu},\\
X^r_{n-2}&=X^r_{n-1}\frac{2\mu}{(n-2)\lambda+2\mu} + X^r_{n-3}\frac{(n-2)\lambda}{(n-2)\lambda+2\mu},\\
\vdots \\
X^r_{\tau+2}&=X^r_{\tau+3}\frac{(n-\tau-2)\mu}{(\tau+2)\lambda+(n-\tau-2)\mu} + X^r_{\tau+1}\frac{(\tau+2)\lambda}{(\tau+2)\lambda+(n-\tau-2)\mu},\\
X^r_{\tau+1}&=X^r_{\tau+2}\frac{(n-\tau-1)\mu}{(\tau+1)\lambda+(n-\tau-1)\mu} + (1+X^r_{\tau})\frac{(\tau+1)\lambda}{(\tau+1)\lambda+(n-\tau-1)\mu},\\
X^r_{\tau} &= X^r_{\tau+1}.
\label{eq:revisits-2}
\end{align}

The set of equations \eqref{eq:revisits}-\eqref{eq:revisits-2} can be recursively solved for $X^r_{\tau}$. Then, the probability that no data loss occurs is $(\frac{(n-\tau)\mu}{\tau\lambda + (n-\tau)\mu})^{1+X^r_{\tau}}$ for one cycle of node repairs. 

\textbf{Total number of lower arc transitions:} For the transitions of the lower row of the CTMC, denote by $X^l_j$, the total number of lower arc transitions at state $j$, then we have the following equations.

\begin{align}
X^l_{n-1}&=\frac{\mu}{(n-1)\lambda+\mu} + X^l_{n-2}\frac{(n-1)\lambda}{(n-1)\lambda+\mu},\\
X^l_{n-2}&=(1+X^l_{n-1})\frac{2\mu}{(n-2)\lambda+2\mu} + X^l_{n-3}\frac{(n-2)\lambda}{(n-2)\lambda+2\mu},\\
\vdots \\
X^l_{\tau+1}&=(1+X^l_{\tau+2})\frac{(n-\tau-1)\mu}{(\tau+1)\lambda+(n-\tau-1)\mu} + X^l_{\tau}\frac{(\tau+1)\lambda}{(\tau+1)\lambda+(n-\tau-1)\mu},\\
X^l_{\tau} &= 1 + X^l_{\tau+1}.
\end{align}

Finding the total number of lower arc transitions is not enough to calculate the average repair cost per time since not all repairs have the same cost in some repair strategies. That is, if $\tau<d$, then some nodes are repaired by downloading $k\alpha$ symbols, whereas the remaining nodes are repaired by downloading $d\beta$. To find the number of lower arc transitions which occur between states $\tau$ and $d$ (which are repaired by downloading $k\alpha$), we denote by $Y^l_j$ the number of lower arc transitions between states $\tau$ and $d$ at state $j$. Then,

\begin{align}
Y^l_{n-1} &= Y^l_{n-2}\frac{(n-1)\lambda}{(n-1)\lambda+\mu}\\
Y^l_{n-2} &= Y^l_{n-1}\frac{2\mu}{(n-2)\lambda+2\mu} + Y^l_{n-3}\frac{(n-2)\lambda}{(n-2)\lambda+2\mu},\\
\vdots \\
Y^l_{d} &= Y^l_{d+1}\frac{(n-d)\mu}{d \lambda+(n-d)\mu} + Y^l_{d-1}\frac{d \lambda}{d \lambda+(n-d)\mu},\\
Y^l_{d-1} &= (1+Y^l_d)\frac{(n-d+1)\mu}{(d-1)\lambda+(n-d+1)\mu} + Y^l_{d-2}\frac{(d-1)\lambda}{(d-1)\lambda+(n-d+1)\mu},\\
\vdots \\
Y^l_{\tau+1}&=(1+Y^l_{\tau+2})\frac{(n-\tau-1)\mu}{(\tau+1)\lambda+(n-\tau-1)\mu} + Y^l_{\tau}\frac{(\tau+1)\lambda}{(\tau+1)\lambda+(n-\tau-1)\mu},\\
Y^l_{\tau} &= 1 + Y^l_{\tau+1}.
\end{align}

If $\tau \geq d$, there is no need to find $Y^l_{\tau}$ since all nodes download $d\beta$ and there are $X^l_{\tau}$ node repairs in total whereas if $\tau < d$, then $Y^l_{\tau}$ repairs are performed by downloading $k\alpha$ symbols and $X^l_{\tau}-Y^l_{\tau}$ node repairs are performed by downloading $d\beta$ symbols. 

\textbf{Expected total time before reaching state $n$:} For the transitions of the lower row of the CTMC, denote by $X^t_j$, the time it takes to reach state $n$ from state $j$. Then, the balance equations for the CTMC for $X^t_j$ can be written as

\begin{align}
	\label{eq:totalTime}
X^t_{n-1}&=\frac{1}{(n-1)\lambda+\mu} + X^t_{n-2}\frac{(n-1)\lambda}{(n-1)\lambda+\mu},\\
X^t_{n-2}&=\frac{1}{(n-2)\lambda+2\mu} + X^t_{n-1}\frac{2\mu}{(n-2)\lambda+2\mu} + X^t_{n-3}\frac{(n-2)\lambda}{(n-2)\lambda+2\mu},\\
\vdots \\
X^t_{\tau+1}&=\frac{1}{(\tau+1)\lambda+(n-\tau-1)\mu} + X^t_{\tau+2}\frac{(n-\tau-1)\mu}{(\tau+1)\lambda+(n-\tau-1)\mu} + X^t_{\tau}\frac{(\tau+1)\lambda}{(\tau+1)\lambda+(n-\tau-1)\mu}.\\
X^t_{\tau} &= \frac{1}{(n-\tau)\mu} + X^t_{\tau+1}.
\label{eq:totalTime-2}
\end{align}

Using \eqref{eq:totalTime}-\eqref{eq:totalTime-2}, we can solve for $X^t_{\tau}$ which is the time it takes to repair to a fully operational system with $n$ live nodes from state $\tau$, when departures occur during the repair process. The initial state for the system is $n$ and therefore the time to revisit state $n$ is $X^t_{\tau} + \sum_{i=\tau+1}^{n}\frac{1}{i\lambda}$.

This yields an average cost per time equal to

\begin{equation}
\label{eq:mostGenRepCost}
r_D(\tau|\textrm{no data loss occurs}) = \frac{c_D(\tau|\textrm{no data loss occurs})}{\textrm{E}[\Delta|\textrm{no data loss occurs}]} =
\begin{cases}
\frac{Y^l_{\tau}k\alpha+(X^l_{\tau}-Y^l_{\tau})d\beta}{X^t_{\tau} + \sum_{i=\tau+1}^{n}\frac{1}{i\lambda}}, & \textrm{if } \tau < d \\
\frac{X^l_{\tau}d\beta}{X^t_{\tau} + \sum_{i=\tau+1}^{n}\frac{1}{i\lambda}}, & \textrm{if } \tau \geq d. \\
\end{cases}
\end{equation}
and a probability of no data loss equal to $(\frac{(n-\tau)\mu}{\tau\lambda + (n-\tau)\mu})^{1+X^r_{\tau}}$.

\subsection{Numerical Results}

We have performed simulations to verify our findings. In the following example, we examine the case where $n=30$, $d=27$, $k=20$, $\mu=10$. Different $\lambda$ values are used, $\lambda=[0.1,0.2,0.4]$, as well as different $\tau$ values, $\tau=[25,27]$.
Simulation results are the average of one million simulations and they are presented in Table~\ref{tb:simulation_results}. Each table entry notes the value obtained from the simulation or the value obtained by analytically evaluation the average repair cost via \eqref{eq:mostGenRepCost}: the first shows the simulation results (S) and the other represents the analytical result (A) as shown before. As expected, increasing $\lambda$ results in more revisits to state $\tau$ due to an increase in the node departure rate. Furthermore, we see that expected time before reaching state $n$ decreases as we increase $\lambda$ for both cases of $\tau = 25$ and $\tau = 27$. For both cases, we observe that the number of nodes repaired by downloading $d\beta$ remains the same for a given value of $\lambda$. This is because the critical number of live nodes for a node to be repaired by downloading $d\beta$ symbols is at $d=27$, since if there are less than $d$ live nodes, then a node must be repaired by downloading $k\alpha$ symbols. In other words, once we reach state $d$ (in the lower row) of the CTMC, we count the number of lower arc transitions between the states $d$ and $n$ to calculate the number of node repairs by downloading $d\beta$ symbols, which remains the same for $\tau=25$ and $\tau=27$, since $\tau$  is not between $d$ and $n$. Finally, the number of nodes repaired by downloading $k\alpha$ increases with $\lambda$ as expected. We can observe that simulation results verify our analytical findings.

\begin{table}[!htb]
	\caption{Numerical results}
	\begin{minipage}{.5\linewidth}
		\caption{Total number of revisits to state $\tau$}
		\centering
\begin{tabular}{cc|c|c|c|c|c|c|}
	\cline{3-8}
	 &  & \multicolumn{6}{c|}{$\lambda$}                                                 \\ \cline{3-8} 
	&  &  \multicolumn{2}{c|}{0.1} & \multicolumn{2}{c|}{0.2} & \multicolumn{2}{c|}{0.4} \\ \cline{3-8} 
	&  &      S     &     A      &      S     &       A    &     S      &       A    \\ \hline
	\multicolumn{1}{|c|}{\multirow{2}{*}{$\tau$}} & 25 &     1.0718      &     1.0719      &     1.1633      &    1.1638       &     1.4660      &    1.4668       \\ \cline{2-8} 
	\multicolumn{1}{|c|}{}                 & 27 &     1.1806      &       1.1806    &     1.4439      &    1.4424       &     2.2118      &   2.2096        \\ \hline
\end{tabular}
	\end{minipage}%
	\begin{minipage}{.5\linewidth}
		\centering
		\caption{Expected total time before reaching state $n$}
\begin{tabular}{cc|c|c|c|c|c|c|}
	\cline{3-8}
	&  & \multicolumn{6}{c|}{$\lambda$}                                                 \\ \cline{3-8} 
	&  &  \multicolumn{2}{c|}{0.1} & \multicolumn{2}{c|}{0.2} & \multicolumn{2}{c|}{0.4} \\ \cline{3-8} 
	&  &      S     &     A      &      S     &       A    &     S      &       A    \\ \hline
	\multicolumn{1}{|c|}{\multirow{2}{*}{$\tau$}} & 25 &     2.0438      &     2.0432      &     1.1770      &    1.1770       &     0.8040      &    0.8034       \\ \cline{2-8} 
	\multicolumn{1}{|c|}{}                 & 27 &     1.2404      &       1.2392    &     0.7458      &    0.7447       &     0.5402      &   0.5405        \\ \hline
\end{tabular}
	\end{minipage} 
	\begin{minipage}{.5\linewidth}
		\caption{Number of node repairs by downloading $d\beta$}
		\centering
\begin{tabular}{cc|c|c|c|c|c|c|}
	\cline{3-8}
	&  & \multicolumn{6}{c|}{$\lambda$}                                                 \\ \cline{3-8} 
	&  &  \multicolumn{2}{c|}{0.1} & \multicolumn{2}{c|}{0.2} & \multicolumn{2}{c|}{0.4} \\ \cline{3-8} 
	&  &      S     &     A      &      S     &       A    &     S      &       A    \\ \hline
	\multicolumn{1}{|c|}{\multirow{2}{*}{$\tau$}} & 25 &     3.4703      &     3.4706      &     4.0263      &    4.0224       &     5.3702      &    5.3696       \\ \cline{2-8} 
	\multicolumn{1}{|c|}{}                 & 27 &     3.4715     &       3.4706    &     4.0237      &    4.0224       &     5.3727      &  5.3696        \\ \hline
\end{tabular}
	\end{minipage}%
	\begin{minipage}{.5\linewidth}
		\centering
		\caption{Number of node repairs by downloading $k\alpha$}
\begin{tabular}{cc|c|c|c|c|c|c|}
	\cline{3-8}
	&  & \multicolumn{6}{c|}{$\lambda$}                                                 \\ \cline{3-8} 
	&  &  \multicolumn{2}{c|}{0.1} & \multicolumn{2}{c|}{0.2} & \multicolumn{2}{c|}{0.4} \\ \cline{3-8} 
	&  &      S     &     A      &      S     &       A    &     S      &       A    \\ \hline
	\multicolumn{1}{|c|}{\multirow{2}{*}{$\tau$}} & 25 &     2.1784      &     2.1782      &     2.4228      &    2.4234       &     3.2620      &   3.2623       \\ \cline{2-8} 
	\multicolumn{1}{|c|}{}                 & 27 &    0     &       0    &     0     &    0      &     0      &   0        \\ \hline
\end{tabular}
	\end{minipage} 
	\label{tb:simulation_results}
\end{table}

In Fig.~\ref{fig:comp_mostgen}, we compare our previous schemes, namely the distributed repair model in Section~\ref{sec:repair_strategies} and the model in Section~\ref{sec:dependent_tau}, with the repair model discussed in this section, which is depicted in the figure with D-MBR-F and D-MSR-F to specify that these codes allow failures within repair process. In the low $\rho$ regime, we observe that there is almost no difference between the performance of models in both MSR and MBR cases. This is because for low $\rho$, the expected number of additional departures during repair is low and the expected time is dominated by terms with $\lambda$ (the upper transition in the CTMC model of Fig.~\ref{fig:general_mc}, which is the same as the CTMC model of Fig.~\ref{fig:mc_d}). In the high $\rho$ regime, differences are observed in the expected average cost per time. Interestingly, when $\rho = 0.1$, it can be observed that for $\tau \geq d -1$, D-MSR-F and D-MBR-F have the highest cost respectively for MSR and MBR cases, whereas for $\tau < d -1$, they are in between the D-MSR and D-MBR. As we keep increasing $\rho$, we observe that D-MSR-F and D-MBR-F becomes even more costly compared to other schemes. In all cases, our model in Section~\ref{sec:dependent_tau} has the lowest $r(\tau)$ respectively for MSR and MBR cases. These observations validate that our model in Section~\ref{sec:repair_strategies} (that does not have dependency of repair process on the number of nodes to be repaired and neglect failures during repair) can be utilized as an approximate model for the models we consider later in the text in the low $\rho$ regime. Therefore, in the low $\rho$ regime, all the optimal threshold statements for the former model (i.e., Proposition~\ref{thm:lambdaProp1}, Proposition~\ref{thm:lambdaProp2}) also hold for the model considered in Sections~\ref{sec:dependent_tau} and \ref{sec:allow_failure}. Note that, the low $\rho$ regime is the only interesting case for mobile cloud storage. At high $\lambda$, the average repair cost and MTTDL performance become prohibitively high and low, respectively, for the system to be viable.

\begin{figure*}[tb]
	\begin{center}
		\setlength{\tabcolsep}{-0.01in}
		\begin{tabular}{ccc}
			\includegraphics[height=2.2in,width=2.4in]{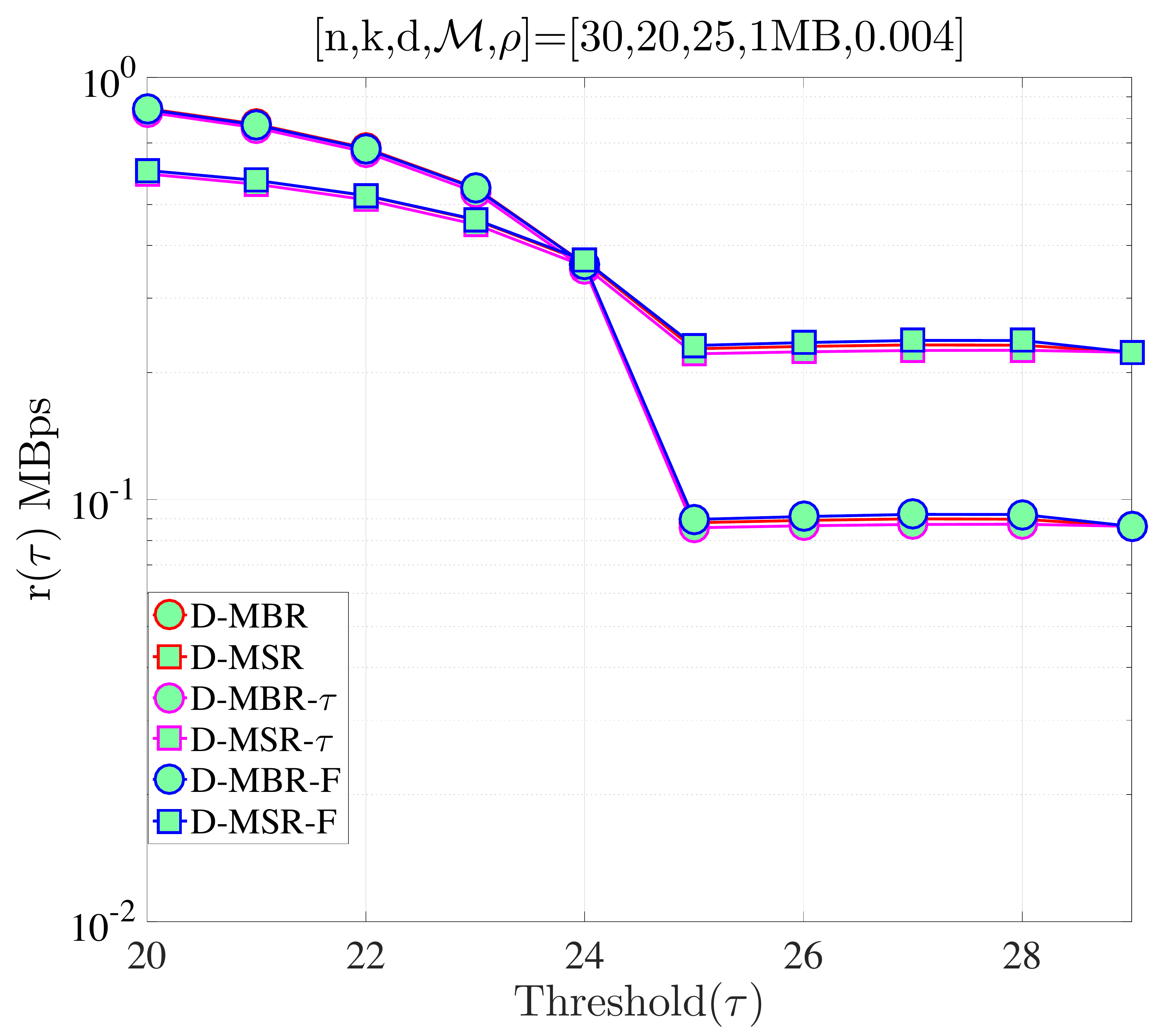} &
			\includegraphics[height=2.2in,width=2.4in]{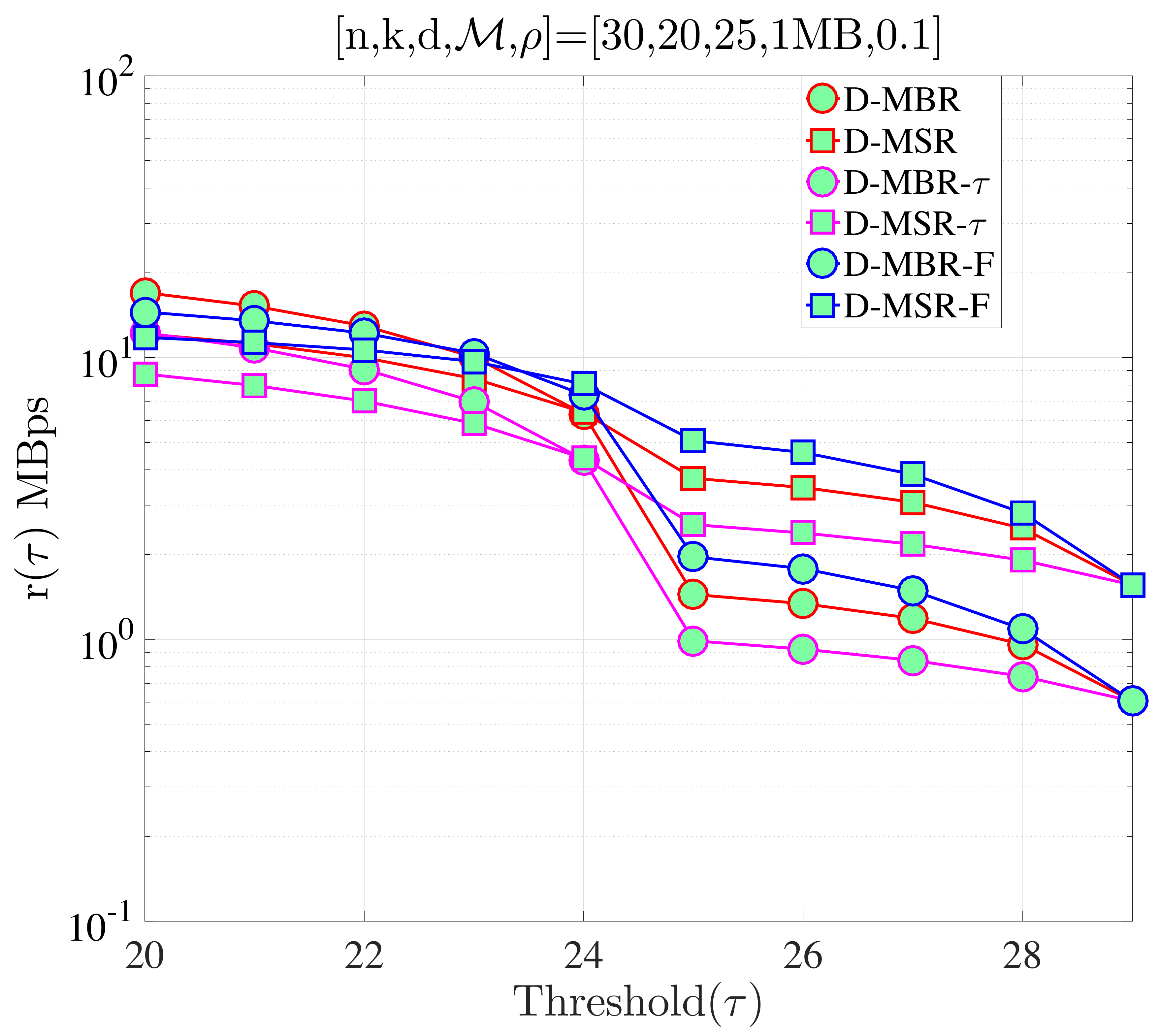} &
			\includegraphics[height=2.2in,width=2.4in]{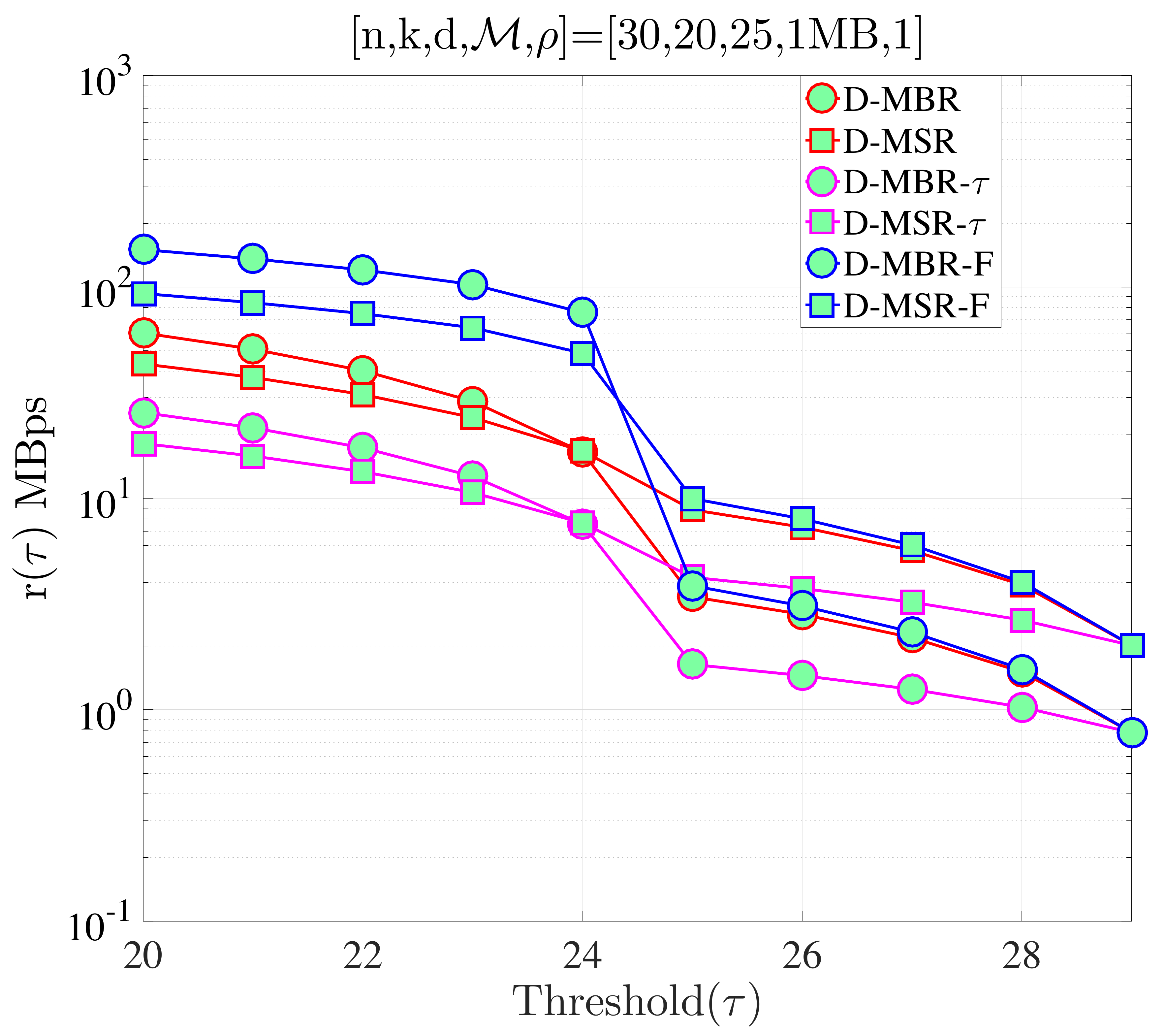} \\
			(a) &(b) & (c)
		\end{tabular}
	\end{center}
	\vspace{-0.1in}
		\caption{Cost $r(\tau)$ vs repair threshold $(\tau)$ for: (a) $\rho = 0.004$, (b) $\rho = 0.1$, (c) $\rho = 1$.}
	\label{fig:comp_mostgen} 
	\vspace{-0.2in}
\end{figure*}

\section{Conclusion and Future Work}
\label{sec:conclusion}

We analyzed threshold-based repair strategies for maintaining files in mobile cloud storage systems. We derived the optimal repair thresholds for both distributed and centralized repair schemes under fragment regeneration and/or reconstruction. Our results showed that optimal thresholds are dependent on system configurations, the underlying code parameters and mobility-to-repair rate ratio. For high mobility-to-repair scenarios, eager repair minimizes the average repair cost per unit of time. Under low mobility-to-repair ratio, lazy repair is optimal in terms of average repair cost. We investigated codes that perform repair through cooperation. We showed that similar to regenerating codes, one can derive optimal thresholds for cooperative regenerating codes. We also investigated the case when the repair process depends on the number of nodes under repair. Finally, we analyzed the case where we lift the restriction that once the repair process is initiated, no more departures occur. We derived set of linear equations to calculate repair cost, which is verified by simulations. We then showed that the initial fixed-rate repair model, which was simple enough to track analytically, is a good approximation of the complex model in the low $\rho$ regime, which is of interest. This allows us to use the optimal threshold results as well as other results we had from the simpler model.

As part of future work, we plan to consider a more advanced repair model in which fragment repairs occur under a fixed bandwidth constraint. This assumption makes the repair rate $\mu$ dependent on the repair threshold $\tau$. Finally, generalizing the model to the case where repairs are initiated at every state with some probability and studying the cost vs. MTTDL tradeoff under this model is an interesting avenue for further research.



\bibliographystyle{IEEEtran} 
\bibliography{IEEEabrv,./refsThresholdDSS}

\begin{thebibliography}{10}
\providecommand{\url}[1]{#1}
\csname url@samestyle\endcsname
\providecommand{\newblock}{\relax}
\providecommand{\bibinfo}[2]{#2}
\providecommand{\BIBentrySTDinterwordspacing}{\spaceskip=0pt\relax}
\providecommand{\BIBentryALTinterwordstretchfactor}{4}
\providecommand{\BIBentryALTinterwordspacing}{\spaceskip=\fontdimen2\font plus
\BIBentryALTinterwordstretchfactor\fontdimen3\font minus
  \fontdimen4\font\relax}
\providecommand{\BIBforeignlanguage}[2]{{%
\expandafter\ifx\csname l@#1\endcsname\relax
\typeout{** WARNING: IEEEtran.bst: No hyphenation pattern has been}%
\typeout{** loaded for the language `#1'. Using the pattern for}%
\typeout{** the default language instead.}%
\else
\language=\csname l@#1\endcsname
\fi
#2}}
\providecommand{\BIBdecl}{\relax}
\BIBdecl

\bibitem{golrezaei2012device}
N.~Golrezaei, A.~G. Dimakis, and A.~F. Molisch, ``{D}evice-to-device
  collaboration through distributed storage,'' in \emph{Proc. of the {GLOBECOM}
  Conference}, 2012.

\bibitem{paakkonen2013device}
J.~P{\"a}{\"a}kk{\"o}nen, C.~Hollanti, and O.~Tirkkonen, ``{D}evice-to-device
  data storage for mobile cellular systems,'' in \emph{Proc. of the GLOBECOM
  Workshops}, 2013.

\bibitem{golrezaei2014scaling}
N.~Golrezaei, A.~G. Dimakis, and A.~F. Molisch, ``Scaling behavior for
  device-to-device communications with distributed caching,'' \emph{{IEEE}
  Trans. on Information Theory}, vol.~60, no.~7, pp. 4286--4298, 2014.

\bibitem{bhagwan2004}
R.~Bhagwan, K.~Tati, Y.-C. Cheng, S.~Savage, and G.~M. Voelker, ``{T}otal
  recall: system support for automated availability management.'' in
  \emph{Proc. of the NSDI Conference}, 2004.

\bibitem{dabek2004designing}
F.~Dabek, J.~Li, E.~Sit, J.~Robertson, M.~F. Kaashoek, and R.~Morris,
  ``Designing a {DHT} for low latency and high throughput,'' in \emph{Proc. of
  the {NDSI} Symposium}, 2004.

\bibitem{hu2010cooperative}
Y.~Hu, Y.~Xu, X.~Wang, C.~Zhan, and P.~Li, ``{C}ooperative recovery of
  distributed storage systems from multiple losses with network coding,''
  \emph{{IEEE} Journal on Selected Areas in Communications}, vol.~28, no.~2,
  pp. 268--276, 2010.

\bibitem{weatherspoon2002}
H.~Weatherspoon and J.~D. Kubiatowics, ``{E}rasure coding vs. replication: {A}
  quantitive comparison,'' in \emph{Proc. of the 1st International Workshop on
  Peer-to-Peer Systems}, 2002.

\bibitem{giroire2010}
F.~Giroire, J.~Monteiro, and S.~P{\'e}rennes, ``{P}eer-to-peer storage systems:
  a practical guideline to be lazy,'' in \emph{Proc. of the GLOBECOM
  Conference}, 2010.

\bibitem{dimakis2010k}
A.~G. Dimakis, P.~B. Godfrey, Y.~Wu, M.~J. Wainwright, and K.~Ramchandran,
  ``{N}etwork coding for distributed storage systems,'' \emph{{IEEE} Trans. on
  Information Theory}, vol.~56, no.~9, pp. 4539--4551, 2010.

\bibitem{Rashmi:Optimal11}
K.~Rashmi, N.~Shah, and P.~Kumar, ``Optimal exact-regenerating codes for
  distributed storage at the {MSR} and {MBR} points via a product-matrix
  construction,'' \emph{{IEEE} Trans. on Information Theory}, vol.~57, no.~8,
  pp. 5227 --5239, Aug. 2011.

\bibitem{Tamo:Zigzag13}
I.~Tamo, Z.~Wang, and J.~Bruck, ``Zigzag codes: {MDS} array codes with optimal
  rebuilding,'' \emph{IEEE Trans. on Information Theory}, vol.~59, no.~3, pp.
  1597--1616, March 2013.

\bibitem{Cadambe:Asymptotic13}
V.~Cadambe, S.~Jafar, H.~Maleki, K.~Ramchandran, and C.~Suh, ``{A}symptotic
  interference alignment for optimal repair of {MDS} codes in distributed
  storage,'' \emph{{IEEE} Trans. on Information Theory}, vol.~59, no.~5, pp.
  2974--2987, 2013.

\bibitem{borthakur2008hdfs}
D.~Borthakur, ``{HDFS} architecture guide,''
  \url{http://hadoop.apache.org/common/docs/current/hdfs_design.pdf}, 2008.

\bibitem{Singleton:Maximum64}
R.~C. Singleton, ``{M}aximum distance q -nary codes,'' \emph{IEEE Trans. on
  Information Theory}, vol.~10, no.~2, pp. 116--118, 1964.

\bibitem{Blaum:EVENODD95}
M.~Blaum, J.~Brady, J.~Bruck, and J.~Menon, ``{EVENODD}: {A}n efficient scheme
  for tolerating double disk failures in {RAID} architectures,'' \emph{IEEE
  Trans. on Computers}, vol.~44, no.~2, pp. 192--202, Feb. 1995.

\bibitem{Huang:STAR05}
C.~Huang and L.~Xu, ``{STAR}: {A}n efficient coding scheme for correcting
  triple storage node failures,'' in \emph{Proc. of the USENIX Conference on
  File and Storage Technologies}, 2005.

\bibitem{calder2011windows}
B.~Calder, Wang \emph{et~al.}, ``{W}indows {A}zure {S}torage: {A} highly
  available cloud storage service with strong consistency,'' in \emph{Proc. of
  the Twenty-Third ACM Symposium on Operating Systems Principles}, 2011.

\bibitem{ford2010availability}
D.~Ford, F.~Labelle, F.~I. Popovici, M.~Stokely, V.-A. Truong, L.~Barroso,
  C.~Grimes, and S.~Quinlan, ``{A}vailability in globally distributed storage
  systems.'' in \emph{Proc. of the OSDI Conference}, 2010.

\bibitem{kermarrec2011repairing}
A.-M. Kermarrec, N.~Le~Scouarnec, and G.~Straub, ``Repairing multiple failures
  with coordinated and adaptive regenerating codes,'' in \emph{Proc. of the
  International Symposium on Network Coding (NetCod)}, 2011.

\bibitem{Shum2013Cooperative}
K.~W. Shum and Y.~Hu, ``Cooperative regenerating codes,'' \emph{IEEE Trans. on
  Information Theory}, vol.~59, no.~11, pp. 7229--7258, Nov 2013.

\bibitem{wang2010mfr}
X.~Wang, Y.~Xu, Y.~Hu, and K.~Ou, ``Mfr: Multi-loss flexible recovery in
  distributed storage systems,'' in \emph{Proc. of the International Conference
  on Communications (ICC)}, 2010.

\bibitem{rawal2016centralized}
A.~S. Rawat, O.~O. Koyluoglu, and S.~Vishwanath, ``Centralized repair of
  multiple node failures with applications to communication efficient secret
  sharing,'' \emph{arXiv preprint arXiv:1603.04822}, 2016.

\bibitem{paakkonen2012distributed}
J.~P{\"a}{\"a}kk{\"o}nen, P.~Dharmawansa, C.~Hollanti, and O.~Tirkkonen,
  ``{D}istributed storage for proximity based services,'' in \emph{Proc. of the
  Swedish Communication Technologies Workshop (Swe-CTW)}, 2012.

\bibitem{paakkonen2014device}
J.~P{\"a}{\"a}kk{\"o}nen, C.~Hollanti, and O.~Tirkkonen, ``Device-to-device
  data storage with regenerating codes,'' in \emph{Multiple Access
  Communications}.\hskip 1em plus 0.5em minus 0.4em\relax Springer, 2015, pp.
  57--69.

\bibitem{pedersen2016distributed}
J.~Pedersen, I.~Andriyanova, F.~Br{\"a}nnstr{\"o}m \emph{et~al.}, ``Distributed
  storage in mobile wireless networks with device-to-device communication,''
  \emph{arXiv preprint arXiv:1601.00397}, 2016.

\bibitem{Eisenberg:On08}
B.~Eisenberg, ``On the expectation of the maximum of iid geometric random
  variables,'' \emph{Statistics \& Probability Letters}, vol.~78, no.~2, pp.
  135--143, 2008.

\end{thebibliography}

\begin{appendices}
	
	\section{Proof of Proposition \ref{thm:lambdaProp1}}
	\label{sec:lambdaProp1}
	\begin{proof}
		To determine $\tau^*$, we compare $r_D(d)$ with the average repair cost at all other possible regeneration states $d+\delta$, for $1 \leq \delta \leq n-d-1$, and check if 
		\begin{equation}
		r_D(d)\leq r_D(d+\delta),~~\forall \delta\in[1,n-d-1].
		\label{rates}
		\end{equation}
		This method is preferred because a straightforward minimization of $r_D(\tau)$ through differentiation is involved due to the harmonic sums. Substituting $r(\tau)$ from \eqref{rtau_d} to \eqref{rates} yields, 
		\begin{align}
		\frac{\gamma (n-d) \lambda \mu}{\mu H_{n,d} + \lambda} & \leq \frac{\gamma (n-d-\delta)\lambda \mu}{\mu H_{n,d+\delta} + \lambda},
		\\
		\rho & \leq \frac{(n-d)H_{d+\delta,d}}{\delta} - H_{n,d}.
		\label{prop1}
		\end{align}
		The inequality in \eqref{prop1} yields the maximum $\rho$ for which $r_D(\tau)$ is minimized at state $\tau=d$. We now examine the behavior of the right hand side (RHS) in \eqref{prop1} as a function of $\delta$ for fixed $n$ and $d$. The RHS  in \eqref{prop1} has the same monotonicity as the function $f(d,\delta)=\frac{H_{d+\delta,d}}{\delta}.$ In Lemma~\ref{thm:Lemmafcn1}, we show that $f$  is monotonically decreasing with $\delta$. As a result, the departure-to-repair rates $\rho$ for which $\eqref{rates}$ holds are also monotonically decreasing with $\delta.$ Substituting the maximum $\delta$ (i.e., $\delta=n-d-1$) to the RHS in \eqref{prop1} yields a departure rate bound
		\begin{equation}
		\rho \leq \frac{H_{n-1, d}}{n-d-1}  -\frac{1}{n},
		\end{equation}
		for which $r_D(d)\leq r_D(d+\delta), \forall \delta\in[1,n-d-1].$ In this case, minimization of  $r_D(\tau)$ is achieved at $\tau^*=d.$

		We now prove that for rates $\rho> \frac{H_{n-1, d}}{n-d-1}  -\frac{1}{n},$ the average cost $r_D(\tau)$ is minimized when $\tau=n-1$. Following a similar reasoning, we compare $r_D(\tau)$ at $\tau=n-1$ with $r_D(\tau)$ at any other possible regeneration threshold. We consider $r_D(n-1)\leq r(n-\delta-1)$, where $1\leq \delta \leq n-d-1$. By substituting $r_D(\tau)$ from \eqref{rtau_d} and simplifying, it follows that
		\begin{equation}
		\rho \geq \frac{H_{n-1,n-\delta-1}}{\delta} -\frac{1}{n}.
		\label{prop1_2}
		\end{equation}
		The RHS of \eqref{prop1_2} has the same monotonicity as the function $g(n-1,\delta)=\frac{H_{n-1,n-1-\delta}}{\delta}$. In  Lemma~\ref{thm:Lemmafcn2}, we show that $g$  is monotonically increasing with $\delta$. Therefore, the minimum $\rho$ for which  $r_D(n-1)\leq r_D(n-1-\delta),~\forall \delta\in[1,n-d-1]$ is obtained when $\delta=n-d-1.$ Substituting this $ \delta$ to \eqref{prop1_2} completes the proof.	
	\end{proof}

	
	
	\begin{lem}
		\label{thm:Lemmafcn1}
		The function
		\begin{equation}\label{eq:fcn1}
		f(x,\delta)=\frac{H_{x+\delta,x}}{\delta}
		\end{equation}
		is a monotonically decreasing function over integers $\delta>0$ for any given integer $x>0$.
	\end{lem}

	\begin{proof}
		We will show that $f(x,\delta+1)< f(x,\delta)$ for any integer $\delta>0$, which implies the monotonically decreasing assertion in the lemma.
		We have
		\begin{align*}
		&f(x,\delta)-f(x,\delta+1) \\
		&=
		\frac{1}{\delta}\left( 
		\sum\limits_{i=x+1}^{x+\delta} \frac{1}{i}
		\right)
		-\frac{1}{\delta+1}\left( 
		\sum\limits_{i=x+1}^{x+\delta+1} \frac{1}{i}
		\right)\\
		&\stackrel{(a)}{=}\frac{S}{\delta}-\frac{S+\frac{1}{x+\delta+1}}{\delta+1}\\
		&=\frac{1}{\delta+1}\left( S\left(\frac{\delta+1}{\delta}-1\right) - \frac{1}{x+\delta+1} \right)\\
		&=\frac{1}{\delta(\delta+1)}\left( S - \frac{\delta}{x+\delta+1} \right)\\
		&\stackrel{(b)}{>} 0,\\
		\end{align*}
		where in (a) we define the sum $S=\sum\limits_{i=x+1}^{x+\delta} \frac{1}{i}$, and (b) follows as there are $\delta$ terms in $S$ and each term is strictly greater than $\frac{1}{x+\delta+1}$.	
	\end{proof}


	
	\begin{lem}
		\label{thm:Lemmafcn2}
		The function
		\begin{equation}\label{eq:fcn2}
		g(x,\delta)=\frac{H_{x,x-\delta}}{\delta}
		\end{equation}
		is a monotonically increasing function over integers $\delta\in[0,x-1]$ for any given integer $x>0$.
	\end{lem} 
	
	\begin{proof}
		We will show that $g(x,\delta+1)> g(x,\delta)$ for any integer $\delta>0$, which implies the monotonically increasing assertion in the lemma.
		We have
		\begin{align*}
		&g(x,\delta+1)-g(x,\delta) \\
		&=
		\frac{1}{\delta+1}\left( 
		\sum\limits_{i=x-\delta}^{x} \frac{1}{i}
		\right)
		-\frac{1}{\delta}\left( 
		\sum\limits_{i=x-\delta+1}^{x} \frac{1}{i}
		\right)\\
		&\stackrel{(a)}{=}\frac{S+\frac{1}{x-\delta}}{\delta+1}-\frac{S}{\delta}\\
		&=\frac{1}{\delta+1}\left( S\left(1-\frac{\delta+1}{\delta}\right) + \frac{1}{x-\delta} \right)\\
		&=\frac{1}{\delta(\delta+1)}\left(  \frac{\delta}{x-\delta} - S \right)\\
		&\stackrel{(b)}{>} 0,\\
		\end{align*}
		where in (a) we define the sum $S=\sum\limits_{i=x-\delta+1}^{x} \frac{1}{i}$, and (b) follows as there are $\delta$ terms in $S$ and each term is strictly smaller than $\frac{1}{x-\delta}$.
	\end{proof}
	
	
	\section{Proof of Lemma~\ref{thm:lambdaProp1}}
	\label{sec:LemmalambdaProp1}
	
	In Proposition~\ref{thm:prop1}, we determined the $\rho$ regime for which lazy repair is more efficient than eager repair, given fixed code parameters. To prove Lemma~\ref{thm:lambdaProp1}, it suffices to show that the highest rate $\rho = \frac{H_{n-1, d}}{n-d-1}  -\frac{1}{n}$ for which $\tau^{\ast} = d$ is strictly positive for any $n$ and $d$ (recall that by definition, $n>d$). We prove this fact by employing a lower bound on $\frac{H_{n-1, d}}{n-d-1} $ , which is proved in Lemma~\ref{thm:lemmaLambda1}.
	\begin{align}
	\frac{H_{n-1, d}}{n-d-1} &\geq \frac{1}{n-1}  \stackrel{(a)}{\Rightarrow} \\
	\frac{H_{n-1, d}}{n-d-1} &> \frac{1}{n} \\
	\frac{H_{n-1, d}}{n-d-1}  - \frac{1}{n} &>0.
	\end{align}
	where in (a), we substituted $\frac{1}{n-1}$ with the strictly smaller term $\frac{1}{n}.$
	
	
	
	\begin{lem}
		\label{thm:lemmaLambda1}
		The function                                               
		\begin{equation}
		f(x,\delta)=\frac{H_{x+\delta,x}}{\delta}
		\end{equation}
		is bounded by 
		\begin{equation}
		\frac{1}{x+\delta} \leq \frac{H_{x+\delta,x}}{\delta} < \frac{1}{\delta},
		\end{equation}
		for all positive integers $x$ and $\delta$.
	\end{lem}

	\begin{proof}
		First, we show that $\frac{H_{x+\delta, x}}{\delta} \geq \frac{1}{x+\delta}$, for all $x>0$ and $\delta>0$.  
		\begin{align*}
		\frac{H_{x+\delta, x}}{\delta} &= \frac{\sum_{i=1}^{x+\delta} \frac{1}{i} - \sum_{i=1}^{x} \frac{1}{i}}{\delta} \\
		& = \frac{\frac{1}{x+1} + \frac{1}{x+2}+\ldots+\frac{1}{x+\delta}}{\delta} \\
		&\stackrel{(a)}{>} \frac{\frac{1}{x+\delta} + \frac{1}{x+\delta}+\ldots+\frac{1}{x+\delta}}{\delta} \\
		&= \frac{\delta  \frac{1}{x+\delta}}{\delta} \\
		&= \frac{1}{x+\delta},
		\end{align*}
		where (a) follows by substituting the first $\delta-1$ terms in the nominator with strictly smaller terms. The equality in this lower bound holds when $\delta  = 1$. We now show the upper bound. 
		\begin{align*}
		\frac{H_{x+\delta, x}}{\delta} &= \frac{\sum_{i=1}^{x+\delta} \frac{1}{i} - \sum_{i=1}^{x} \frac{1}{i}}{\delta} \\
		& = \frac{\frac{1}{x+1} + \frac{1}{x+2}+\ldots+\frac{1}{x+\delta}}{\delta} \\
		&\stackrel{(b)}{<} \frac{\frac{1}{x} + \frac{1}{x}+\ldots+\frac{1}{x}}{\delta} \\
		&= \frac{\delta  \frac{1}{x}}{\delta} \\
		&= \frac{1}{x},
		\end{align*}
		where (b) follows by substituting the $\delta$ terms in the nominator with strictly larger terms.
	\end{proof}
	
	
	\section{Proof of Proposition \ref{thm:prop2}}
	\label{sec:AppA}
	
	\begin{proof}
		The proof follows along the same lines as Proposition~\ref{thm:prop1}. We compare the repair cost at $r(d)$ with the repair cost at any other possible state $d-\delta$ for $1 \leq \delta \leq d-k$ to check when the inequality
		\begin{equation}
		r(d)\leq r(d-\delta), \quad \delta\in[1,d-k]
		\end{equation}
		is satisfied.
		Substituting for $r(\tau)$ using \eqref{rtau_d}, we obtain
		\begin{equation}
		\frac{\lambda\mu(n-d)\gamma}{\mu H_{n,d} + \lambda} \leq \frac{\lambda\mu(k\alpha\delta+(n-d)\gamma)}{\mu H_{n, d - \delta} + \lambda},
		\end{equation}
		from which we get:
		\begin{equation}
		\rho \geq  \frac{(n-d)\gamma H_{d,d-\delta}}{k\alpha\delta}-\frac{1}{n}-H_{n-1,d}
		\label{lfunc}
		\end{equation}
		Expression \eqref{lfunc} yields a bound on the minimum  $\rho$ for which the optimal repair threshold is $\tau^{*}=d$. 
		We notice that RHS of the inequality above has the same monotonicity as the function $g(d,\delta)=\frac{H_{d,d-\delta}}{\delta}$ defined in Lemma~\ref{thm:Lemmafcn2} in Appendix~\ref{sec:lambdaProp1}, from which we observe that this function is a monotonically increasing function of $\delta$.	 Substituting the maximum $\delta^*=d-k$ yields the departure-to-repair rate bound,      
		\begin{equation}
		\rho \geq \frac{\gamma(n-d)H_{d,k}}{k\alpha(d-k)}-H_{n,d},
		\end{equation}
		for which $r(d)\leq r(d-\delta), \forall \delta\in[1,d-k].$ For this rate regime, the optimal repair threshold is at $\tau^* = d.$

		We now prove that for rates  $\rho \leq \frac{\gamma(n-d)H_{d,k}}{k\alpha(d-k)}-H_{n-1,d}-\frac{1}{n}$, the average cost $r(\tau)$ per unit of time is minimized when $\tau=k$. We compare $r(k)$ with $r(k+\delta)$ to analyze when the following inequality holds.
		\begin{equation}
		r(k)\leq r(k+\delta), \quad \delta\in[1,d-k].
		\end{equation}
		On substituting for $r(\tau)$ from \eqref{rtau_d}, we get:
		\begin{eqnarray*}
			\frac{\lambda\mu (k\alpha(d-k)+\gamma(n-d))}{\mu H_{n,k} + \lambda}
			\leq  \frac{\lambda\mu(k\alpha(d-k-\delta)+\gamma(n-d))}{\mu H_{n,k+\delta} + \lambda},
		\end{eqnarray*}
		from which we obtain
		\begin{equation}
		\rho \leq   \frac{(k\alpha(d-k)+\gamma(n-d))H_{k+\delta,k}}{k\alpha\delta}-\frac{1}{n}-H_{n-1,k}.
		\label{prop4}
		\end{equation}
		The expression above yields a bound on the maximum departure-to-repair rate $\rho$ for which the optimal repair threshold is $\tau^{*}=k$. We now study the behavior of \eqref{prop4} as a function of $\delta$ for fixed $n, k$ and $d$. 
		We notice that RHS of the inequality above has the same monotonicity as the function $f(k,\delta)=\frac{H_{k+\delta,k}}{\delta}$ defined in Lemma~\ref{thm:Lemmafcn1} in Appendix~\ref{sec:lambdaProp1}, from which we observe that this function is a monotonically decreasing function of $\delta$. Therefore, $\delta^*=d-k$ yields the minimum value for the RHS of \eqref{prop4}, which implies that when                                                                                                                                                                                                                                                                         
		\begin{eqnarray}
		\rho &\leq & \frac{(k\alpha(d-k)+\gamma(n-d))H_{d,k}}{k\alpha(d-k)}-H_{n-1,k}-\frac{1}{n} \nonumber\\
		&= & \frac{\gamma(n-d)H_{d,k}}{k\alpha(d-k)}-H_{n,d},
		\end{eqnarray}
		$r(\tau)$ is optimized at $\tau^*=k.$		
	\end{proof}
	
	
	\section{Proof of Lemma \ref{thm:lambdaProp2}}
	\label{sec:lemmalambdaProp2}
	
	\begin{proof}
		We prove Lemma~\ref{thm:lambdaProp2} by finding the relationship between $n, k, \gamma,$ and  $\alpha$ for which the upper bound on the rate $\rho$ when $\tau^*=k$ becomes negative. 
		\begin{align*}
		\frac{\gamma(n-d)H_{d,k}}{k\alpha(d-k)}-H_{n,d}  &=(n-d)\left(\frac{\gamma H_{d,k}}{k\alpha(d-k)}-\frac{H_{n,d}}{n-d}\right) \\
		&\stackrel{(a)}{<} (n - d)\left ( \frac{\gamma}{k^2 \alpha } -  \frac{H_{n,d}}{n-d} \right) \\
		&\stackrel{(b)}{\leq}(n-d)\left ( \frac{\gamma}{k^2 \alpha }  -  \frac{1}{n}\right) \\
		&=\frac{n-d}{k^2 \alpha}\left( n\gamma - k^2 \alpha \right). \\
		\end{align*}
		where in (a) we have used Lemma~\ref{thm:lemmaLambda1} of Appendix~\ref{sec:LemmalambdaProp1} to substitute term $\frac{H_{d,k}}{d-k}$ with its upper bound $\frac{1}{k}$ and in (b), we have used the same lemma to substitute  $\frac{H_{n,d}}{n-d}$ with its lower bound  $\frac{1}{n}$. Setting $\frac{n-d}{k^2 \alpha} \left ( n \gamma - k^2 \alpha \right)<0$ yields $n\gamma < k^2\alpha$ since the multiplicative term is strictly positive.
	\end{proof}


\section{Proof of Proposition \ref{thm:prop3}}
\label{sec:AppB}

\begin{proof}
	To determine $\tau^{*}$, we compare $r(k)$ with other possible repair states $k+\delta$, i.e., we analyze when
	\begin{equation}
	r(k)\leq r(k+\delta),
	\end{equation}
	for $1\leq \delta \leq n-k-1$. On substituting for $r(\tau)$ from \eqref{rtau_c}, we obtain
	\begin{align*}
	\frac{\lambda\mu\alpha(k+n-k-1)}{\mu H_{n,k} + \lambda} \leq \frac{\lambda \mu\alpha(k+n-k-\delta-1)}{\mu H_{n,k+\delta} + \lambda}.
	\end{align*}
	We have
	\begin{equation}\label{prop3}
	\rho \leq\frac{(k\alpha+\alpha(n-k-1))H_{k+\delta,k}}{\alpha\delta}-\frac{1}{n}-H_{n-1,k}.
	\end{equation}
	
	Inequality \eqref{prop3} yields the maximum departure-to-repair rate $\rho$ for which it is more cost-efficient to repair at state $\tau=k$ than any other state $\tau=k+\delta.$ We now examine the behavior of the RHS of \eqref{prop3} as a function of $\delta$ for fixed $k$ and $d$.  
	We notice that RHS of the inequality above has the same monotonicity as the function $f(k,\delta)=\frac{H_{k+\delta,k}}{\delta}$ defined in Lemma~\ref{thm:Lemmafcn1} in Appendix~\ref{sec:lambdaProp1}, from which we observe that this function is a monotonically decreasing function of $\delta$. Substituting $\delta^*=n-k-1$ to the RHS of \eqref{prop3} yields 
	\begin{equation}
	\rho \leq \frac{kH_{n-1, k}}{(n-k-1)} - \frac{1}{n},
	\end{equation}
	for which the optimal repair threshold is at $\tau^* = k.$
	
	We now evaluate if there is a departure rate regime for which the average cost per unit of time is minimized at $\tau^*= n-1$. That is, we analyze when
	\begin{equation}
	r(n-1)\leq r(n-\delta-1).
	\end{equation}
	where $1\leq \delta \leq n-k-1$. On substituting for $r(\tau)$ from \eqref{rtau_c}, we get
	\begin{eqnarray}
	\frac{\lambda\mu\alpha k}{\mu H_{n,n-1}+\lambda} &\leq & \frac{\lambda\mu\alpha(k+\delta)}{\mu H_{n,n-\delta-1} + \lambda} \Rightarrow\\
	\rho & \geq & \frac{kH_{n-1,n-\delta-1}}{\delta} -\frac{1}{n}.
	\label{eq:cent_2}
	\end{eqnarray}
	We notice that RHS of the inequality above has the same monotonicity as the function $g(n-1,\delta)=\frac{H_{n-1,n-1-\delta}}{\delta}$ defined in Lemma~\ref{thm:Lemmafcn2} in Appendix~\ref{sec:lambdaProp1}, from which we observe that this function is a monotonically increasing function of $\delta$. Substituting $\delta^*=n-k-1$  yields 
	\begin{equation}
	\rho \geq \frac{kH_{n-1, k}}{(n-k-1)} - \frac{1}{n},
	\end{equation}
	for which the optimal repair threshold is at $\tau^{*}=n-1$.
	
\end{proof}


\section{Proof of Proposition \ref{thm:prop4}}
\label{sec:Prop4}
\begin{proof}
	Let us consider the following inequality:
	\begin{equation}
	r_D(\tau)<r_C(\tau).
	\label{rtau_dc}
	\end{equation}
	
	According to \eqref{rtau_c}, the average repair cost $r_C(\tau)$ of centralized repair depends only on $\alpha$, when $n$, $k$, and $d$ are fixed. As $\alpha_{MSR} \leq \alpha_{MBR}$, MSR codes minimize $r_C(\tau)$. Thus, we select MSR codes for centralized repair in our comparison. Similarly, for given $n,k,$ and $d$, the average repair cost $r_D(\tau)$ of distributed repair depends only on the repair bandwidth $\gamma$. As $\gamma_{MBR} \leq \gamma_{MSR}$, MBR codes are selected to minimize $r_D(\tau)$. Substituting \eqref{msr} and \eqref{mbr} in $r_C(\tau)$ and $r_D(\tau)$, respectively, we obtain
	\begin{equation}
	\frac{\mathcal{M}(2d)(n-\tau)\lambda\mu}{k(2d-k+1)(\mu H_{n,\tau} + \lambda)}  <  \frac{\mathcal{M} (k+n-\tau-1)\lambda\mu}{k(\mu H_{n,\tau} + \lambda)},
	\end{equation}
	which implies
	\begin{equation}
	k+n-\tau-1  <  2d.
	\label{comp_prop}
	\end{equation}

	Inequality \eqref{comp_prop} determines the minimum number of surviving nodes for which MBR distributed repair emerges as the most cost-efficient strategy. The left hand side (LHS) of \eqref{comp_prop} is a decreasing function of $\tau$. Maximizing the LHS yields the relationship between $n,k,$ and $d$ for which distributed MBR {\it always} outperforms centralized MSR. This occurs when $\tau=d$. Substituting $\tau=d$ results in $d > \frac{n+k-1}{3}$. If we reverse the direction of the inequality in \eqref{rtau_dc}, we obtain
	\begin{equation}
	2d  < n+k-\tau-1.
	\label{reverse}
	\end{equation}
	Minimizing the RHS of \eqref{reverse} yields the relationship between $n,k,$ and $d$ for which centralized MSR {\it always} outperforms distributed MBR. This occurs when $\tau=n-1$. Substituting $\tau=n-1$ results in $d<\frac{k}{2}$. However, by the definition of regenerating codes, we have $d \geq k.$ Therefore, there is no condition for which centralized MSR repair always outperforms distributed MBR repair.
\end{proof}


\section{Proof of Proposition \ref{thm:prop5}}
\label{sec:Prop5}
\begin{proof}
	To determine the optimal repair strategy we compare $r_C(\tau)$ with $r_D(\tau)$ for $k \leq \tau^{*} < d$
	\begin{center}
		$r_C(\tau) < r_D(\tau)$.
	\end{center}
	Substituting $r(\tau)$ for distributed repair and centralized repair from \eqref{rtau_d} and \eqref{rtau_c}, respectively, we obtain:
	\begin{equation}
	(\alpha(k+n-\tau-1))\frac{\lambda\mu}{H_{n,\tau} + \lambda}< (\alpha k(d-\tau)+\gamma(n-d))\frac{\lambda\mu}{H_{n,\tau} + \lambda}. 
	\label{optrep_c}
	\end{equation}
	For MSR codes, $\alpha_{MSR} \leq \gamma_{MSR}$ and for MBR codes, $\alpha_{MBR}=\gamma_{MBR}$. Thus, for each case, we have $\alpha\leq \gamma$. By choosing the lowest $\gamma$, we consider when the parameters satisfy
	\begin{eqnarray}
	\alpha(k+n-\tau-1) &<& \alpha k(d-\tau)+\alpha(n-d)\nonumber\\
	\Rightarrow	 k-1 &< & k(d-\tau) - (d- \tau)  \nonumber \\
	\Rightarrow	k-1 &<&(k-1)(d-\tau)  \nonumber\\
	\Rightarrow	\tau & < & d. \label{eq:ineq}
	\end{eqnarray}
	As $k \leq \tau^* <d$, inequality \eqref{eq:ineq} is always true and hence, centralized repair outperforms distributed repair. As explained in Proposition \ref{thm:prop4}, for centralized repair, MSR codes minimize the average repair cost rate per unit of time as compared to MBR codes. Thus, centralized repair using MSR codes yields the optimal repair strategy.
\end{proof}

	\section{Proof of Proposition \ref{thm:propMTTDL}}
	\label{App:MTTDL}
	Starting from state $n$, expected time to reach state $\tau$ is $\sum_{i=\tau+1}^{n}=\frac{1}{i\lambda}=\frac{H_{n,\tau}}{\lambda}$. Once the system is at state $\tau$, with probability $p=\frac{\tau\lambda}{\tau\lambda+\mu}$, the system will transition to state $\tau-1$. If this occurs, the recovery will not be possible since the number of fragments are below the repair threshold. However, DSS can still serve the users if $\tau-1 \geq k$. Until we reach state $k-1$, the data is not lost. The expected time to reach state $k-1$ from state $\tau$ is $\sum_{i=k}^{\tau} \frac{1}{i\lambda}=\frac{H_{\tau,k-1}}{\lambda}$. On the other hand, with probability $1-p$, recovery will be initiated and the system will be back to state $n$, which takes $\frac{1}{\mu}$ time. At that point, we will go thorough the same process again. Accordingly, we have
	\begin{equation}
	TDL = \begin{cases}
		\frac{H_{n,\tau}}{\lambda} + \frac{H_{\tau,k-1}}{\lambda}, & \textrm{ w.p } p \\
		\frac{2H_{n,\tau}}{\lambda} + \frac{1}{\mu} + \frac{H_{\tau,k-1}}{\lambda} & \textrm{ w.p. } (1-p)p\\
		\frac{3H_{n,\tau}}{\lambda} + \frac{2}{\mu} +  \frac{H_{\tau,k-1}}{\lambda} & \textrm{ w.p. } (1-p)^2p \\
		\dots & \dots \\
		\frac{iH_{n,\tau}}{\lambda} + \frac{i-1}{\mu} +  \frac{H_{\tau,k-1}}{\lambda} & \textrm{ w.p. } (1-p)^{i-1}p
	\end{cases}
\end{equation}		
from which the expected time to data loss, $E[\textrm{TDL}]$, can be calculated.	
	 
\end{appendices}

\end{document}